\documentclass[journal]{IEEEtran}
\IEEEoverridecommandlockouts
\usepackage{cite}
\usepackage[utf8]{inputenc}
\usepackage{amsmath,amssymb,amsfonts}
\usepackage{algorithm}
\usepackage{algpseudocode}
\usepackage{graphicx}
\usepackage{caption}
\usepackage{subcaption}
\usepackage[compact]{titlesec}  
\titlespacing{\section}{2pt}{2pt}{2pt}
\usepackage{textcomp}
\usepackage{xcolor}
\usepackage{comment}
\usepackage{hyperref}
\def\BibTeX{{\rm B\kern-.05em{\sc i\kern-.025em b}\kern-.08em
    T\kern-.1667em\lower.7ex\hbox{E}\kern-.125emX}}
\usepackage{bm}
\usepackage{enumitem}
\usepackage[normalem]{ulem}
\usepackage{cases}
\usepackage{color}
\usepackage{geometry}
\usepackage{array}
\usepackage{amsthm}

\newtheorem{lemma}{Lemma}

\newtheorem{remark}{Remark}

\newtheorem{assumption}{Assumption}
\usepackage{ upgreek }
\usepackage{soul}
\usepackage{nomencl} 
\makenomenclature   

\newcolumntype{C}[1]{>{\centering\arraybackslash}p{#1}}

\usepackage[utf8]{inputenc}
\usepackage{algorithm}
\usepackage{algpseudocode}

\begin{document}

\title{Communication-Efficient Quantum Federated Learning over Large-Scale Wireless Networks  

\thanks{}
}

\author{\IEEEauthorblockN{Shaba Shaon, Christopher G. Brinton, \textit{Senior member, IEEE}, and Dinh C. Nguyen }

\thanks{Shaba Shaon and Dinh C. Nguyen are with ECE Department,  The University of Alabama in Huntsville, Huntsville, AL 35899, USA, emails: ss0670@uah.edu, dinh.nguyen@uah.edu.}
\thanks{Christopher G. Brinton is with the Elmore Family School of Electrical and Computer Engineering, Purdue University, West Lafayette, IN, 47907, USA, e-mail: cgb@purdue.edu.}
}




\maketitle

\begin{abstract}
Quantum federated learning (QFL) combines the robust data processing of quantum computing with the privacy-preserving features of federated learning (FL). However, in large-scale wireless networks, optimizing sum-rate is crucial for unlocking the true potential of QFL, facilitating effective model sharing and aggregation as devices compete for limited bandwidth amid dynamic channel conditions and fluctuating power resources. This paper studies a novel sum-rate maximization problem within a muti-channel QFL framework, specifically designed for non-orthogonal multiple access (NOMA)-based large-scale wireless networks. We develop a sum-rate maximization problem by jointly considering quantum device's channel selection and transmit power. Our formulated problem is a non-convex, mixed-integer nonlinear programming (MINLP) challenge that remains non-deterministic polynomial time (NP)-hard even with specified channel selection parameters. The complexity of the problem motivates us to create an effective iterative optimization approach that utilizes the sophisticated quantum approximate optimization algorithm (QAOA) to derive high-quality approximate solutions. Additionally, our study presents the first theoretical exploration of QFL convergence properties under full device participation, rigorously analyzing real-world scenarios with nonconvex loss functions, diverse data distributions, and the effects of quantum shot noise. Extensive simulation results indicate that our multi-channel NOMA-based QFL framework enhances model training and convergence behavior, surpassing conventional algorithms in terms of accuracy and loss. Moreover, our quantum-centric joint optimization approach achieves more than a $100\%$ increase in sum-rate while ensuring rapid convergence, significantly outperforming the state-of-the-arts. 
\end{abstract}

\section{Introduction}

Quantum federated learning (QFL) is an innovative approach that integrates quantum computing with federated learning (FL) to enable decentralized model training across quantum-powered edge devices such as quantum processors or intelligent quantum sensors \cite{chehimi2023foundations, qu2023dtqfl}. In this framework, quantum-powered edge devices such as quantum processors collaborate with a quantum cloud server at a central base station (BS)  to share the updates of quantum machine learning (QML) models, e.g., variational quantum algorithms (VQAs) \cite{qu2023dtqfl}, and collectively improve the global model while maintaining privacy and data security. The use of quantum computing in quantum edge devices offers significant advantages in computational efficiency and power, leveraging quantum circuits to accelerate learning and enhance model performance \cite{tang2021cutqc}. In particular, in this noisy intermediate-scale quantum (NISQ) era, QML is typically trained using a hybrid model that combines quantum circuits, such as parameterized quantum circuits (PQC), with classical optimization techniques such as  gradient descent. This approach leverages quantum processing for feature encoding and classical optimization to adjust the circuit parameters. We note that the NISQ-inherent QML models, i.e., VQAs' parameters, remain in \textit{classical form} instead of quantum states after model training \cite{chehimi2023foundations}, making it compatible for sharing and transmitting via \textit{wireless networks}  for further aggregation.

In future wireless networks like 6G, communication will become a significant bottleneck in QFL with large-scale quantum edge devices due to massive data exchange and limited power resources \cite{quy2024federated}. Therefore, sum-rate maximization, such as optimizing model transmission rates, will be crucial to ensure efficient and scalable learning across quantum-enabled edge networks. \textcolor{black}{In large-scale QFL systems, wireless communication becomes a critical bottleneck because PQC parameters generated by quantum edge devices must be frequently transmitted to the server \cite{chehimi2023foundations}. These parameters scale with circuit depth and qubit count \cite{tang2021cutqc}, creating substantial uplink demand that is constrained by limited bandwidth \cite{chehimi2023foundations}, channel fading \cite{wu2018spectral}, and the low transmit power of quantum edge hardware \cite{fellous2023optimizing}. Specifically, in multi-channel NOMA-based systems, the limited number of available channels relative to the large number of participating devices often results in intensified competition for wireless resources, thereby aggravating the communication bottlenecks \cite{liu2019performance}. Furthermore, the noise sensitivity of quantum circuits introduces fluctuations in parameter quality \cite{endo2021hybrid}, increasing the reliance on timely and reliable wireless links for stable global updates.} \textcolor{black}{Although the transmitted information in QFL is classical \cite{chehimi2023foundations}, the underlying quantum hardware behavior leads to communication dynamics that differ from classical FL \cite{quy2024federated}. Measurement variability and quantum noise cause PQC parameters to change more abruptly \cite{jose2022error}, requiring more precise and consistent update transmissions to maintain convergence. Combined with the strict power and hardware limitations of emerging quantum processors \cite{fellous2023optimizing}, these effects introduce unique communication pressures that highlight the importance of wireless resource optimization in QFL systems.} Furthermore, as the number of quantum edge devices grows, it becomes increasingly complex to train and optimize models in a large-scale system. Traditional optimization techniques, such as convex optimization, often struggle to handle the increasing complexity of network orchestration and resource allocation in such large-scale environments \cite{rana2025allocation}. These limitations pose challenges in sum-rate maximization over wireless links, ultimately degrading communication efficiency in FL. Fortunately, the quantum approximate optimization algorithm (QAOA) has recently emerged as a promising solution to accelerate optimization in wireless QFL \cite{gulbahar2025majority, xiang2025choco}. QAOA, based on quantum-inspired techniques, is particularly well-suited for optimization problems that arise in complex systems like large-scale  wireless   networks. QAOA can be used to efficiently search for optimal resource allocation policies in QFL systems, which is critical in managing communication bandwidth, energy usage, and computation resources in large-scale QFL networks for future intelligent  wireless  networks.

\subsection{Related Works}
\textit{1) QFL:} Recently, several studies have explored a range of QFL schemes \cite{park2025dynamic, park2025entanglement, gurung2025performance, qu2025daqfl, wang2024quantum, qu2024qfsm, abou2024privacy, ren2024enhancing, rofougaran2024federated}, revealing diverse approaches and methodologies within the field. In \cite{park2025dynamic}, the authors introduced a dynamic QFL framework that utilizes quantum neural networks to enhance the efficiency and robustness of collaborative UAV surveillance under challenging communication conditions. The work in \cite{park2025entanglement} developed an entangled slimmable QFL system with a depth-controllable quantum neural network architecture, optimizing it for diverse internet of things (IoT) environments by utilizing superposition-coded parameter communication and interference mitigation techniques. The studies in \cite{gurung2025performance, qu2025daqfl, rofougaran2024federated} worked on enhancing data privacy, model accuracy, and robustness in heterogeneous environments through innovative QFL approaches, addressing challenges such as non-IID data and local model drift. The work in \cite{wang2024quantum} developed a QFL model by implementing a decentralized ring topology and introducing the use of quantum weights for entirely quantum-based training, addressing security and data privacy challenges in deep learning. The authors \cite{qu2024qfsm} presented a QFL algorithm that integrates QML with 5G internet of vehicles, using a quantum minimal gated unit for training, to enhance computational efficiency, privacy protection, and speech emotion recognition accuracy. The works in \cite{abou2024privacy, ren2024enhancing} utilized QFL approaches to enhance network security and efficiency, with the former improving intrusion detection in consumer networks and the latter enhancing stability in smart grids.

\textit{2) Optimization and Sum-Rate Maximization  for Communication-Efficient FL/QFL: }  
Several studies have explored optimization strategies for communication-efficient  FL, including successive convex approximation (SCA) and Lyapunov optimization. Gradient-based techniques, such as first-order and second-order methods, leverage function derivatives to solve both convex and non-convex problems. For example, the works in \cite{yang2020energy}, \cite{tran2019federated} exploited SCA to optimize energy-efficient resource allocation for FL in wireless networks. While enhancements like accelerated first-order, proximal, and stochastic gradient methods have improved efficiency, these approaches remain insufficient in large-scale FL systems with high number of optimization variables  due to large number of devices. Recent study in \cite{hazarika2024quantum} considered a quantum-based sequential training approach to optimize  communication costs in QFL. \textit{However, an optimization using QAOA has not been studied and the benefits of QAOA in a large-scale QFL network have not investigated.} 

In FL/QFL optimization, sum-rate maximization is crucial as it determines the transmission rate for model sharing, directly impacting communication efficiency and aggregation speed. For instance, the authors in \cite{hu2024ofdma} studied the sum-rate maximization problem in classical FL using a Lagrange-dual based method. Communication efficiency for FL was also considered in \cite{salgia2024communication, sun2024communication} by using classical techniques such as heuristic optimizations. \textit{However, the system model is small-scale, and quantum-inspired has not been investigated for sum-rate maximization.} The challenge in FL optimizations intensifies in large-scale wireless networks, where devices compete for limited bandwidth while ensuring timely and reliable updates.  Dynamic channel conditions over communication links and fluctuating power resources at FL devices further complicate the optimization process. Effective sum-rate maximization is essential for optimizing communication efficiency, e.g., model sharing latency, and QFL convergence rates.

\textit{3) QAOA for Wireless Optimization: }
 QAOA has been utilized mostly to formulate and solve mixed-integer binary programming problems \cite{kasi2024quantum}. The authors in \cite{manasa2024optimizing} incorporated QAOA to optimize energy distribution in smart grids with a view to improving load balancing, energy efficiency, and computational speed. While the work in \cite{van2024optimal} presented QAOA-based service placement optimization framework in 6G mobile edge computing to minimize service costs and delay, the researchers in \cite{colella2024quantum} utilized this quantum-based approach to optimize reconfigurable intelligent surfaces in multipath wireless environments. \textit{Despite these research efforts, the application of QAOA has not been investigated in the context of QFL.} 

\subsection{Key Contributions}
Motivated by these limitations in the literature, this paper for the first time presents an optimized QFL framework over large-scale wireless networks with theoretical guarantees. Our contributions are highlighted as follows:
\begin{itemize}
    \item We propose a novel multi-channel QFL framework over NOMA-based large-scale wireless networks \ref{section:systemmodel}. We formulate the sum-rate maximization problem for this network, addressing the critical requirements for resource efficiency and scalability in such networks. In doing so, we jointly optimize device's channel selection and transmit power \ref{section:problemformulation}.
    \item The development of the sum-rate maximization formulation presents a complex computational challenge that requires sophisticated problem-solving approaches. Upon analysis, we determined that the problem is non-convex and qualifies as a mixed-integer nonlinear programming (MINLP) challenge. Furthermore, it remains NP-hard even when channel selection parameters are specified. Therefore, we develop and implement a quantum-centric optimization approach utilizing QAOA and block \textcolor{black}{coordinate} descent (BCD) technique to effectively tackle this problem. We divide the original problem into two sub-problems and solve them iteratively using BCD to find optimal solution for the original problem. We transform each of the classical sum-rate maximization sub-problems into a quadratic unconstrained binary optimization (QUBO) format, which we further reformulate into a Hamiltonian expression to facilitate compatibility with quantum processors \ref{section:proposedsolution}.
    \item We conduct a rigorous convergence analysis of QFL under full device participation, considering practical scenarios such as nonconvex loss functions and varied data distributions (both IID and non-IID) among quantum devices \ref{section:convergenceanalysis}. Additionally, we examine the impact of quantum shot noise on the convergence of QFL. To the best of our knowledge, this is the first theoretical exploration of the convergence properties of QFL.
    \item We conduct extensive numerical experiments to evaluate model training, convergence behavior, and sum-rate performance of the proposed multi-channel NOMA-based QFL framework over large-scale wireless networks. Simulation results reveal that our proposed QFL framework outperforms state-of-the-art algorithms in terms of model loss and accuracy convergence. Moreover, our quantum-centric optimization approach achieves more than a $100\%$ increase in sum-rate compared to state-of-the-art schemes \ref{simulationresultsandevaluations}.
\end{itemize}

\section{System Model} \label{section:systemmodel}
This section presents the QFL system model for our considered NOMA-based multi-channel large-scale wireless network. We start by defining several fundamental terms that are frequently used throughout our analysis.

\textit{\textbf{Definition 1.}} \textit{NISQ processor:} A NISQ processor is a quantum computing device with a limited number of qubits that is subject to noise and errors. 

\textit{\textbf{Definition 2.}} \textit{PQC:} A PQC is a quantum circuit with gate operations that are defined by tunable parameters. They are generally optimized through classical feedback loops and incorporated to solve problems in quantum computing, for instance variational quantum algorithms.

\textit{\textbf{Definition 3.}} \textit{VQA:} A VQA is a hybrid approach that leverages both quantum and classical computation. This algorithm employs a PQC to generate quantum states and uses a classical optimizer to iteratively update the circuit's parameters to minimize a cost function. Fig.~\ref{Fig:Overview} illustrates the integration of VQA with quantum computation and SGD for local model training at the quantum device.

\textit{\textbf{Definition 5.}} \textit{Shot noise:} Shot noise refers to statistical fluctuations that occur due to the limited number of measurement shots used to estimate the outcomes of a quantum computation.

\textit{\textbf{Definition 7.}} \textit{Decoherence:} Decoherence is the phenomenon through which a quantum system interacts with its surrounding environment, resulting in the degradation of quantum properties like superposition and entanglement, and consequently, the loss of quantum information.

\textit{\textbf{Definition 8.}} \textit{Gate infidelity:} Gate infidelity quantifies the deviation of a quantum gate's real-world operation from its ideal, error-free behavior, serving as a metric for the imperfections or errors in the gate's implementation.

\subsection{QFL Learning Model}
In this work, we consider a QFL system in which a centralized server orchestrates a network of $N$ distributed quantum-enabled devices. Assume that each quantum device employs a VQA algorithm such as a quantum neural network (QNN) to build a QML model by executing a PQC with the assistance of a classical loss function optimization method, e.g., stochastic gradient decent (SGD). In each global iteration, devices conduct local training utilizing their private datasets through quantum SGD across $T_{n}$ local iterations. Upon completion of local training, devices send their classical model parameters back to the central server for aggregation. The server then employs federated averaging to update the global model. \textcolor{black}{We define the complete set of devices as $\mathcal{N} = \{1, 2, \dots, N\}$, whereas $\mathcal{K} = \{1, 2, \dots, K\}$ denotes the sets of global iterations. For each device $n \in \mathcal{N}$, the set of local iterations is represented by $\mathcal{T}_{n} = \{1, 2, \dots, T_{n}\}$. Accordingly, we define the set of all local iteration indices across devices as $\mathbb{T} = \{{\mathcal{T}_{1}, \mathcal{T}_{2}, \dots, \mathcal{T}_{N}}\}$.} We assume that the PQC applies a unitary transformation $U(\boldsymbol{w_{n}})$ on $q$ qubits. The PQC is parameterized by a vector $\boldsymbol{w_{n}} = (\boldsymbol{\theta}_{n}^{1}, \boldsymbol{\theta}_{n}^{2}, ..., \boldsymbol{\theta}_{n}^{P})$, where $P$ denotes the number of parameters associated with the PQC, and each parameter $\boldsymbol{\theta}_{n}^{p} \in \mathbb{R}$ for $p = 1, 2, ..., P$. A conventional architecture for a PQC, frequently called an ansatz, specifies the unitary transformations as:
\begin{align}
    U(\boldsymbol{\theta}_{n}) = \prod_{p=1}^{P} U_{p} (\boldsymbol{\theta}_{n}^{p}) V_{p}, \label{eqn1}
\end{align}
where $\boldsymbol{\theta}_{n}^{p}$ represents the operation of the corresponding unitary $U_{p}(\boldsymbol{\theta}_{n}^{p})$. The parameterized gate is characterized as 
\begin{align}
    U_{p} (\boldsymbol{\theta}_{n}^{p}) = e^{(-i\frac{\boldsymbol{\theta}_{n}^{p}}{2}G_{p})}, \label{eqn2}
\end{align}
where $G_p \in \{I, X, Y, Z\}^{\otimes q}$ is the Pauli string generator. We mention that the unitary $V_{p}$ in \eqref{eqn1} is fixed and independent of the parameters $\boldsymbol{\theta}$. The PQC is applied to $q$ qubits, all of which start in the ground state $\lvert 0 \rangle$. This results in the following parameterized quantum state corresponding to device $n$
\begin{align}
    {\lvert \Psi_{n} \rangle} = U(\boldsymbol{\theta}_{n}) \lvert 0 \rangle.
\end{align}
The pure-state density matrix associated with this quantum state is written as
\begin{align}
    \Psi_{n}(\boldsymbol{\theta}_{n}) = \lvert \Psi_{n}(\boldsymbol{\theta}_{n}) \rangle \langle \Psi_{n}(\boldsymbol{\theta}_{n}) \rvert.
\end{align}
The goal of a QNN is to employ the PQC to create a quantum state that optimally executes a specified ML task. More specifically, the QNN seeks to determine the parameter vector $\boldsymbol{\theta}_{n} \in \mathbb{R}^{P}$ that solves the following optimization problem
\begin{align}
    \min_{\boldsymbol{\theta}_{n} \in \mathbb{R}^{P}} \left\{ f_{n}(\boldsymbol{\theta}_{n}) := \mathcal{C}(\langle Z \rangle_{\lvert \Psi_{n}(\boldsymbol{\theta}_{n}) \rangle} , y)\right\}, \label{eqn5}
\end{align}
\textcolor{black}{where $f_{n}(\boldsymbol{\theta}_{n})$  denotes a generic task-dependent loss function (e.g., cross-entropy or mean-squared error) for device $n$. The loss $\mathcal{C}(\cdot,y)$ compares the PQC's predicted output, given by the expectation value $\langle Z \rangle_{\lvert \Psi_{n}(\boldsymbol{\theta}_{n}) \rangle}$, with the ground-truth label $y$ of the training sample.} \textcolor{black}{The PQC produces the output quantum state $\lvert \Psi_{n}(\boldsymbol{\theta}_{n}) \rangle$, and the model's prediction is obtained by measuring the expectation value of the observable $Z$ acting on this state, given by
\begin{align}
    \langle Z \rangle_{\lvert \Psi_{n}(\boldsymbol{\theta}_{n}) \rangle} = \langle 0 \rvert U^\dagger(\boldsymbol{\theta}_{n}) Z U(\boldsymbol{\theta}_{n}) \lvert 0 \rangle = \operatorname{Tr}(Z \Psi_{n}(\boldsymbol{\theta}_{n})),
\end{align}}
where $\langle 0 \rvert$ and $U^\dagger(\boldsymbol{\theta}_{n})$ are the Hermitian conjugates of $\lvert 0 \rangle$ and $U(\boldsymbol{\theta}_{n})$, respectively. Moreover, $\operatorname{Tr}(Z \Psi_{n}(\boldsymbol{\theta}_{n}))$ is the trace of the product of the observable $Z$ and the density matrix $\Psi_{n}(\boldsymbol{\theta}_{n})$. The minimum value of the loss function $f_{n}(\boldsymbol{\theta}_{n})$ is expressed as
\begin{align}
    {f_{n}}^* = \min_{\boldsymbol{\theta}_{n} \in \mathbb{R}^{P}} f_{n}(\boldsymbol{\theta}_{n}).
\end{align}
The Hermitian matrix $Z$, which represents the observable, is written in terms of its eigendecompositions as
\begin{align}
    Z = \sum_{y=1}^{N_z} h_y \Pi_y,
\end{align}
where $\{ h_y \}_{y=1}^{N_z}$ are the $N_{z} \leq 2^{n}$ distinct eigenvalues of $Z$, and $\{ \Pi_y \}_{y=1}^{N_{z}}$ are the projection operators onto the corresponding eigenspaces. \textcolor{black}{Using the eigendecomposition, the expectation value of the observable for device $n$ is written as
\begin{align}
    \langle Z \rangle_{\lvert \Psi_{n}(\boldsymbol{\theta}_{n}) \rangle} = \sum_{y=1}^{N_z} h_y \operatorname{Tr}(\Pi_y \Psi_{n}(\boldsymbol{\theta}_{n})).
\end{align}
Substituting this expression into the general loss definition in \eqref{eqn5}, the loss function becomes
\begin{align}
    f_{n}(\boldsymbol{\theta}_{n}) = \mathcal{C}(\sum_{y=1}^{N_z} h_y \operatorname{Tr}(\Pi_y \Psi_{n}(\boldsymbol{\theta}_{n})),y).
\end{align}}
In case of QNN, the optimization of the parameter vector $\boldsymbol{\theta}_{n}$ for the PQC is usually performed iteratively using stochastic gradient descent (SGD). Thus, the optimization problem in \eqref{eqn5} can be solved by implementing a hybrid quantum-classical optimization method. Specifically, $H$ measurements of the observable $Z$ are taken for the output state $\lvert \Psi_{n}(\boldsymbol{\theta}_{n}) \rangle$ of the PQC, resulting in a set of samples $\{ Z_{h} \}_{h=1}^{H}$, where $Z_{h} \in \{Z_1, Z_2, ..., Z_{H}\}$. \textcolor{black}{From these samples, an estimate of the expectation value is computed as
\begin{align}
    \hat{\langle Z \rangle}_{\lvert \Psi_{n}(\boldsymbol{\theta}_{n}) \rangle} = \frac{1}{H} \sum_{h=1}^{H} Z_h. \label{eqn10new}
\end{align}
We mention that the $\hat{(\cdot)}$ notation used in this paper denotes empirical estimates. This estimated expectation in \eqref{eqn10new} is then passed into the classical loss function $\mathcal{C}(\cdot,y)$ to evaluate the loss and compute the stochastic gradients used in the SGD update.} In this work, we concentrate on the standard implementation of QNN using SGD optimization technique. We denote the initial parameter vector for device $n$ as $\boldsymbol{\theta}_{n}^{0} \in \mathbb{R}^M$, and the update rule for the parameter vector $\boldsymbol{\theta}_{n}$ is expressed as 
\begin{align}
    \boldsymbol{\theta}_{n}^{t+1} = \boldsymbol{\theta}_{n}^{t} - \eta_{t} \hat{g}_{n}^{t},
\end{align}
\textcolor{black}{where $t = 1, 2, ..., T_{n}$, with $T_{n}$ being the total number of local SGD iterations performed by device $n$ during a global round}, $\eta_{t} > 0$ denotes the learning rate at iteration $t$, and $\hat{g}_{n}^{t}$ represents the stochastic gradient estimate of the loss function $f_{n}(\boldsymbol{\theta}_{n})$ at iteration $t$. In this context, $\hat{g}_{n}^{t}$ is expressed as
\begin{align}
    \hat{g}_{n}^{t} = \nabla f_{n}(\boldsymbol{\theta}_{n}) \bigg|_{\boldsymbol{\theta}_{n} = \boldsymbol{\theta}_{n}^t} = \begin{bmatrix} \frac{\partial f_{n}(\boldsymbol{\theta}_{n})}{\partial \boldsymbol{\theta}_{n}^{1}} \\ \vdots \\ \frac{\partial f_{n}(\boldsymbol{\theta}_{n})}{\partial \boldsymbol{\theta}_{n}^{P}} \end{bmatrix}_{\boldsymbol{\theta}_{n} = \boldsymbol{\theta}_{n}^t}.
\end{align}

\begin{figure}[h]
    \centering
    \includegraphics[width=0.8\linewidth]{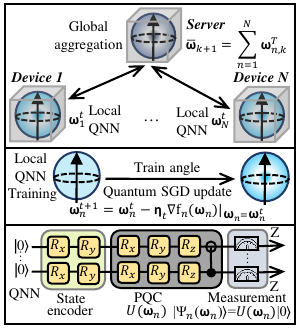}
    \caption{\footnotesize Proposed QFL framework where distributed quantum devices collaborate with a central server to train a shared ML model. Each device encodes practical data into quantum form using a state encoder, processes it through a PQC with trainable angle parameters, and then employs measurement outcomes to update local model parameters using the standard SGD method, subsequently sending the updated local models to the server for aggregation.}
    \label{Fig:Overview}
\end{figure}

\subsection{Network Architecture}
In this work, we explore a multi-channel QFL framework where quantum devices communicate with BS across multiple dedicated channels. Specifically, we consider $N$ quantum devices share $C$ channels, where the number of devices typically exceeds the number of channels, i.e. $N > C$. Considering more devices than channels optimizes resource utilization by dynamically allocating limited channels to a larger pool of devices, enhancing scalability and efficiency. Moreover, this configuration mirrors real-world constraints where available communication resources are often outnumbered by participating devices. Let us denote the set of $N$ devices and the set of $C$ channels as $\mathcal{N} = \{1, 2, \dots, N\}$ and $\mathcal{C} = \{1, 2, \dots, C\}$, respectively. \textcolor{black}{Throughout this work, we assume full device participation in every global round, enabled by a multi‑channel NOMA uplink architecture that allows all $N$ quantum devices to transmit concurrently via channel reuse, with the resulting co-channel interference captured explicitly through the signal-to-interference-plus-noise ratio (SINR) expressions.} In the subsequent sub-sections, we detail the channel and signal modeling approaches adopted in our study. 

\subsection{Channel Model}
\textcolor{black}{All communication links in our QFL framework are purely classical wireless channels, modeled using path loss, shadowing, Rayleigh fading, and co-channel interference. The role of quantum techniques is restricted to (i) local model training at the quantum edge devices using PQC-based QNN models, and (ii) the QAOA-based optimization procedure executed at the BS. This design is fully appropriate for our setting because the quantum devices only need to exchange classical PQC parameters with the BS \cite{quy2024federated}, meaning no quantum states, entanglement distribution, or quantum communication links are required.} In our system model, the communication link between a quantum device and the BS is susceptible to co-channel interference. This interference arises when another device simultaneously uses the same communication channel to transmit model parameters to the BS. Such overlap can degrade the quality of the connection, impacting the transmission efficiency and reliability of the legitimate link. In our framework, device $n$ maintains a dual role: being part of a legitimate link to the BS and becoming a source of interference, adversely affecting the signals transmitted by any other device $m$ sharing the same channel $c$. It is assumed that all the channels in this system exhibit both large-scale and small-scale impairment components. Path loss and shadowing are the large-scale components, which remain constant over long period of time, while Rayleigh fading represents the small-scale component, characterized by time-block fading. In this block fading model, the channel's characteristics are assumed to remain constant over a certain block of time. To describe the time-correlation of Rayleigh fading, we employ Jake's model \cite{kim2010does}. We represent the channel coefficients of the channel instances during the $b^{\text{th}}$ and $(b+1)^{\text{th}}$ time-block as $g^{(b)}$ and $g^{(b+1)}$, then we have
\begin{align}
    g^{(b+1)} = \epsilon g^{(b)} + \sqrt{1 - \epsilon^{2}} \delta, \label{eqn14'}
\end{align}
where $\epsilon (0 \leq \epsilon < 1)$ represents the correlation factor accounting for the correlations of the channel realizations in two successive blocks, and both $\delta$ and $g^{(0)}$ are complex Gaussian random variables following the distribution $\mathcal{C}\mathcal{N}(0,1)$. Besides, let us denote the path loss and the shadowing components of the interference link generated from device $n$ to the BS as $\chi_{n,m,c}$ and $\beta_{n,m,c}$ on channel $c$, respectively. Moreover, the corresponding Rayleigh fading component caused by device $n$ to the BS on channel $c$ at block $b$ is represented by $g_{n,m,c}^{(b)}$. We note that these components degrade the quality of the signal transmitted by device $m$ directing to the BS, where $m \in \mathcal{N}, m \neq n$, and $m$ uses the same channel $c$ as $n$ to transmit model parameters to the BS. Therefore, the channel gain corresponding to the interference link cause by device $n$ between device $n$ at block $b$ is written as
\begin{align}
    h_{n,m,c}^{(b)} = \chi_{n,m,c} \beta_{n,m,c} |g_{n,m,c}^{(b)}|. \label{eqn15'}
\end{align}

\subsection{Frame Structure and Signal Model}
In our framework, given the asynchronous nature of local model training by quantum devices, i.e., any quantum device may complete its local model training at any time during a global round followed by transmission of local model parameters to the server, each global round $k$ is assumed to be segmented into $b$ time blocks to facilitate channel modeling. In the frame structure employed by any link starting from device $n$ and ending at the BS in the network under consideration, each frame consists of a preprocessing phase and a transmission phase. During the preprocessing phase, each device gathers necessary information to select a channel and a transmit power for subsequent transmission. Following this decision, each device sends its updated local model parameters to the BS during the transmission phase.  However, it is crucial to mention that devices can still use the same channel to transmit model parameters at block $b$, potentially causing co-channel interference. This occurs because orthogonal pilot signals primarily ensure accurate channel state information approximation, but do not eliminate the possibility of signal overlap during simultaneous transmissions by two or more devices using with the same channel.

We use a binary variable $\zeta_{n,c}^{(b)}$ to indicate whether device $n$ transmits on channel $c$ during block $b$; $\zeta_{n,c}^{(b)} = 1$ if it does, and $\zeta_{n,c}^{(b)} = 0$ otherwise. Each device is allowed to select at most one channel during a block so that the sum of their channel selections does not exceed one. Hence, we have
\begin{align}
    \sum_{c=1}^{C} \zeta_{n,c}^{(b)} \leq 1, \forall n \in \mathcal{N}. \label{eqn3}
\end{align}
We mention that the constraint in \eqref{eqn3} is widely adopted in state-of-the-art literature \cite{wang2016optimal}. Let us denote the transmit power of device $n$ during block $b$ as $p_{n}^{(b)}$, then the corresponding received signal-to-interference-plus-noise ratio (SINR) on channel $c$ is formulated as
\begin{align}
    \gamma_{n,c}^{(b)} = \frac{\zeta_{n,c}^{(b)} h_{n,c}^{(b)} p_{n}^{(b)}}{\sigma^{2} + \sum_{m \in \mathcal{N},m \neq n} \zeta_{m,c}^{(b)} h_{m,n,c}^{b} p_{m}^{(b)}}, \label{eqn17'}
\end{align}
where $h_{n,c}^{(b)}$ represents the channel gain of the legitimate link from device $n$ to the BS on channel $c$ during block $b$, and $\sigma^{2}$ is the additive white Gaussian noise (AWGN) power. Thus, if we denote bandwidth as $B$, the achievable transmission rate by device $n$ on channel $c$ during block $b$ is calculated as
\begin{align}
    r_{n,c}^{(b)} = B \log_{2}(1 + \gamma_{n,c}^{(b)}). \label{eqn18'}
\end{align}
\textcolor{black}{Throughout this work, the term ``wireless'' refers specifically to a classical 6G-oriented uplink model in which quantum devices communicate with the BS over multiple dedicated channels. The system follows a multi-channel NOMA structure, meaning that several devices may transmit over the same channel simultaneously \cite{liu2019performance}, resulting in co-channel interference that determines the achievable SINR and data rate. As mentioned before, all exchanged information consists of classical PQC model parameters \cite{quy2024federated}, while the quantum aspects of the framework arise solely from local PQC-based QNN training and the QAOA-based optimization at the BS.}

\section{Analysis of Convergence Behavior of QFL} \label{section:convergenceanalysis}
In this section, we present a comprehensive convergence analysis of our proposed QFL framework considering full-device participation scenario under real-world assumptions such as non-convex loss functions, diverse data distributions, and the effects of quantum shot noise. \textit{Our findings reveal that the convergence rate is influenced by the total number of iterations, the number of total devices, the subset of devices participating in each global iteration, and the number of measurements performed on the parameterized quantum circuit.} 

In this subsection, we detail our QFL algorithms accommodating full device participation, detailing the core mechanisms for efficient local training and global aggregation.
Algorithm~\ref{algorithm1} showcases steps implemented in QFL with full device participation mode where the server initiates the process by broadcasting an initial global model to all participating devices. Each quantum-enabled device $n \in \mathcal{N}$ receives the global model and begins local training using quantum SGD on their private datasets, performing $T_{n}$ iterations per global round. Upon completing the local iterations, devices update their models based on the learned parameters. Following the updates, devices transmit their classical model parameters back to the server. The server aggregates these received parameters using federated averaging, thereby updating and refining the global model for the next round of distribution. This cycle repeats across $K$ global rounds when the model converges. Algorithm~\ref{algorithm1} operates under full device participation, as the server aggregates updated parameters from all devices at each global round.

\begin{algorithm}[ht!]
  \caption{QFL algorithm for full device participation}
  \label{algorithm1}
\begin{algorithmic}[1]
 \footnotesize
    \State \textbf{Input:} learning rate $\eta_{k}$, $\boldsymbol{\theta}_{0}$ as an initial model shared by all the devices, \textcolor{black}{$\mathbb{T}$} as the set of all local iteration indices across devices
            \State \textbf{for} $k = 1, 2, ..., K$ \textbf{do:}
            \State \indent \textcolor{black}{The BS invokes Algorithm \ref{algorithm3} to perform joint resource}
            \State \indent \textcolor{black}{allocation, yielding optimal channel assignments and}
            \State \indent \textcolor{black}{power levels: ${\zeta^{b}}^*$, ${x^{(b)}}^*$}
            \State \indent \textbf{Parallel for all the devices} $n \in \mathcal{N}$ \textbf{do:}
            \State \indent \indent \textbf{for} $t = 1, 2, ..., \textcolor{black}{T_{n}}$, \textcolor{black}{\textbf{where $\mathcal{T}_{n} \in \mathbb{T}$,}} \textbf{do:}
            \State \indent \indent \indent $\boldsymbol{\theta}_{n,k}^{t+1} = \boldsymbol{\theta}_{n,k}^{t} - \eta_{k} \Tilde{g}_{n,k}^{t}$
            \State \indent \indent Each device $n \in \mathcal{N}$ sends \textcolor{black}{$\boldsymbol{\theta}_{n,k}^{T_{n}}$} to the server
            \State \indent \indent \textbf{end parallel for}
            \State \indent \textbf{Server computes}
            \State \indent \indent $\bar{\boldsymbol{\theta}}_{k+1} = \frac{1}{N} \sum_{n \in \mathcal{N}} [\textcolor{black}{\boldsymbol{\theta}_{n,k}^{T_{n}}}]$
            \State \indent Server broadcasts $\bar{\boldsymbol{\theta}}_{k+1}$ to all the devices
            \State \indent \textbf{end for}
            \State \textbf{end for}
            \State \textbf{Output:} $\bar{\boldsymbol{\theta}}_{K} = \frac{1}{N} \sum_{n \in \mathcal{N}} \textcolor{black}{\boldsymbol{\theta}_{n,K-1}^{T_{n}}}$
\end{algorithmic} 
\end{algorithm}

\subsection{Notation and Definition}
We focus on the following optimization problem: $\min_{\boldsymbol{\theta}} f(\boldsymbol{\theta}) \triangleq \sum_{n=1}^{N} f_{n}(\boldsymbol{\theta})$, where $f(\boldsymbol{\theta})$ is the global objective function. In this federated framework, it is assumed that each device $n$ trains its local model on dataset $\mathcal{S}_{n}$ containing $S_{n}$ data points sampled in an independent and identically distributed (i.i.d.) fashion from the local distribution $\mathcal{D}_{n}$. We note that, since the local datasets are generated from different distributions, we carefully consider the heterogeneity of these distributions while analyzing the convergence of QFL. We define $g_{n} = \frac{1}{|\mathcal{S}_{n}|} \nabla f_{n}(\boldsymbol{\theta})\stackrel{\triangle}{=} \frac{1}{|\mathcal{S}_{n}|} \nabla f(\boldsymbol{\theta};\mathcal{S}_{n})$, where $f(\boldsymbol{\theta};\mathcal{S}_{n})$ represents the full gradient. Moreover, we denote the stochastic gradient as $\Tilde{g_{n}} \stackrel{\triangle}{=} \frac{1}{B} \nabla f(\boldsymbol{\theta};\xi_{n})$, where $\xi_{n} \subseteq \mathcal{S}_{n}$ is a uniformly sampled mini-batch with $|\xi_{n}| = B$. The corresponding quantities evaluated at device $n$'s local solution $\boldsymbol{\theta}_{n,k}^{t}$ during local iteration $t$ of the $k^{\text{th}}$ global round are denoted by $g_{n,k}^{t}$ for the full gradient and $\Tilde{g}_{n,k}^{t}$ for the stochastic gradient. We also define the following notations: $\boldsymbol{\theta}_{k}^{t} = [\boldsymbol{\theta}_{1, k}^{t}, \boldsymbol{\theta}_{2, k}^{t}, \dots, \boldsymbol{\theta}_{N, k}^{t}]$,
\begin{align}
    \xi_{k}^{t} = [\xi_{1, k}^{t}, \xi_{2, k}^{t}, \dots, \xi_{N, k}^{t}],
\end{align}
in order to represent the set of local solutions and sampled mini-batches associated with the devices during local iteration $t$ at $k^{\text{th}}$ global round, respectively. The following notations will be useful for the convergence analysis of the QFL framework: $\bar{\boldsymbol{\theta}}_{k}^{t} \stackrel{\triangle}{=} \frac{1}{N} \sum_{n \in \mathcal{N}} \boldsymbol{\theta}_{n,k}^{t}$,
$
    \Tilde{g}_{k}^{t} \stackrel{\triangle}{=} \frac{1}{N} \sum_{n \in \mathcal{N}} \Tilde{g}_{n,k}^{t},
$
$
    g_{k}^{t} \stackrel{\triangle}{=} \frac{1}{N} \sum_{n \in \mathcal{N}} g_{n,k}^{t}.
$
Thus, the local SGD update at device $n$ is followed as
$
    \boldsymbol{\theta}_{n,k}^{t+1} = \boldsymbol{\theta}_{n,k}^{t} - \eta_{k} \Tilde{g}_{n,k}^{t}
$
It is apparent that 
\begin{align}
    \bar{\boldsymbol{\theta}}_{k}^{t+1} = \bar{\boldsymbol{\theta}}_{k}^{t} - \eta_{k} \Tilde{g}_{k}^{t}. \label{eqn30}
\end{align}
It is to be mentioned that $\mathbb{E} \Tilde{g}_{k}^{t} = g_{k}^{t}$, where $\mathbb{E}$ represents function's expectation. In the subsequent analysis, we assume that $\lambda$ is an upper bound on the gradient diversity among the local objectives, i.e.,
$
    \frac{\sum_{n=1}^{N} ||g_{n,k}^{t}||_{2}^{2}}{||\sum_{n=1}^{N}g_{n,k}^{t}||_{2}^{2}} \leq \lambda.
$
In the following subsection, we delineate the foundational assumptions that underlie our convergence analysis.
\subsection{Assumptions}
\textcolor{black}{To contextualize the assumptions used in our convergence analysis, we first outline the underlying intuition motivating their adoption. The $L-\text{smoothness}$ condition ensures that the gradients of the local objective functions do not vary too abruptly, which stabilizes the updates across heterogeneous devices \cite{li2019convergence}. The Polyak-Lojasiewicz (PL) condition serves as a realistic alternative to strong convexity. This assumption has been widely used in analyzing variational quantum algorithms, where the loss landscape is nonconvex \cite{chen2024differentially}, however often satisfies PL properties in practice. Finally, the bounded variance assumption accounts for randomness introduced both by mini-batch sampling, which are inherent to PQC-based training \cite{li2019convergence}. Together, these assumptions provide a standard and well-motivated foundation for establishing convergence guarantees in frameworks such as QFL.} 
\begin{assumption}[Smoothness and Lower Boundedness \textcolor{black}{\cite{li2019convergence}}] \label{Assumption1}
\textit{The local objective function $f_{n}(.)$ associated with device $n$ is differentiable for $1 \leq n \leq N$ and is $L-\text{smooth}$, i.e., $||\nabla f_{n}(\mathbf{u}) - \nabla f_{n}(\mathbf{v})|| \leq L||\mathbf{u}-\mathbf{v}||, \forall{\mathbf{u},\mathbf{v}} \in \mathbb{R}^{d}$.}
\end{assumption}
\begin{assumption}[$\mu$-Polyak-Lojasiewicz (PL) \textcolor{black}{\cite{chen2024differentially}}] \label{Assumption2}
\textit{The global objective function $f(.)$ is differentiable and satisfy the Polyak-Lojasiewicz (PL) condition with constant $\mu$, i.e., $\frac{1}{2} ||\nabla f(\boldsymbol{\theta})||_{2}^{2} \geq \mu (f(\boldsymbol{\theta})-f({\boldsymbol{\theta}}^{*}))$ holds $\forall \boldsymbol{\theta} \in \mathbb{R}^{d}$ with ${\boldsymbol{\theta}}^{*}$ being the optimal solution of global objective.}
\end{assumption}
\begin{assumption}[Bounded Local Variance \textcolor{black}{\cite{li2019convergence}}] \label{Assumption3}
\textit{For every local dataset $S_{n}$, $n = 1, 2, \dots, N$, we can sample an independent mini-batch $\xi_{n} \subseteq \mathcal{S}_{n}$ with $|\xi_{n}| = B$ and compute an unbiased stochastic gradient $\Tilde{g_{n}} = \frac{1}{B} \nabla f(\boldsymbol{\theta};\xi_{n})$, $\mathbb{E}[\Tilde{g_{n}}] = g_{n} = \frac{1}{|\mathcal{S}_{n}|} \nabla f(\boldsymbol{\theta};\mathcal{S}_{n})$ with the variance bounded as
$
    \mathbb{E}[||\Tilde{g_{n}} - g_{n}||^{2}] \leq C_{1} ||g_{n}||^{2} + \frac{\sigma^{2}}{B}. 
$
where $C_1$ is a non-negative constant and inversely proportional to the mini-batch size, and $\sigma$ is another constant controlling the variance bound.}
\end{assumption}
\noindent
From the update rule in \eqref{eqn30} and assumption on the L-smoothness of the objective function, we have the following inequality:
$
    f(\bar{\boldsymbol{\theta}}_{k}^{t+1}) - f(\bar{\boldsymbol{\theta}}_{k}^{t}) \leq - \eta_{k} \langle \nabla f(\bar{\boldsymbol{\theta}}_{k}^{t}), \Tilde{g}_{k}^{t} \rangle + \frac{\eta_{k}^{2}  L}{2} ||\Tilde{g}_{k}^{t}||^{2}.
$
Now, taking expectation on both sides of the inequality results in 
$
    \mathbb{E}[f(\bar{\boldsymbol{\theta}}_{k}^{t+1}) - f(\bar{\boldsymbol{\theta}}_{k}^{t})] \leq -\eta_{k} \mathbb{E}[\langle \nabla f(\bar{\boldsymbol{\theta}}_{k}^{t}), \Tilde{g}_{k}^{t} \rangle] + \frac{\eta_{k}^{2}  L}{2} \mathbb{E}[||\Tilde{g}_{k}^{t}||^{2}]
$
By taking the average for all the local and global iterations, we get
\begin{align}
    &\frac{1}{KT} \sum_{k=1}^{K} \sum_{t=1}^{T} \mathbb{E}[f(\bar{\boldsymbol{\theta}}_{k}^{t+1}) - f(\bar{\boldsymbol{\theta}}_{k}^{t})] \nonumber \\
    & \hspace{6em}\leq \frac{1}{KT} \sum_{k=1}^{K} \sum_{t=1}^{T} (-\eta_{k} \mathbb{E}[\langle \nabla f(\bar{\boldsymbol{\theta}}_{k}^{t}), \Tilde{g}_{k}^{t} \rangle]) \nonumber \\
    & \hspace{6em}+ \frac{1}{KT} \sum_{k=1}^{K} \sum_{t=1}^{T} \frac{\eta_{k}^{2}  L}{2} \mathbb{E}[||\Tilde{g}_{k}^{t}||^{2}]. \label{eqn35}
\end{align}
Moving forward, we now systematically determine bounds for each term appearing on the right-hand side of \eqref{eqn35}. Specifically, Lemma~\ref{lemma1} is utilized to ascertain a bound for the first term in this equation. Subsequently, Lemma~\ref{lemma3} is employed to derive a bound for the second term. Additionally, Lemma~\ref{lemma2} focuses on a term that originates from the analysis in Lemma~\ref{lemma1}, particularly, it addresses the final term delineated in Lemma~\ref{lemma1}, providing its bound to further improve our understanding of the overall equation's dynamics.

\subsection{Convergence Rates}
In this section, we analyze the convergence rate of the QFL algorithm under full device participation scenario. We next present several lemmas that are utilized in deriving the main result. 


\textcolor{black}{Before presenting Lemma 1, we note that Assumption 1, i.e., $L-\text{smoothness}$, ensures that the gradients of the local objectives do not change abruptly, which is essential for bounding the interaction between the stochastic and full gradients. This smoothness property directly enables the inequality established in the lemma below.}

\begin{lemma} \label{lemma1}
Let Assumption \ref{Assumption1} hold, the expected inner product between stochastic gradient and full gradient is bounded by
\begin{align}
    & -\eta_{k} \mathbb{E}\bigg[\langle \nabla f(\bar{\boldsymbol{\theta}}_{k}^{t}), \Tilde{g}_{k}^{t} \rangle\bigg] \leq -\frac{\eta_{k}}{2}||\nabla f(\bar{\boldsymbol{\theta}}_{k}^{t})||^2 \nonumber \\
    & - \frac{\eta_{k}}{2}||\sum_{n=1}^{N}\nabla f_{n}(\boldsymbol{\theta}_{n,k}^{t})||^{2} + \frac{\eta_{k} L^{2}}{2} \sum_{n=1}^{N} ||\bar{\boldsymbol{\theta}}_{k}^{t} - \boldsymbol{\theta}_{n,k}^{t}||^{2}.
\end{align}
\end{lemma}
\begin{proof}
See \ref{prooflemma1} in the appendix. \renewcommand{\qedsymbol}{}
\end{proof}

\noindent
\textcolor{black}{Lemma 2 relies on Assumption 3, which bounds the variance of stochastic gradients arising from mini-batch sampling. This condition is crucial for controlling how far local iterates deviate from their averaged model, leading to the bound shown next.}

\begin{lemma} \label{lemma2}
Provided that Assumption \ref{Assumption3} is fulfilled, the expected upper bound of the divergence of $\boldsymbol{\theta}_{n,k}^{t}$ is given as 
\color{black}
\begin{align}
    & \frac{1}{K} \sum_{k=1}^{K} \sum_{n=1}^{N} \frac{1}{T_{n}}\sum_{t=1}^{T_{n}}\bigg[\mathbb{E}||\bar{\boldsymbol{\theta}}_{k}^{t} - \boldsymbol{\theta}_{n,k}^{t}||\bigg] \nonumber \\ 
    & \leq \frac{(2C_{1} + T_{n}(T_{n}+1))}{KT} \eta_{k}^{2} \frac{N+1}{N} \frac{1}{K} \sum_{k=1}^{K} \sum_{n=1}^{N} \frac{1}{T_{n}} \nonumber \\ 
    & \times \sum_{t=1}^{T_{n}} ||g_{n,k}^{t}||^{2} + \frac{\eta_{k}^{2} (N+1)(T_{n}+1)\sigma^{2}}{NB} \nonumber \\
    & \leq \frac{\lambda \eta_{K}^{2} (2C_{1} + T_{n}(T_{n}+1))}{KT} \frac{N+1}{N} \frac{1}{K} \sum_{k=1}^{K} \sum_{n=1}^{N} \nonumber \\
    & \times \frac{1}{T_{n}}\sum_{t=1}^{T_{n}} ||g_{n,k}^{t}||^{2} + \frac{\eta_{k}^{2} KT_{n} (N+1)(T_{n}+1)\sigma^{2}}{NB}.
\end{align}
\color{black}
\end{lemma}
\noindent
\begin{proof}
See \ref{prooflemma2} in the appendix. \renewcommand{\qedsymbol}{}
\end{proof}

\noindent
\textcolor{black}{Lemma 3 again relies on Assumption 3, as the bounded-variance condition allows us to upper-bound the second moment of the stochastic gradients. This assumption ensures that the randomness introduced by stochastic sampling does not accumulate uncontrollably during local updates.}

\begin{lemma} \label{lemma3}
Under Assumption \ref{Assumption3}, the expected upper bound of $\mathbb{E}[||\Tilde{g}_{k}^{t}||^2]$ is expressed as
\begin{align}
    \mathbb{E}\bigg[||\Tilde{g}_{k}^{t}||^2\bigg] &\leq \bigg(\frac{C_{1}}{N}+1\bigg) \bigg[\sum_{n=1}^{N}||\nabla f_{n}(\boldsymbol{\theta}_{n,k}^{t})||^{2}\bigg] + \frac{\sigma^{2}}{NB} \nonumber \\ 
    &\leq \lambda \bigg(\frac{C_{1}}{N}+1\bigg) \bigg[\sum_{n=1}^{N}||\nabla f_{n}(\boldsymbol{\theta}_{n,k}^{t})||^{2}\bigg] + \frac{\sigma^{2}}{NB}.
\end{align}
\end{lemma}
\begin{proof}
See \ref{prooflemma3} in the appendix. \renewcommand{\qedsymbol}{}
\end{proof}

\noindent
\textcolor{black}{Lemma 4 exploits the structure of quantum measurement noise by relating the variance of the gradient estimator to the number of measurement shots performed on the PQC. This characterization is consistent with standard statistical models used in VQAs.}

\begin{lemma} \label{lemma4}
The variance of the gradient estimate can be upper bounded as 
\begin{align}
\text{var}(\xi_{k}^{t}) \leq \frac{1}{N} \sum_{n \in \mathcal{N}} \frac{\nu N_{z} D Tr(Z^{2})}{2H}. \label{eqn45}
\end{align}
\end{lemma}
\begin{proof}
See \ref{prooflemma4} in the appendix. \renewcommand{\qedsymbol}{}
\end{proof}

\noindent
\textit{\textbf{Theorem 1.}}
\textit{Let Assumptions \ref{Assumption1}, \ref{Assumption2}, \ref{Assumption3} hold, then the upper bound of the convergence rate of the global model training considering full device participation after $K$ global rounds satisfies}
\color{black}
\begin{align}
        &\frac{1}{K} \frac{1}{\frac{1}{N}\sum_{n=1}^{N}T_{n}} \sum_{k=1}^{K} \sum_{t=1}^{T_{n}} \mathbb{E}||\nabla f(\bar{\boldsymbol{\theta}}_{k}^{t})||^{2} \leq \frac{2 [f(\bar{\boldsymbol{\theta}}_{1}^{0}) - f^{*}]}{\eta_{k} K\frac{1}{N}\sum_{n=1}^{N}T_{n}} \nonumber \\
        &+ \frac{L \eta \sigma^{2}}{NB} + \frac{2 \eta_{k}^{2} \sigma^{2} L^{2} (\frac{1}{N}\sum_{n=1}^{N}T_{n}+1)}{B} \bigg(1+\frac{1}{N}\bigg) \nonumber \\
        &+ \frac{1}{N} \sum_{n \in \mathcal{N}} \frac{\nu N_{z} D Tr(Z^{2})}{2H}. \label{eqn42new}
\end{align}
\color{black}
\begin{proof}
See \ref{prooftheorem1} in the appendix. \renewcommand{\qedsymbol}{}
\end{proof}

\begin{remark}
The convergence upper bound in Theorem 1 achieved under full device participation scenario shows that the convergence of QFL algorithm is significantly influenced by the total number of iterations, the number of total devices, and the number of measurements performed on the PQC.
\end{remark}
\begin{remark}
It becomes apparent from \eqref{eqn42new} that increasing the number of devices $N$ in a full device participation scenario enhances the convergence performance of QFL. Moreover, a higher number of quantum measurement shots of the PQC also helps in achieving faster QFL convergence. However, the last term of this equation encompasses the summation of the variance of the gradient estimate. Specifically, this reflects the noise variance associated with quantum shot noise. Therefore, as more devices participate in QFL, the accumulated noise from all devices intensifies with each global round. Hence, a simultaneous increase in the number of devices and measurement shots in moderation proves to be more effective for enhancing convergence performance. We will conduct extensive simulations in Section~\ref{simulationresultsandevaluations} to further investigate this phenomenon.
\end{remark}

\subsection{Complexity Analysis}
We now study the complexity analysis of of the quantum training process. \textcolor{black}{Let Assumption \ref{Assumption1} hold, which states that the local loss function $f_{n}(\cdot)$ is $L-\text{smooth}$, i.e., has Lipschitz-continuous gradients; under this condition,} we obtain the following inequality:
$
    f_{n}(\boldsymbol{\theta}) \leq f_{n}(\boldsymbol{w'}) + \nabla {f_{n}(\boldsymbol{\theta}')}^{T}(\boldsymbol{\theta}-\boldsymbol{w'}) + \frac{L}{2}||\boldsymbol{\theta}-\boldsymbol{w'}||^{2},
$
for all $\boldsymbol{\theta}, \boldsymbol{w'} \in \Theta$. Furthermore, Assumption \ref{Assumption2} applies, \textcolor{black}{which imposes the $\mu$-PL condition commonly used in non-convex optimization settings to ensure functional descent,} there exists a constant $\mu > 0$ and $\mu \leq L$ such that the inequality
$
    ||\nabla f_{n}(\boldsymbol{\theta})||^{2} \geq 2 \mu (f_{n}(\boldsymbol{\theta}) - {f_{n}}^{*}),
$
holds for all $\boldsymbol{\theta} \in \Theta$. In practice, the constant $\mu$ depends on the number of qubits $n$ and the number of gate parameters $P$, as a larger circuit is typically characterized by smaller gradient norms. This correlation reflects how the complexity of quantum circuits influences the behavior of their gradients, with more extensive and parameter-rich systems often exhibiting more nuanced and subdued gradient magnitudes. The following theorem gives a bound on the optimality of the local quantum training output \textcolor{black}{$\boldsymbol{\theta}^{T_{n}}$ after $T_{n}$} iterations. The proof of the theorem is derived from the classical convergence analysis of Stochastic Gradient Descent (SGD) and is detailed in (\cite{ajalloeian2020convergence}, Theorem 6). Under the assumption of $L$-smoothness and $\mu$-PL, for any given initial point $\boldsymbol{\theta}^{0}$, the following bound holds for any fixed learning rate $\eta_{k} = \mu \leq \frac{1}{L}$:
$
    \textcolor{black}{\mathbb{E}[f_{n}(\boldsymbol{\theta}^{T_{n}})]} - {L}^{*} \leq \textcolor{black}{(1 - \eta \mu)^{T_{n}}} ([f_{n}(\boldsymbol{\theta}^{0})] - {L}^{*}) + \frac{1}{2} [\frac{\eta L V}{\mu}],
$
where $V = \frac{\nu N_{z} D Tr(Z^{2})}{2H}$ is the upper bound of the variance of the gradient estimate, and the expectation is taken over the distribution of the measurement outputs. Furthermore, given some target error level $\delta > 0$, for learning rate $\eta = \eta^{\text{shot-noise}} \leq \min\{\frac{1}{L},\frac{\delta \mu}{L V}\}$, a number of iteration, given as
\begin{align}
    \textcolor{black}{T_{n}^{\text{shot-noise}}} = \mathcal{O} \bigg( \log\frac{1}{\delta} + \frac{V}{\delta \mu} \bigg) \frac{L}{\mu}, \label{eqn26new}
\end{align}
is sufficient to ensure an error $\mathbb{E}[f_{n}(\boldsymbol{\theta}^{T})-{f_{n}}^{*}] = \mathcal{O}(\delta)$.  By choosing the learning rate $\eta$ sufficiently small, this error can be made arbitrarily small, thereby guaranteeing convergence in at most $T_{n}^{\text{shot-noise}}$ iterations. This strategic adjustment reduces the computational overhead by stabilizing the training process, thereby preventing excessive iterations that could strain system resources.

\section{Sum-rate Maximization Problem Formulation} \label{section:problemformulation}
\subsection{Problem Statement}
In this work, we aim to maximize the sum-rate of the network under consideration by jointly optimizing channel selection and transmit power of the quantum devices during each block. \textcolor{black}{Maximizing the uplink sum-rate helps reduce the total transmission delay for sending locally trained PQC parameters from devices to the BS, which often serves as a major bottleneck in multi-channel NOMA-based wireless QFL systems. By ensuring faster and more reliable model uploads, the system achieves better synchronization and more timely global aggregation. This, in turn, leads to faster convergence and improved scalability of the overall QFL training process.} Therefore, we formulate the problem as

\noindent
\textit{\underline{Problem 1:}}
\begin{subequations} 
\begin{align}
\max_{\zeta^{(b)}, p^{(b)}} \quad & \sum_{n=1}^{N} \sum_{c=1}^{C} r_{n,c}^{(b)} \label{eqn6a}\\ 
\textrm{s.t.} \quad & 0 < p_{n}^{(b)} \le P^{\max}, \quad \forall n \in \mathcal{N}, \label{eqn6b}\\
& \zeta_{n,c}^{(b)} \in \{0,1\}, \quad \forall n \in \mathcal{N}, \forall c \in \mathcal{C}, \label{eqn6c}\\
& \sum_{c=1}^{C} \zeta_{n,c}^{(b)} \leq 1, \quad \forall n \in \mathcal{N}, \label{eqn6d}
\end{align} 
\end{subequations}
where $\zeta^{(b)} = \{\zeta_{1,1}^{(b)}, \zeta_{1,2}^{(b)}, \dots, \zeta_{N,C}^{(b)}\} \in \mathbb{R}_{1 \times NC}$ and $p^{(b)} = \{p_{1}^{(b)}, p_{2}^{(b)}, \dots, p_{N}^{(b)}\} \in \mathbb{R}_{1 \times N}$. 
In \textit{Problem 1}, constraint \eqref{eqn6b} ensures the transmit power of device $n$ remains positive and does not exceed the maximum limit $P^{\max}$. While constraint \eqref{eqn6c} is about the channel selection being a binary variable, constraint \eqref{eqn6d} arises from the assumption that each device is allowed to select at most one channel during a time block. 

\subsection{Nature of the Formulated Problem}
\textit{Problem 1} is non-convex in nature because of the objective function \eqref{eqn6a}. Moreover, it is a mixed-integer nonlinear programming problem (MINLP) because it feature a combination of discrete and continuous variables. Furthermore, even for a given $\zeta^{(b)}$, it becomes NP-hard due to the coupling of $p^{(b)}$ \cite{luo2008dynamic}. Due to its complex nature, \textit{Problem 1} is extremely challenging to solve and intractable for classical optimization methods. 

\section{Proposed Quantum-Centric Optimization Solution} \label{section:proposedsolution}
In this work, we propose a quantum-centric optimization strategy utilizing quantum approximate optimization algorithm (QAOA) to address \textit{Problem 1}. Initially, we transform the classical MINLP problem into a quadratic unconstrained binary optimization (QUBO) framework. Subsequently, we convert the QUBO formulation into a Hamiltonian expression suitable for execution on a quantum processor. \textcolor{black}{The proposed QAOA-based solution is particularly well-suited to the discrete, high-dimensional resource allocation problem encountered in large-scale, multi-channel NOMA-based QFL systems. Its ability to efficiently explore combinatorial solution spaces enables fast, high-quality approximate optimization. Although NISQ devices are inherently noise-prone, this challenge is effectively mitigated by using a sufficient number of measurement shots, ensuring that QAOA remains both reliable and practical in realistic noisy quantum hardware settings.} For computational convenience, we divide \textit{Problem 1} into two distinct subproblems and apply the aforementioned steps to each subproblem. We then utilize block coordinate descent (BCD) technique to derive solutions for the original formulation of \textit{Problem 1}.

\noindent
\textit{\underline{Sub-problem 1:}}
\begin{subequations} 
\begin{align}
\max_{\zeta^{(b)}} \quad & \sum_{n=1}^{N} \sum_{c=1}^{C} r_{n,c}^{(b)} \label{eqn7a}\\ 
\textrm{s.t.} \quad & \zeta_{n,c}^{(b)} \in \{0,1\}, \quad \forall n \in \mathcal{N}, \forall c \in \mathcal{C}, \label{eqn7b}\\
& \sum_{c=1}^{C} \zeta_{n,c}^{(b)} \leq 1, \quad \forall n \in \mathcal{N}. \label{eqn7c}
\end{align} 
\end{subequations}

\noindent
\textit{\underline{Sub-problem 2:}}
\begin{subequations} 
\begin{align}
\max_{p^{(b)}} \quad & \sum_{n=1}^{N} \sum_{c=1}^{C} r_{n,c}^{(b)} \label{eqn8a}\\ 
\textrm{s.t.} \quad & 0 < p_{n}^{(b)} \le P^{\max}, \quad \forall n \in \mathcal{N}. \label{eqn8b}
\end{align} 
\end{subequations}

\subsection{Sub-problem 1}
\noindent
\textit{Transformation of Classical Problem into QUBO Formulation:}
In order to transform the classical problem into a QUBO formulation, we introduce penalty parameters $\lambda_{0}$ and $\lambda_{1}$. Using these parameters, \textit{Sub-problem 1} is reformulated as follows:
\begin{subequations} 
\begin{align}
\max_{\zeta^{(b)}} \bigg[& \lambda_{0} \sum_{n=1}^{N} \sum_{c=1}^{C} r_{n,c}^{(b)} + \lambda_{1} \sum_{n=1}^{N} \bigg(\sum_{c=1}^{C} \zeta_{n,c}^{(b)} - 1\bigg)^{2}\bigg] \label{eqn9a}
\end{align} 
\end{subequations}
In \eqref{eqn9a}, the parameter $\lambda_{0}$ ensures effective consideration of the quantum devices' channel selection in the optimization process of the objective function, while the parameter $\lambda_{1}$ enforces the constraint on channel selection. It is to note that a problem must be expressed in a specific quadratic form in order to be solved by implementing QAOA, which is
\begin{align}
    \min_{\zeta^{(b)}} {\zeta^{(b)}}^{T} Q \zeta^{(b)}. \label{eqn10}
\end{align}
From \eqref{eqn10}, we can say that QAOA is typically designed to minimize a particular cost function. The cost function is constructed in such a way that its minimum value corresponds to the optimal solution of the problem. However, this work aims to maximize sum-rate of the QFL framework. Hence, we convert our maximization problem into a minimization problem by taking a negative of the original function as follows:
\begin{subequations} 
\begin{align}
\min_{\zeta^{(b)}} - \bigg[& \lambda_{0} \sum_{n=1}^{N} \sum_{c=1}^{C} r_{n,c}^{(b)} + \lambda_{1} \sum_{n=1}^{N} \bigg(\sum_{c=1}^{C} \zeta_{n,c}^{(b)} - 1\bigg)^{2}\bigg] \label{eqn11a}
\end{align} 
\end{subequations}
or,
\begin{subequations} 
\begin{align}
\min_{\zeta^{(b)}} \bigg[&\underbrace{- \lambda_{0} \sum_{n=1}^{N} \sum_{c=1}^{C} r_{n,c}^{(b)}}_{\text{(a)}} \underbrace{- \lambda_{1} \sum_{n=1}^{N} \bigg(\sum_{c=1}^{C} \zeta_{n,c}^{(b)} - 1\bigg)^{2}}_{\text{(b)}}\bigg] \label{eqn12a}
\end{align} 
\end{subequations}
We now construct the Q matrix based on the problem formulation in \eqref{eqn11a}.

\noindent
\textit{Construction of Corresponding Q Matrix:}
The Q matrix for a QUBO problem is a representation of the corresponding objective function. Hence, in case of \textit{Sub-problem 1}, the Q matrix embodies the objective function for optimizing channel selection of devices in our considered QFL network. In order to construct the Q matrix, we expand each term in \eqref{eqn11a} as follows:

\noindent
\textit{\underline{(a):}}
{
\footnotesize
\begin{align}
    &- \lambda_{0} \sum_{n=1}^{N} \sum_{c=1}^{C} r_{n,c}^{(b)} \nonumber \\
    &= - \lambda_{0} \sum_{n=1}^{N} \sum_{c=1}^{C} B \log_{2}(1 + \gamma_{n,c}^{(b)}) \nonumber \\
    &= - \lambda_{0} B \sum_{n=1}^{N} \sum_{c=1}^{C} \log_{2}\bigg(1 + \frac{\zeta_{n,c}^{(b)} h_{n,c}^{(b)} p_{n}^{(b)}}{\sigma^{2} + \sum_{m \in \mathcal{N},m \neq n} \zeta_{m,c}^{(b)} h_{m,n,c}^{b} p_{m}^{(b)}}\bigg). \label{eqn13}
\end{align}
}
Using Taylor expansion, we write \eqref{eqn13} as,
\begin{align}
    - \lambda_{0} B \sum_{n=1}^{N} \sum_{c=1}^{C} \frac{\zeta_{n,c}^{(b)} h_{n,c}^{(b)} p_{n}^{(b)}}{\ln{2}(\sigma^{2} + \sum_{m \in \mathcal{N},m \neq n} \zeta_{m,c}^{(b)} h_{m,n,c}^{b} p_{m}^{(b)})}. \label{eqn14}
\end{align}
\textcolor{black}{Here, using the first-order Taylor approximation is sufficient, as the large-scale, multi-channel NOMA uplink system typically operates under low-SINR conditions due to co-channel interference among multiple devices sharing the same channel. This operating regime ensures the approximation remains accurate while enabling a tractable QUBO-based formulation for QAOA. Moreover, this linearization technique has been widely employed in prior wireless optimization works (e.g., \cite{tran2021uav, tan2011spectrum, d2018learning}) to simplify logarithmic rate expressions under similar interference-limited scenarios.}

\noindent
\textit{\underline{(b):}}
\begin{align}
    &- \lambda_{1} \sum_{n=1}^{N} \bigg(\sum_{c=1}^{C} \zeta_{n,c}^{(b)} - 1\bigg)^{2} \nonumber \\
    &=- \lambda_{1} \sum_{n=1}^{N} \bigg(\sum_{c=1}^{C} {\zeta_{n,c}^{(b)}}^{2} - 2 \zeta_{n,c}^{(b)} + 1\bigg). \label{eqn15}
\end{align}

Given binary variable $\zeta_{n,c}^{b}$, where $N$ denotes the number of devices and $C$ is the number of channels, the matrix entries are determined as follows:

The diagonal entries, denoted by $Q_{(n,c),(n,c)}$, represent linear terms associated with individual channel selection decisions. They are calculated as
\begin{align}
    Q_{(n,c),(n,c)} =& -\lambda_{0} B \frac{ h_{n,c}^{(b)} p_{n}^{(b)}}{\ln{2}(\sigma^{2} + \sum_{m \in \mathcal{N},m \neq n} \zeta_{m,c}^{(b)} h_{m,n,c}^{b} p_{m}^{(b)})} \nonumber \\
    &\hspace{8em}+2 \lambda_{1},
\end{align}
where the term $-\lambda_{0} B \frac{ h_{n,c}^{(b)} p_{n}^{(b)}}{\ln{2}(\sigma^{2} + \sum_{m \in \mathcal{N},m \neq n} \zeta_{m,c}^{(b)} h_{m,n,c}^{b} p_{m}^{(b)})}$ applies a penalty related to the channel gain, transmit power and other noise terms associated with device $n$. Besides, the term $2 \lambda_{1}$ enforces a penalty regarding the channel selection constraint of device $n$.

The off-diagonal entries, denoted by $Q_{(n,c),(x,y)}$, depicts quadratic interactions between channel selection decisions. They are calculated as
\begin{align}
    Q_{(n,c),(x,y)} = 
    \begin{cases} 
        -\lambda_{1}, & \text{if } n = x \hspace{1em} \text{and} \hspace{1em} c \neq y, \\
        0, & \text{otherwise,}
    \end{cases}
\end{align}
where the term $-\lambda_{1}$ is an availability constraint penalty applied when the same device $n$ is allocated different channels $c$ and $y$. This penalty ensures that each device is assigned only one channel. All other off-diagonal elements are set to 0 as they represent independent decisions without any interaction.

\noindent
\textit{Transformation of QUBO Formulation into Hamiltonian Expression:}
In order to transform the QUBO formulation into a Hamiltonian expression, we replace each binary variable $\zeta_{n,c}^{(b)} \in \{0,1\}$ with spin variable $z_{n,c}^{(b)} \in \{-1,1\}$, where
\begin{align}
    \zeta_{n,c}^{(b)} = \frac{1 - z_{n,c}^{(b)}}{2}. \label{eqn18}
\end{align}
We now have to substitute \eqref{eqn18} into \eqref{eqn12a}. In order to do that, we substitute \eqref{eqn18} into \eqref{eqn14} and \eqref{eqn15}. Putting \eqref{eqn18} into \eqref{eqn14}, we get,
\begin{align}
    \hspace{-.2em}- \lambda_{0} B \sum_{n=1}^{N} \sum_{c=1}^{C} \frac{(1 - z_{n,c}^{(b)}) h_{n,c}^{(b)} p_{n}^{(b)}}{2\ln{2}(\sigma^{2} + \sum_{m \in \mathcal{N},m \neq n} (\frac{1 - z_{m,c}^{(b)}}{2}) h_{m,n,c}^{b} p_{m}^{(b)})}. \label{eqn19}
\end{align}
When we expand the expression in \eqref{eqn19}, it yields a combination of constants and linear terms in $z_{n,c}^{(b)}$. Similarly, substituting \eqref{eqn18} into \eqref{eqn15}, we get,
\begin{align}
    - \lambda_{1} \sum_{n=1}^{N} \bigg(\sum_{c=1}^{C} {\bigg(\frac{1 - z_{n,c}^{(b)}}{2}\bigg)}^{2} - 2 \bigg(\frac{1 - z_{n,c}^{(b)}}{2}\bigg) + 1\bigg). \label{eqn20}
\end{align}
Expanding the expression in \eqref{eqn20} results in a combination of constants, linear terms in $z_{n,c}^{(b)}$, and quadratic terms in $z_{n,c}^{(b)} z_{n,e}^{(b)}$. After expanding all terms, they are subsequently restructured into a Hamiltonian format where each spin variable $z_{n,c}^{(b)}$ is represented by a corresponding Pauli Z operator. Specifically, every linear term involving $z_{n,c}^{(b)}$ is mapped to an individual $Z_{n,c}^{(b)}$ operator, and each quadratic term involving products such as $z_{n,c}^{(b)} z_{n,e}^{(b)}$ is expressed as the product of two Pauli Z operators, $Z_{n,c}^{(b)} Z_{n,e}^{(b)}$. Consequently, the Hamiltonian is formulated as:
\begin{align}
    H_{C} = \sum_{i} I_{i}Z_{i} + \sum_{i < j} J_{ij} Z_{i}Z_{j},
\end{align}
where $I_{i}$ denotes the coefficients corresponding to each linear $Z_{i}$ term derived from $z_{n,c}^{(b)}$, and $J_{ij}$ are the coefficients for each quadratic $Z_{i} Z_{j}$ term, originating from the quadratic interactions between $z_{n,c}^{(b)}$ and $z_{n,e}^{(b)}$.

\subsection{Sub-Problem 2}
We note that \textcolor{black}{QAOA} is tailored for optimizing binary variables. However, in \textit{Sub-problem 2}, we address the issue of transmit power, which must be transformed into a binary variable to be compatible with QAOA. Hence, we convert $p_{n}^{(b)}$ into binary variable $x_{n}^{(b)}$ as follows:
\begin{align}
    p_{n}^{(b)} = \sum_{i=0}^{q-1}2^{i}x_{n,i}^{(b)},
\end{align}
where $q$ is the number of bits required to convert the decimal value of transmit power into binary. We denote the set of binary variables related to transmit power as $x^{(b)} = \{x_{1,0}^{(b)}, x_{1,1}^{(b)}, \dots, x_{1,q-1}^{(b)}, x_{2,0}^{(b)}, \dots, x_{N,q-1}^{(b)}\}$, and the set of required bits as $\mathcal{\iota} = \{1, 2, \dots, q-1\}$. \textit{Sub-problem 2} now takes the form:

\noindent
\textit{\underline{Sub-problem 2 (equivalent):}}
\begin{subequations} 
\begin{align}
\max_{x^{(b)}} \quad & \sum_{n=1}^{N} \sum_{c=1}^{C} r_{n,c}^{(b)} \label{eqn23a}\\ 
\textrm{s.t.} \quad & x_{n,i}^{(b)} \in \{0,1\}, \quad \forall n \in \mathcal{N}, \forall i \in \mathcal{\iota}, \label{eqn23b}\\
&\sum_{i=0}^{q-1}2^{i}x_{n,i}^{(b)} > 0, \quad \forall n \in \mathcal{N}, \label{eqn23c}\\
&\sum_{i=0}^{q-1}2^{i}x_{n,i}^{(b)} \le P^{\max}, \quad \forall n \in \mathcal{N}, \label{eqn23d}
\end{align} 
\end{subequations}
where 
$
    r_{n,c}^{(b)} = B \log_{2}(1 + \gamma_{n,c}^{(b)}),
$
and
$
    \gamma_{n,c}^{(b)} = \frac{\zeta_{n,c}^{(b)} h_{n,c}^{(b)} \sum_{i=0}^{q-1}2^{i}x_{n,i}^{(b)}}{\sigma^{2} + \sum_{m \in \mathcal{N},m \neq n} \zeta_{m,c}^{(b)} h_{m,n,c}^{b} \sum_{i=0}^{q-1}2^{i}x_{m,i}^{(b)}}.
$

\noindent
\textit{Transformation of Classical Problem into QUBO Formulation:}
We now transform the classical problem into a QUBO formulation by introducing penalty parameters $\lambda_{2}$, $\lambda_{3}$, and $\lambda_{4}$. Thus, \textit{Sub-problem 2 (equivalent)} is restructured as:

$
\max_{x^{(b)}} \bigg[ \lambda_{2} \sum_{n=1}^{N} \sum_{c=1}^{C} r_{n,c}^{(b)} + \lambda_{3} \sum_{n=1}^{N} \bigg(0 - \sum_{i=0}^{q-1}2^{i}x_{n,i}^{(b)}\bigg)^{2} 
+ \lambda_{4} \sum_{n=1}^{N} \bigg(\sum_{i=0}^{q-1}2^{i}x_{n,i}^{(b)} - P^{\max}\bigg)^{2}\bigg].
$
We adjust our maximization problem to a minimization format by negating the original function as follows:
$
\min_{x^{(b)}} - \bigg[ \lambda_{2} \sum_{n=1}^{N} \sum_{c=1}^{C} r_{n,c}^{(b)} + \lambda_{3} \sum_{n=1}^{N} \bigg(- \sum_{i=0}^{q-1}2^{i}x_{n,i}^{(b)}\bigg)^{2}
+ \lambda_{4} \sum_{n=1}^{N} \bigg(\sum_{i=0}^{q-1}2^{i}x_{n,i}^{(b)} - P^{\max}\bigg)^{2}\bigg],
$
or,
\begin{subequations} 
\begin{align}
&\min_{x^{(b)}} \bigg[ \underbrace{- \lambda_{2} \sum_{n=1}^{N} \sum_{c=1}^{C} r_{n,c}^{(b)}}_{\text{(a)}} \underbrace{+ \lambda_{3} \sum_{n=1}^{N} \bigg(\sum_{i=0}^{q-1}2^{i}x_{n,i}^{(b)}\bigg)^{2}}_{\text{(b)}} \nonumber \\
&\hspace{5em}\underbrace{- \lambda_{4} \sum_{n=1}^{N} \bigg(\sum_{i=0}^{q-1}2^{i}x_{n,i}^{(b)} - P^{\max}\bigg)^{2}}_{\text{(c)}}\bigg]. \label{eqn28a}
\end{align} 
\end{subequations}
We now develop the Q matrix based on the problem formulation outlined in \eqref{eqn28a}. 

\noindent
\textit{Construction of Corresponding Q Matrix:}
In case of \textit{Sub-problem 2}, the Q matrix represent the objective function for optimizing transmit power of devices in our QFL framework. To build the Q matrix, we expand each term in \eqref{eqn11a} as follows:

\noindent
\textit{\underline{(a):}}
\begin{align}
    &- \lambda_{2} \sum_{n=1}^{N} \sum_{c=1}^{C} r_{n,c}^{(b)} = - \lambda_{2} \sum_{n=1}^{N} \sum_{c=1}^{C} B \log_{2}(1 + \gamma_{n,c}^{(b)}) \nonumber \\
    &= - \lambda_{2} B \sum_{n=1}^{N} \sum_{c=1}^{C} \log_{2}\bigg(1 \nonumber \\
    &+ \frac{\zeta_{n,c}^{(b)} h_{n,c}^{(b)} \sum_{i=0}^{q-1}2^{i}x_{n,i}^{(b)}}{\sigma^{2} + \sum_{m \in \mathcal{N},m \neq n} \zeta_{m,c}^{(b)} h_{m,n,c}^{b} \sum_{i=0}^{q-1}2^{i}x_{m,i}^{(b)}}\bigg). \label{eqn29}
\end{align}

Using Taylor expansion, we write \eqref{eqn29} as,
{
\footnotesize
\begin{align}
    - \lambda_{2} B \sum_{n=1}^{N} \sum_{c=1}^{C} \frac{\zeta_{n,c}^{(b)} h_{n,c}^{(b)} \sum_{i=0}^{q-1}2^{i}x_{n,i}^{(b)}}{\ln{2}(\sigma^{2} + \sum_{m \in \mathcal{N},m \neq n} \zeta_{m,c}^{(b)} h_{m,n,c}^{b} \sum_{i=0}^{q-1}2^{i}x_{m,i}^{(b)})}. \label{eqn30}
\end{align}
}
\textit{\underline{(b):}}
\begin{align}
    \lambda_{3} \sum_{n=1}^{N} \bigg(\sum_{i=0}^{q-1}2^{i}x_{n,i}^{(b)}\bigg)^{2} = \lambda_{3} \sum_{n=1}^{N} \sum_{i=0}^{q-1} \bigg(2^{i}\bigg)^{2} \bigg(x_{n,i}^{(b)}\bigg)^{2}. \label{eqn31}
\end{align}
\textit{\underline{(c):}}
\begin{align}
    &- \lambda_{4} \sum_{n=1}^{N} \bigg(\sum_{i=0}^{q-1}2^{i}x_{n,i}^{(b)} - P^{\max}\bigg)^{2} \nonumber \\
    &= - \lambda_{4} \sum_{n=1}^{N} \bigg(\sum_{i=0}^{q-1} \bigg(2^{i}\bigg)^{2} \bigg(x_{n,i}^{(b)}\bigg)^{2} - 2 \sum_{i=0}^{q-1}2^{i}x_{n,i}^{(b)} P^{\max} \nonumber \\
    &\hspace{10em} + \bigg(P^{\max}\bigg)^{2}\bigg). \label{eqn32}
\end{align}
Considering the binary variable $x_{n,i}^{b}$, with $N$ representing the number of devices and $(q-1)$ indicating the number of required bits for expressing transmit power of a device in binary, the matrix entries are determined as follows: 

We represent the diagonal entries as $Q_{(n,i),(n,i)}$ which depicts linear terms associated with individual transmit power. They are calculated as
\begin{align}
    &Q_{(n,i),(n,i)} \nonumber \\
    &= - \lambda_{2} B \frac{\zeta_{n,c}^{(b)} h_{n,c}^{(b)} \sum_{i=0}^{q-1}2^{i}x_{n,i}^{(b)}}{\ln{2}(\sigma^{2} + \sum_{m \in \mathcal{N},m \neq n} \zeta_{m,c}^{(b)} h_{m,n,c}^{b} \sum_{i=0}^{q-1}2^{i}x_{m,i}^{(b)})}, \nonumber \\
    &\hspace{10em} - 2 \sum_{i=0}^{q-1}2^{i}x_{n,i}^{(b)} P^{\max},
\end{align}
where the term {\footnotesize$- \lambda_{2} B \frac{\zeta_{n,c}^{(b)} h_{n,c}^{(b)} \sum_{i=0}^{q-1}2^{i}x_{n,i}^{(b)}}{\ln{2}(\sigma^{2} + \sum_{m \in \mathcal{N},m \neq n} \zeta_{m,c}^{(b)} h_{m,n,c}^{b} \sum_{i=0}^{q-1}2^{i}x_{m,i}^{(b)})}$} enacts a penalty related to the channel selection, channel gain, transmit power and other noise terms associated with device $n$. Besides, the term $- 2 \sum_{i=0}^{q-1}2^{i}x_{n,i}^{(b)} P^{\max}$ employs a penalty regarding the transmit power constraint of device $n$.

The off-diagonal entries, denoted by $Q_{(n,i),(x,y)}$, depicts quadratic interactions between channel selection decisions. They are calculated as
\begin{align}
    Q_{(n,i),(x,y)} = 
    \begin{cases} 
        (\lambda_{3} - \lambda_{4}) \sum_{i=0}^{q-1} \bigg(2^{i}\bigg)^{2}, \text{if } &n = x, \nonumber \\
        &i \neq y, \\
        0, & \text{otherwise,}
    \end{cases}
\end{align}
where the term $(\lambda_{3} - \lambda_{4})$ is a constraint penalty regarding transmit power of device $n$. All other off-diagonal entries are set to 0 as they represent independent decisions without any interaction.

\noindent
\textit{Transformation of QUBO Formulation into Hamiltonian Expression:}
To convert the QUBO formulation into a Hamiltonian expression, we replace each binary variable $x_{n,i}^{(b)} \in \{0,1\}$ by spin variable $z_{n,i}^{(b)} \in \{-1,1\}$, where
\begin{align}
    x_{n,i}^{(b)} = \frac{1 - z_{n,i}^{(b)}}{2}. \label{eqn34}
\end{align}
We then insert \eqref{eqn34} into \eqref{eqn28a}. To facilitate this, we replace \eqref{eqn34} in \eqref{eqn30}, \eqref{eqn31}, and \eqref{eqn32}. Putting \eqref{eqn34} into \eqref{eqn30}, we get,
{
\footnotesize
\begin{align}
    - \lambda_{2} B \sum_{n=1}^{N} \sum_{c=1}^{C} \frac{\zeta_{n,c}^{(b)} h_{n,c}^{(b)} \sum_{i=0}^{q-1}2^{i}(1 - z_{n,i}^{(b)})}{2\ln{2}(\sigma^{2} + \sum_{m \in \mathcal{N},m \neq n} \zeta_{m,c}^{(b)} h_{m,n,c}^{b} \sum_{i=0}^{q-1}2^{i}\frac{1 - z_{n,i}^{(b)}}{2})}. \label{eqn35}
\end{align}
}
Expanding the expression in \eqref{eqn35} yields a combination of constants and linear terms in $z_{n,i}^{(b)}$. Likewise, substituting \eqref{eqn34} into \eqref{eqn31}, we get,
\begin{align}
    \lambda_{3} \sum_{n=1}^{N} \sum_{i=0}^{q-1} \bigg(2^{i}\bigg)^{2} \bigg(\frac{1 - z_{n,i}^{(b)}}{2}\bigg)^{2}. \label{eqn36}
\end{align}
Expanding the expression in \eqref{eqn36} results in a combination of constants, linear terms in $z_{n,i}^{(b)}$, and quadratic terms in $z_{n,i}^{(b)} z_{n,f}^{(b)}$. Similarly, Putting \eqref{eqn34} into \eqref{eqn32}, we get,
{
\footnotesize
\begin{align}
    &- \lambda_{4} \sum_{n=1}^{N} \bigg(\sum_{i=0}^{q-1} \bigg(2^{i}\bigg)^{2} \bigg(\frac{1 - z_{n,i}^{(b)}}{2}\bigg)^{2} - 2 \sum_{i=0}^{q-1}2^{i}\bigg(\frac{1 - z_{n,i}^{(b)}}{2}\bigg) P^{\max} \nonumber \\
    &\hspace{10em} + \bigg(P^{\max}\bigg)^{2}\bigg). \label{eqn37}
\end{align}
}
When the expression in \eqref{eqn37} is expanded, it produces constants, linear terms represented by $z_{n,i}^{(b)}$, and quadratic terms involving $z_{n,i}^{(b)} z_{n,f}^{(b)}$.
Thus, after expanding all terms, they are formed into a Hamiltonian format where each spin variable $z_{n,i}^{(b)}$ is represented by a corresponding Pauli Z operator. Specifically, every linear term involving $z_{n,i}^{(b)}$ is mapped to an individual $Z_{n,i}^{(b)}$ operator, and each quadratic term involving products such as $z_{n,i}^{(b)} z_{n,f}^{(b)}$ is expressed as the product of two Pauli Z operators, $Z_{n,c}^{(b)} Z_{n,f}^{(b)}$. Consequently, the Hamiltonian is formulated as:
\begin{align}
    H_{C} = \sum_{\rho} I_{\rho}Z_{\rho} + \sum_{\rho < \phi} J_{\rho \phi} Z_{\rho}Z_{\phi},
\end{align}
where $I_{\rho}$ denotes the coefficients corresponding to each linear $Z_{\rho}$ term derived from $z_{n,i}^{(b)}$, and $J_{\rho \phi}$ are the coefficients for each quadratic $Z_{\rho} Z_{\phi}$ term, originating from the quadratic interactions between $z_{n,i}^{(b)}$ and $z_{n,f}^{(b)}$.

\subsection{Proposed Algorithm}
Based on our development, we propose a quantum-centric optimization algorithm tailored to address \textit{Sub-problem 1} and \textit{Sub-problem 2} using the QAOA, as presented in Algorithm \eqref{alg:qaoa}. As mentioned before, we convert each sub-problem as QUBO formulation, eventually transforming them into Hamiltonian expression. The algorithm iteratively optimizes an ansatz circuit using gradient descent optimizer. The final solution yields the optimized binary variable matrix (channel selection matrix $\zeta^{(b)}$ for \textit{Sub-problem 1} or transmit power matrix $x^{(b)}$ for \textit{Sub-problem 2}), as well as the minimized objective value. \textcolor{black}{Therefore, Algorithm \ref{algorithm2} presents the QAOA-based quantum optimization subroutine, which is employed as a solver within Algorithm 3. It efficiently addresses channel selection and power allocation sub-problems using PQC.} 

Moreover, we solve the original \textit{Problem 1} in an iterative fashion employing the BCD technique, as outlined in Algorithm \eqref{alg:bcd}. At each iteration, the algorithm updates the variables $\zeta^{(b)}$ and $x^{(b)}$ \textcolor{black}{using QAOA-based Algorithm \ref{algorithm2}}. The process continues until convergence to an optimal solution with minimized objective value. \textit{We will also validate the \textbf{computational complexity} of our QAOA-based joint optimization algorithm through numerical results presented in \ref{computationoverhead}.}

\begin{algorithm}
\footnotesize
\caption{QAOA-based Algorithm for Solving \textit{Sub-problem 1} and \textit{Sub-problem 2}} \label{alg:qaoa}
\begin{algorithmic}[1]
\State \textbf{Input:} QUBO matrix $Q$, number of qubits, number of QAOA layers $p$, initial parameters $(\beta_0, \gamma_0)$, quantum backend, other network parameters.

\State Design a parameterized ansatz circuit for QAOA with $p$ layers, initialized with parameter vector $(\beta_0, \gamma_0)$.
\State Initialize a classical optimizer, specifically the gradient descent optimizer, for iterative parameter adjustment.

\Function{Obj}{$\beta$, $\gamma$}
\begin{itemize}
    \item Bind the parameters $(\beta, \gamma)$ to the ansatz circuit.
    \item Execute the circuit on the quantum backend to compute the expectation value $\langle \beta, \gamma | H_C | \beta, \gamma \rangle$.
    \item \Return the measured expectation value as the objective for optimization.
\end{itemize}
\EndFunction

\State Optimize ansatz parameters:
\Repeat
\begin{itemize}
    \item Update $(\beta, \gamma)$ using the gradient descent optimizer on \Call{Obj}{$\beta$, $\gamma$}.
    \item Track the objective function values to observe convergence.
\end{itemize}
\Until{convergence.}
\State \textbf{Output:} Optimized binary variable matrix (channel selection matrix $\zeta^{(b)}$ for \textit{Sub-problem 1} or transmit power matrix $x^{(b)}$ for \textit{Sub-problem 2}), minimized objective value.
\end{algorithmic} \label{algorithm2}
\end{algorithm}

\begin{algorithm}
\footnotesize
\caption{BCD-based Joint Optimization Algorithm} \label{alg:bcd}
\begin{algorithmic}[1]
\State \textbf{Input:} Set the iteration index $i=0$, take necessary inputs required to solve \textit{Sub-problem 1} and \textit{Sub-problem 2} i.e., respective QUBO matrix $Q$, number of qubits, number of QAOA layers $p$, initial parameters $(\beta_0, \gamma_0)$, other network parameters, and quantum backend;
\Repeat
    \Statex Set $i \gets i+1$
    \Statex Solve \textit{Sub-problem 1} to update $\zeta^{b}$ \textcolor{black}{using Algorithm \ref{algorithm2}};
    \Statex Solve \textit{Sub-problem 2} to update $x^{(b)}$ \textcolor{black}{using Algorithm \ref{algorithm2}};
\Until{convergence.}
\State \textbf{Output:} Optimized parameters $({\zeta^{b}}^*, {x^{(b)}}^*)$, minimized objective value.
\end{algorithmic} \label{algorithm3}
\end{algorithm}

\section{Simulation Results and Evaluations} \label{simulationresultsandevaluations}
\subsection{Parameter Setting for QFL Framework}
All simulations to evaluate the performance of our proposed QFL framework were performed on a server configured with an NVIDIA GeForce RTX $4090$ GPU, $64$GB of RAM, and Ubuntu $22.04$ as the operating system. We conduct experiments on the standard MNIST dataset for handwritten digit recognition, comprising $60,000$ training examples and $10,000$ testing examples. \textcolor{black}{It is a publicly available and widely used benchmark in both FL and quantum machine learning research, making it suitable for evaluating convergence behavior under various data heterogeneity settings.} In our simulations, we investigate the convergence properties of QFL using a \textit{TorchQuantum}-based quantum model, which integrates a PQC with standard loss optimization method SGD across 4 qubits. The model's architecture features randomized quantum gates and trainable rotation layers. In our experiments, we explore full device participation scenario within our federated setup. Training utilizes the Adam optimizer at a learning rate of $0.001$. For full device participation, we engage 30 devices. The QFL framework is tested under both IID and non-IID data distributions. Moreover, we conduct experiments using different numbers of quantum measurement shots-1, 40, and 100 on the MNIST dataset to evaluate their effects on the convergence of QFL, specifically examining how quantum shot noise impacts performance metrics such as average loss and accuracy.

\subsection{Parameter Setting for QAOA-based Sum-rate Maximization}
All simulations for the QAOA-based sum-rate maximization were also conducted on the same server equipped with an NVIDIA GeForce RTX $4090$ GPU, $64$GB of RAM, and running Ubuntu $22.04$ as the operating system. \textcolor{black}{All wireless communication parameters used in the QAOA-based simulations are generated in accordance with the comprehensive system model described in Section~\ref{section:systemmodel}. Specifically, the wireless environment is modeled using distance-dependent path loss, 8 dB log-normal shadowing, and temporally correlated Rayleigh fading governed by Jake’s model, as formalized in equations~\eqref{eqn14'} and~\eqref{eqn15'}. The channel allocation and interference scenarios are simulated using a multi-channel NOMA uplink structure, where multiple quantum devices may concurrently access the same channel, leading to co-channel interference. The resulting SINR and achievable data rate for each device are computed based on equations~\eqref{eqn17'} and~\eqref{eqn18'}, ensuring consistency between the theoretical analysis and the simulated environment.} For our simulation setup, we model a system with $C = 4$ channels and the number of devices $N = \{50, 200, 500$\}. The maximum transmit power $P^{\max}$ and noise power $\sigma$ are set to $20$ dBm and $-114$  dBm, respectively \cite{tan2020deep}. The channel correlation factor $\epsilon$ is considered to be $0.6$ \cite{tan2020deep}. Regarding the design of the ansatz, it comprises $p = 2$ layers alternating between problem and mixing Hamiltonians. The problem Hamiltonian involves controlled $RZ$ rotations driven by the $Q$ matrix elements, while the mixing Hamiltonian utilizes $RX$ gates to explore the solution space. Initially, all qubits are prepared in an equal superposition state using Hadamard gates to ensure a uniform exploration starting point. The circuit is parameterized by angles $\beta$ and $\gamma$ across two layers, and the optimization process iteratively adjusts these parameters to minimize the objective function. The simulations are executed in a Python environment, leveraging \textit{PennyLane} for quantum circuit simulation and optimization, alongside \textit{NumPy} for numerical computations and \textit{matplotlib} for visualizing the results. 

\subsection{Numerical Results for QFL Framework Performance Evaluation}
\begin{figure}[h!]
    \centering
    \footnotesize
    \begin{subfigure}[t]{0.49\linewidth} 
        \centering
        \includegraphics[width=\linewidth]{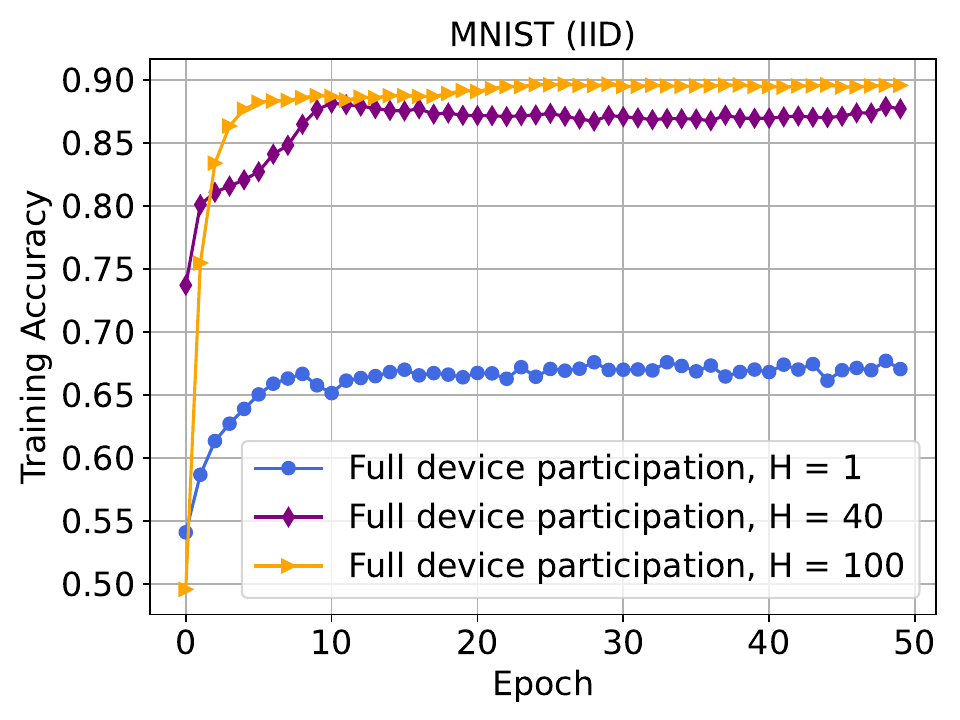}
        \caption{\footnotesize Training accuracy comparison of QFL on MNIST dataset (IID) for full device participation with varying number of quantum measurement shots.}
        \label{fig81}
    \end{subfigure}
    \hfill 
    \begin{subfigure}[t]{0.49\linewidth} 
        \centering
        \includegraphics[width=\linewidth]{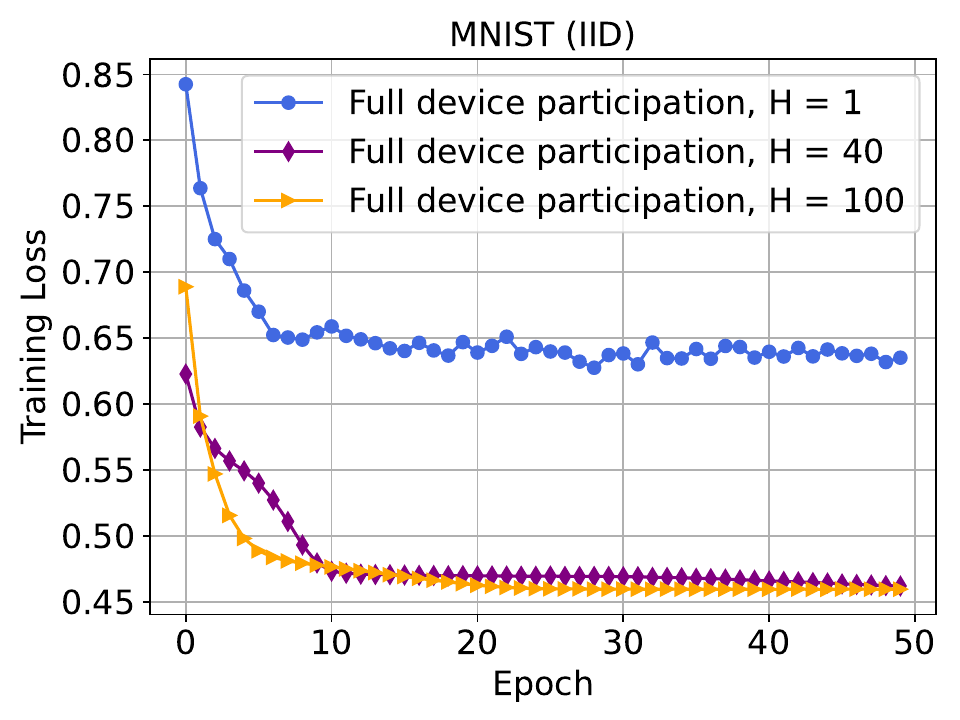}
        \caption{\footnotesize Training loss comparison of QFL on MNIST dataset (IID) for full device participation with varying number of quantum measurement shots.}
        \label{fig82}
    \end{subfigure}
    \caption{\footnotesize Training performance comparison of QFL for MNIST dataset (IID) for full device participation with varying number of quantum measurement shots.}
    \label{fig80}
\end{figure}

\begin{figure}[h!]
    \centering
    \footnotesize
    \begin{subfigure}[t]{0.49\linewidth} 
        \centering
        \includegraphics[width=\linewidth]{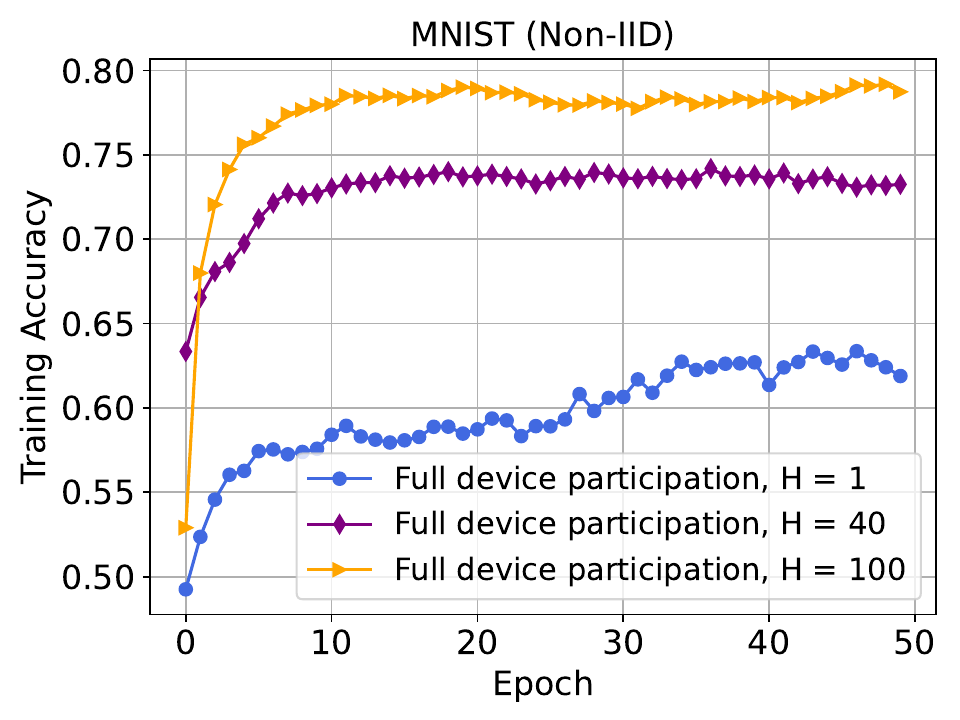}
        \caption{\footnotesize Training accuracy comparison of QFL on MNIST dataset (non-IID) for full device participation with varying number of quantum measurement shots.}
        \label{fig91}
    \end{subfigure}
    \hfill 
    \begin{subfigure}[t]{0.49\linewidth} 
        \centering
        \includegraphics[width=\linewidth]{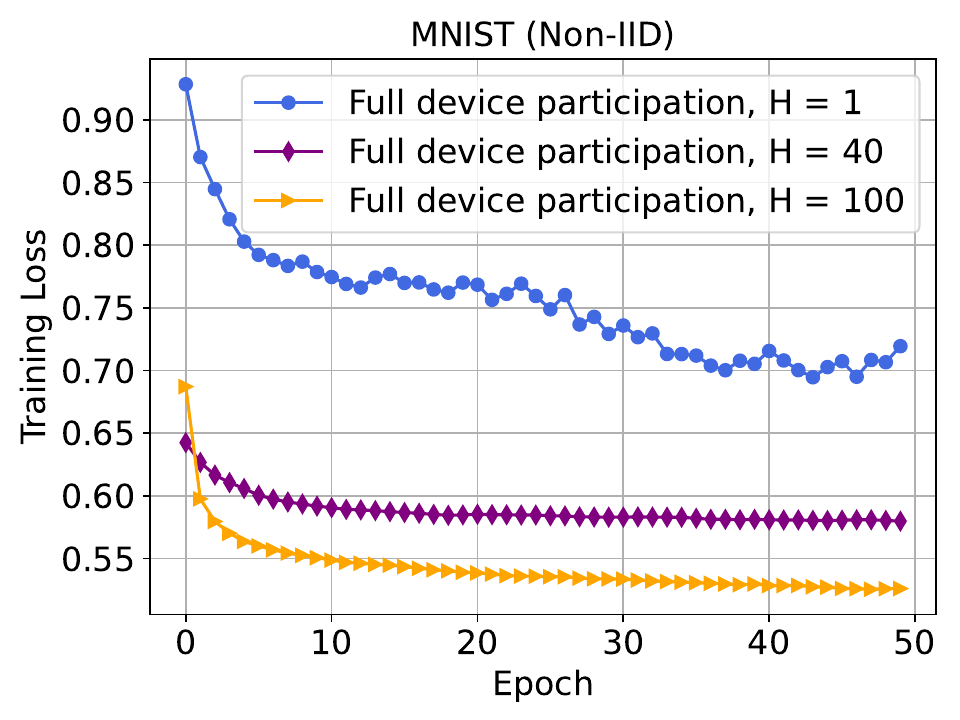}
        \caption{\footnotesize Training loss comparison of QFL on MNIST dataset (non-IID) for full device participation with varying number of quantum measurement shots.}
        \label{fig92}
    \end{subfigure}
    \caption{\footnotesize Training performance comparison of QFL on Cifar10 and MNIST datasets (non-IID) for full device participation with varying number of quantum measurement shots.}
    \label{fig90}
\end{figure}

Fig.~\ref{fig80} illustrates the training performance in QFL on MNIST dataset (IID), showcasing improved accuracy with increasing number of quantum measurement shots under full device participation scenario. Fig.~\ref{fig81} shows the training accuracy for QFL considering full device participation with varying levels of quantum measurement shots ($H$=1, $H$=40, $H$=100). From the plot, it is evident that as the number of quantum measurement shots increases, there is a significant enhancement in training accuracy. This improvement is attributed to the reduced quantum shot noise effect achieved through the aggregation of more measurement outcomes, which enhances the reliability and stability of the quantum computations related to model training. Starting from a single measurement shot ($H$=1) to a higher count ($H$=40), and subsequently increasing further ($H$=100), consistently improves the QFL's accuracy. This result indicates that increasing the number of quantum measurement shots effectively reduces the instability and errors introduced by quantum shot noise, thus substantially enhancing the overall training performance in QFL systems. The loss performance of QFL in Fig.~\ref{fig82} is in harmony with the accuracy performance of it in Fig.~\ref{fig81}. Both figures demonstrate that increasing the number of quantum measurement shots ($H$=1, $H$=40, $H$=100) results in significant enhancements in performance. Increasing the number of measurement shots from $H$=1 to $H$=40, and subsequently to $H$=100, effectively minimizes the errors and variability inherent in quantum measurement, thereby leading to improved loss metrics. This trend shades light on the crucial role of scaling up the number of quantum measurement shots to achieve optimal training outcomes in QFL setups.

Fig.~\ref{fig90} presents training performance of QFL under the non-IID setting on MNIST datasets considering full device participation and varying numbers of quantum measurement shots. The performance trend is similar to the one observed in the IID setting in Fig.~\ref{fig80}. This consistency highlights that, regardless of data distribution differences between IID and non-IID conditions, the influence of increasing quantum measurement shots remains unchanged, systematically enhancing boosting training accuracy and diminishing loss across both data configurations.

\subsection{Numerical Results for QAOA-based Sum-rate Maximization Evaluation}
\begin{figure}[h!]
    \centering
    \footnotesize
    \begin{subfigure}[t]{0.49\linewidth}
        \centering
        \includegraphics[width=\linewidth]{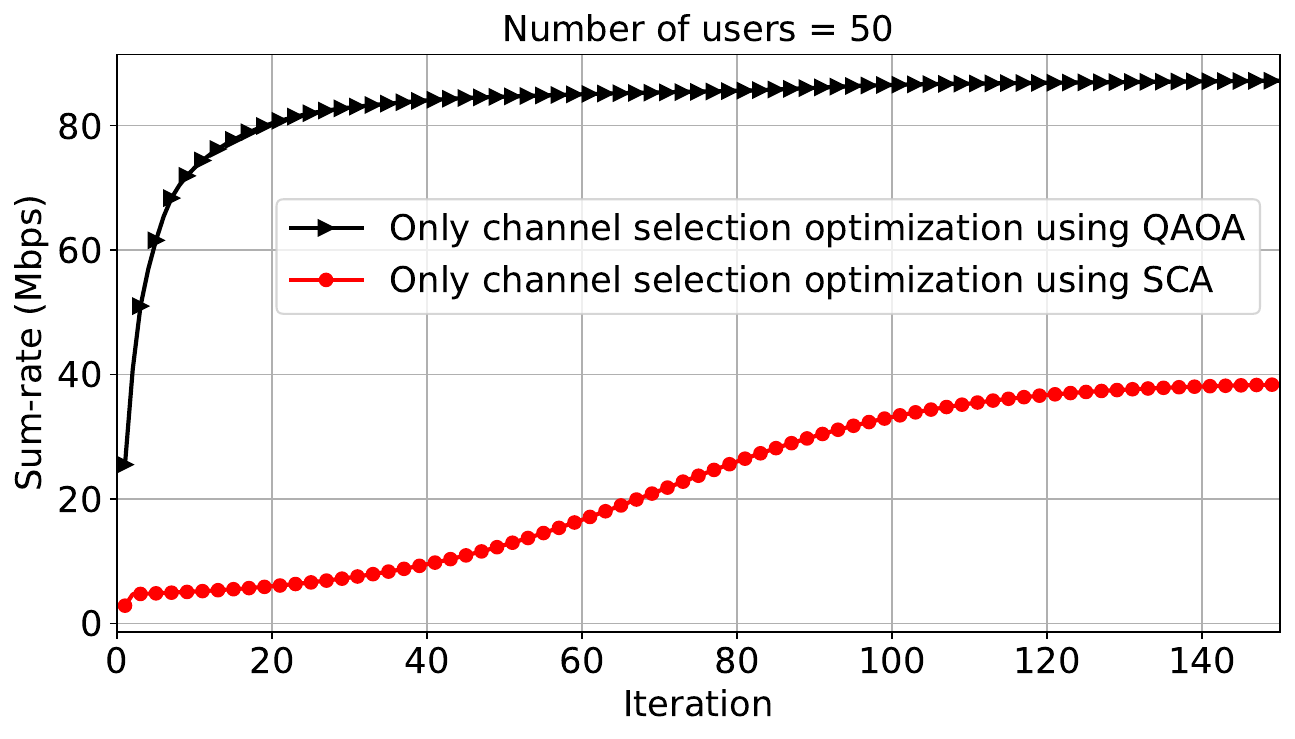}
        \caption{\footnotesize Comparison of channel selection optimization using QAOA and SCA for $N = 50$.}
        \label{fig5}
    \end{subfigure}
    \hfill
    \begin{subfigure}[t]{0.49\linewidth}
        \centering
        \includegraphics[width=\linewidth]{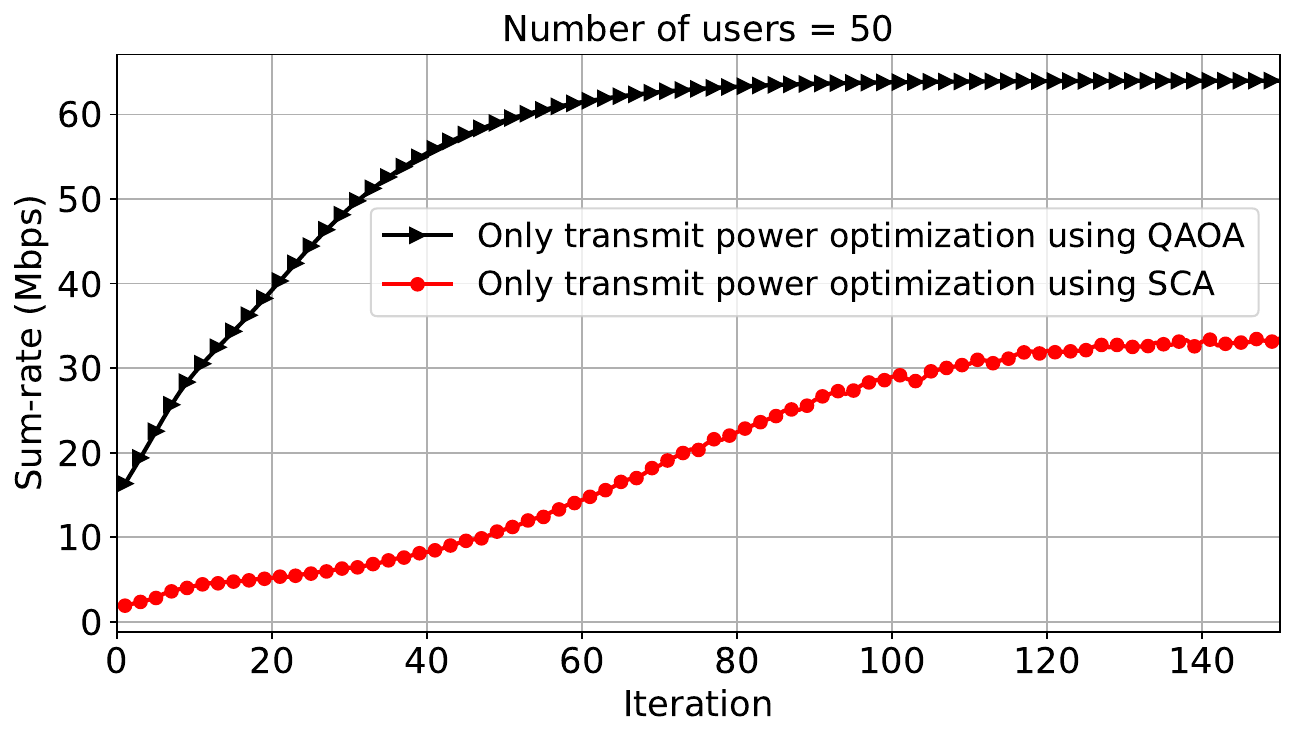}
        \caption{\footnotesize Comparison of transmit power optimization using QAOA and SCA for $N = 50$.}
        \label{fig4}
    \end{subfigure}
    \begin{subfigure}[t]{0.49\linewidth}
        \centering
        \includegraphics[width=\linewidth]{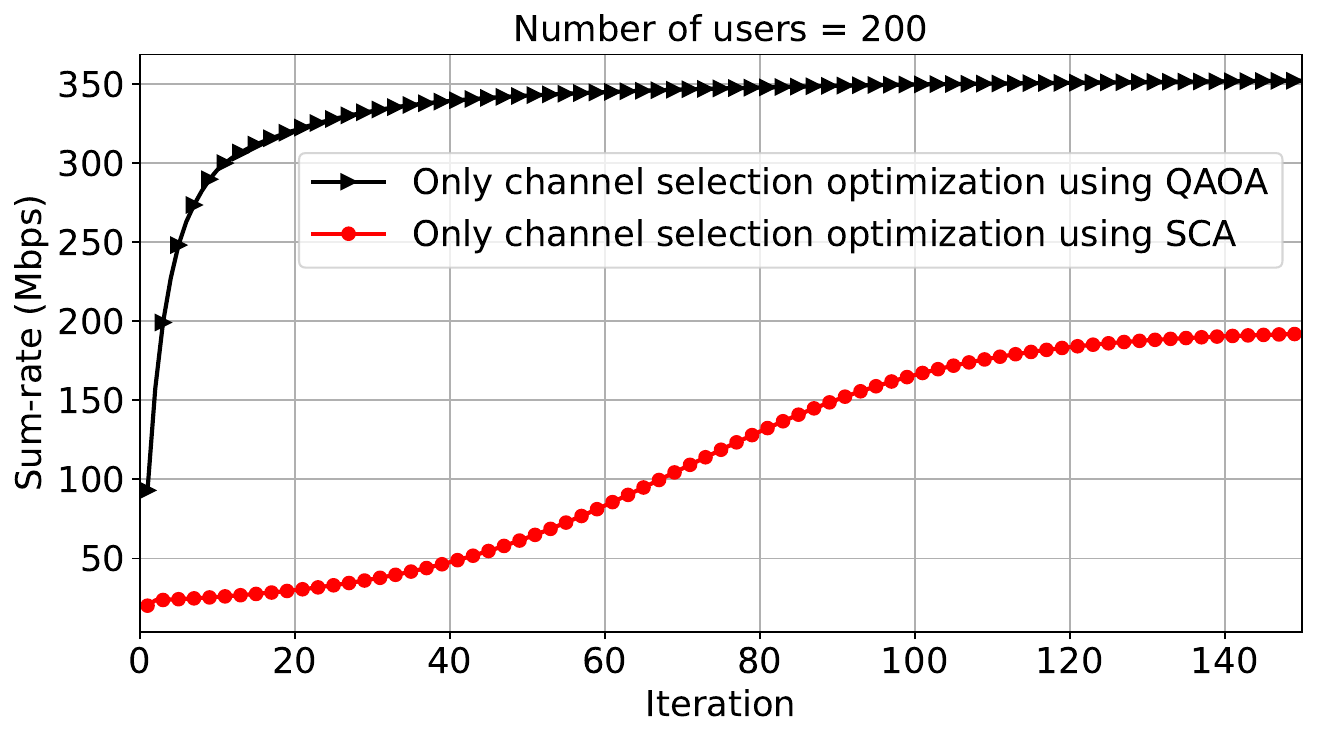}
        \caption{\footnotesize Comparison of channel selection optimization using QAOA and SCA for $N = 200$.}
        \label{fig6}
    \end{subfigure}
    \hfill
    \begin{subfigure}[t]{0.49\linewidth}
        \centering
        \includegraphics[width=\linewidth]{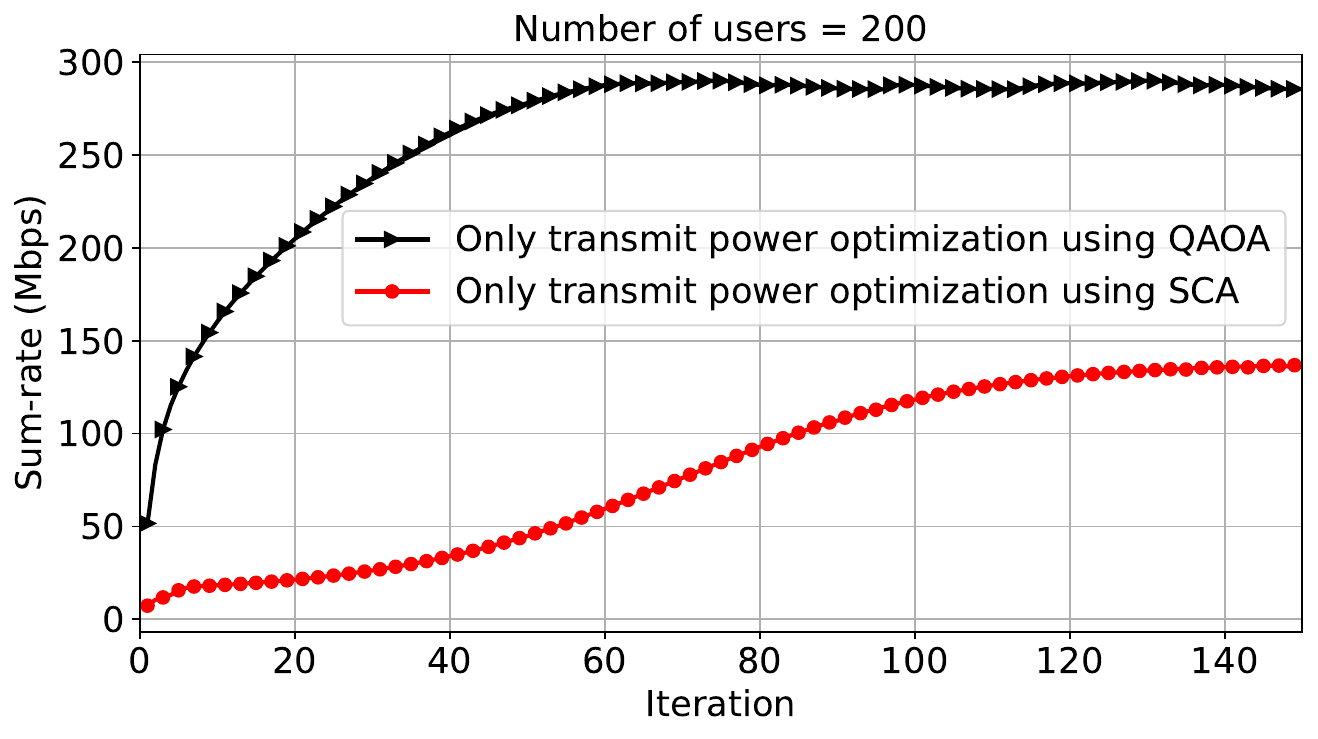}
        \caption{\footnotesize Comparison of transmit power optimization using QAOA and SCA for $N = 200$.}
        \label{fig7}
    \end{subfigure}
    \caption{\footnotesize Comparison of transmit power and channel selection optimization using QAOA and SCA for varying device numbers.}
    \label{figcomp1}
\end{figure}

\begin{figure*}[ht!]
    \centering
    \begin{subfigure}{0.32\textwidth}
        \centering
        \includegraphics[width=\linewidth]{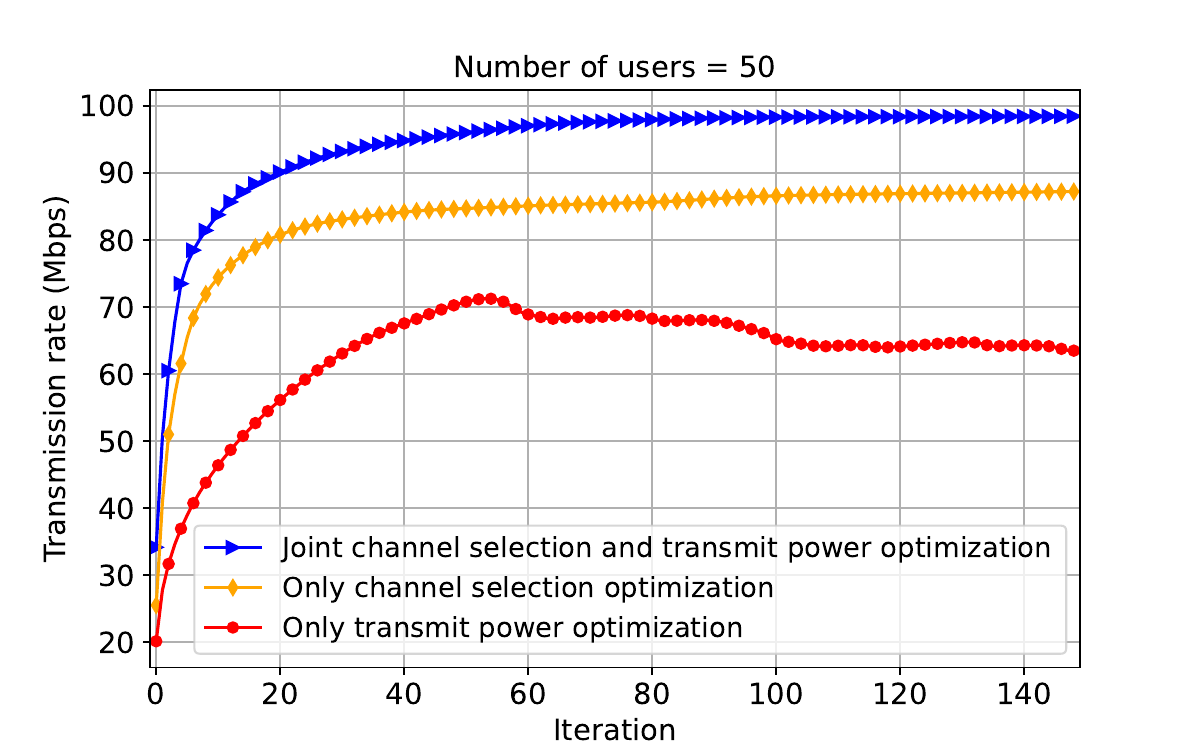} 
        \caption{\footnotesize Comparative analysis of sum-rate performance across optimization schemes for number of devices $N = 50$.}
        \label{fig:sub1}
    \end{subfigure}
    \hfill 
    \begin{subfigure}{0.32\textwidth}
        \centering
        \includegraphics[width=\linewidth]{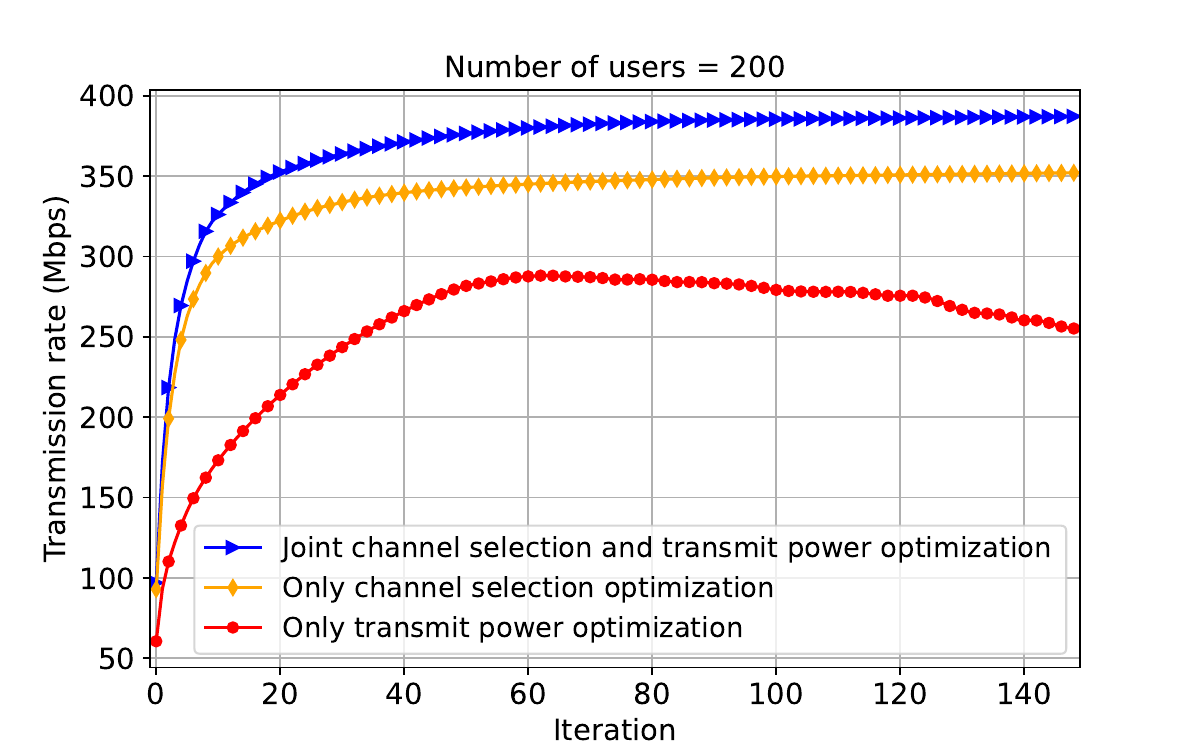} 
        \caption{\footnotesize Comparative analysis of sum-rate performance across optimization schemes for number of devices $N = 200$.}
        \label{fig:sub2}
    \end{subfigure}
    \hfill 
    \begin{subfigure}{0.32\textwidth}
        \centering
        \includegraphics[width=\linewidth]{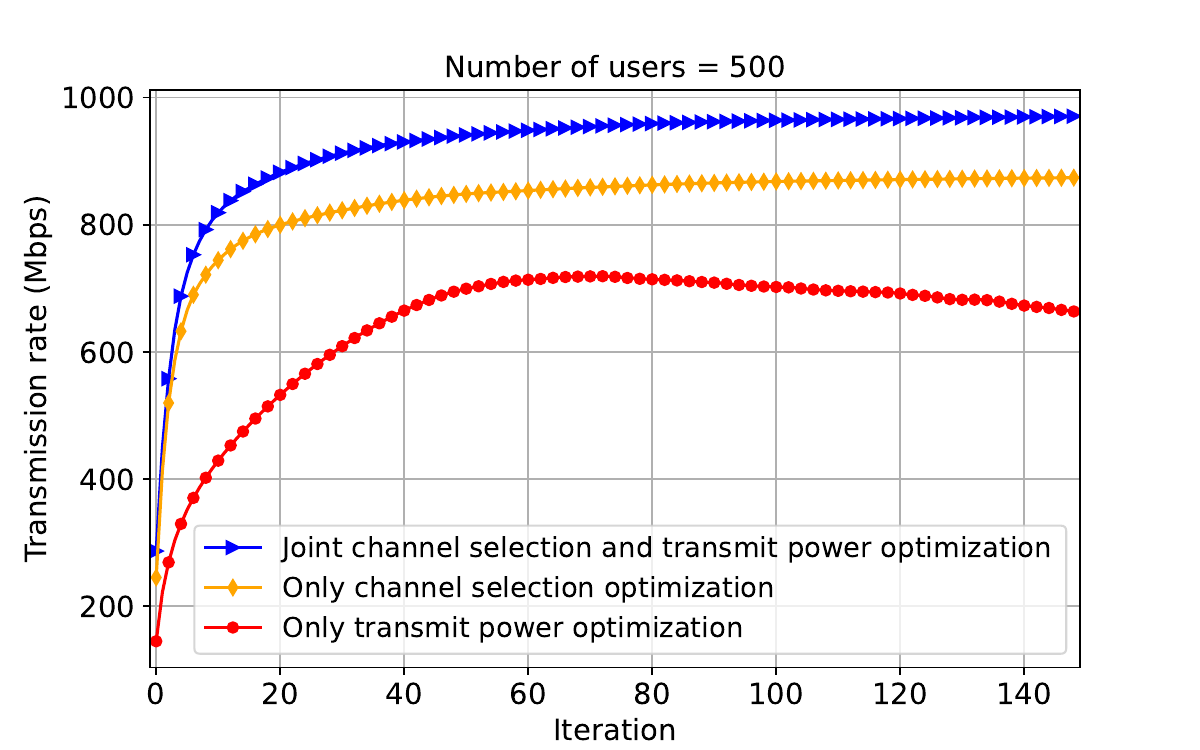} 
        \caption{\footnotesize Comparative analysis of sum-rate performance across optimization schemes for number of devices $N = 500$.}
        \label{fig:sub3}
    \end{subfigure}
    \caption{\footnotesize Comparative analysis of sum-rate performance across optimization schemes for varying number of devices.}
    \label{figcomp2}
\end{figure*}

\begin{figure}[h!]
    \centering
    \footnotesize
    \begin{subfigure}[t]{0.49\linewidth} 
        \centering
        \includegraphics[width=\linewidth]{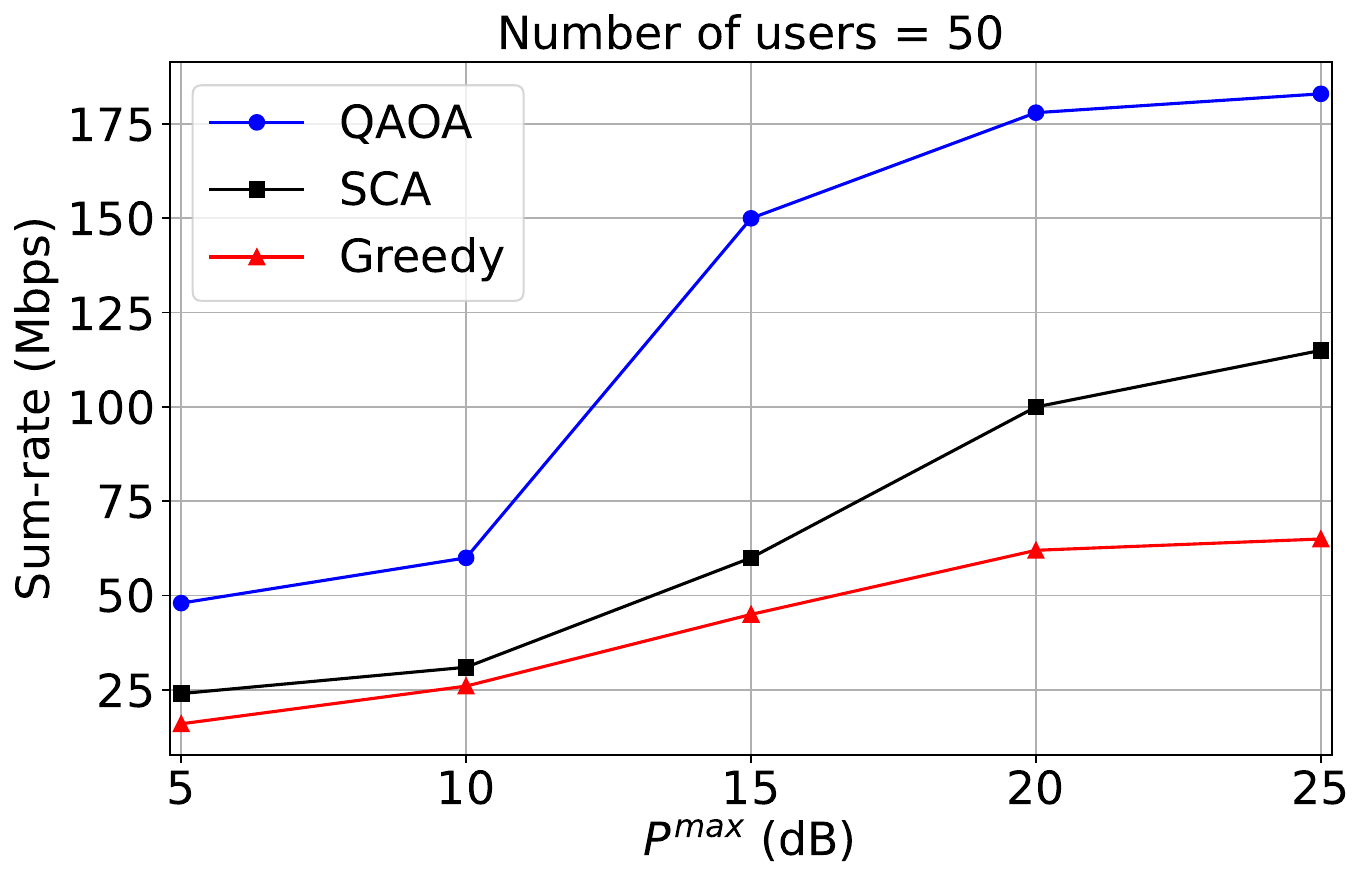}
        \caption{\footnotesize Comparison of sum-rate performance across different schemes as device's maximum transmit power increases for number of devices, $N = 50$.}
        \label{sub11}
    \end{subfigure}
    \hfill 
    \begin{subfigure}[t]{0.49\linewidth} 
        \centering
        \includegraphics[width=\linewidth]{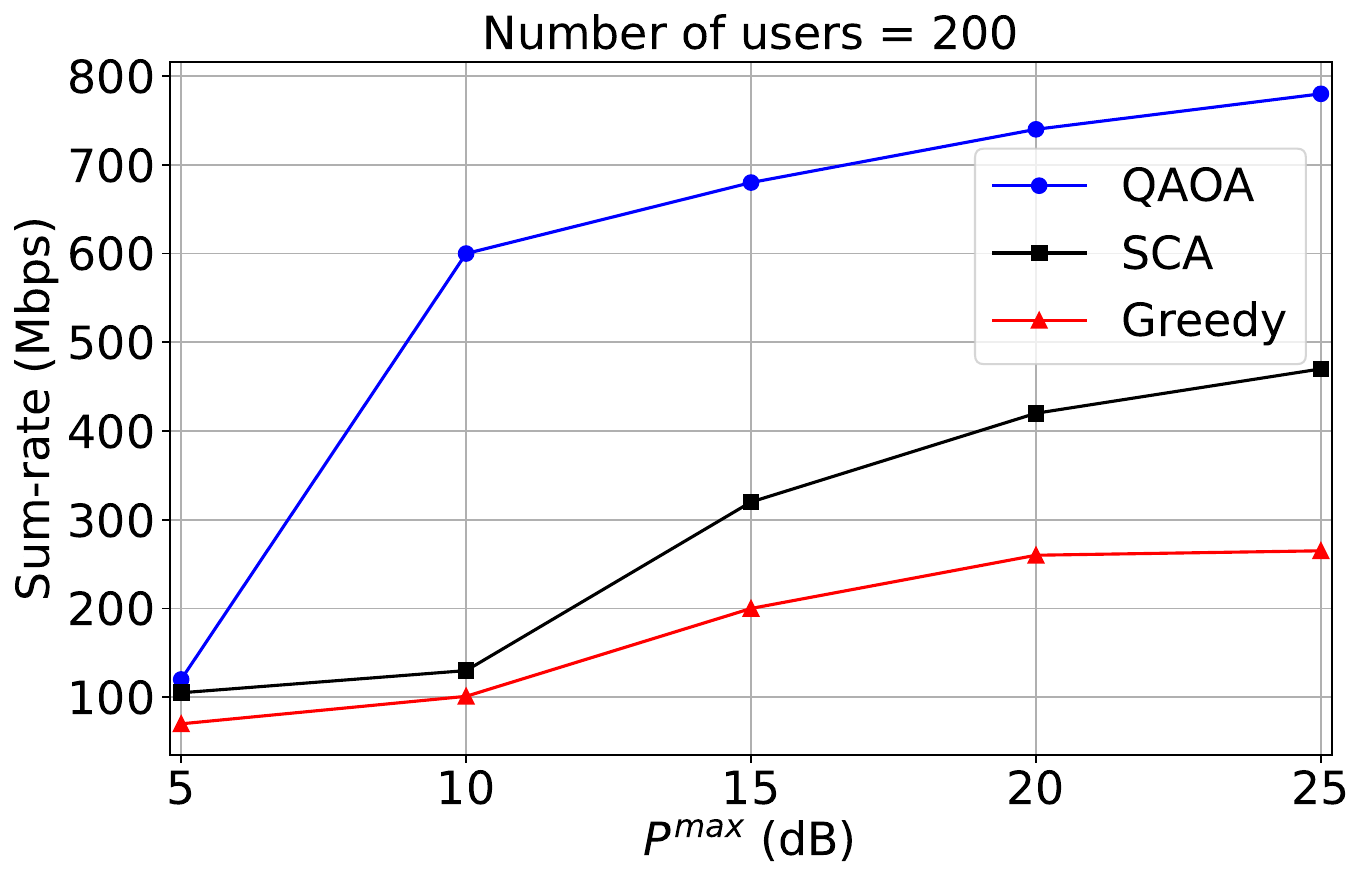}
        \caption{\footnotesize Comparison of sum-rate performance across different schemes as device's maximum transmit power increases for number of devices, $N = 200$.}
        \label{sub12}
    \end{subfigure}
    \caption{\footnotesize Comparison of sum-rate performance across different schemes as devices maximum transmit power increases for varying number of devices.}
    \label{figcomp3}
\end{figure}

Fig.~\ref{figcomp1} shows a comparative analysis of the performance of channel selection and transmit power optimization techniques, QAOA and SCA, over 150 iterations. Simulations are conducted for both  optimization techniques as the number of devices increases from $N = 50$ and $N = 200$ showcasing the scalability and effectiveness of each method under varying network loads. Fig.~\ref{fig5} depicts sum-rate (in Mbps) versus number of iterations, comparing the QAOA with the SCA technique for $N = 200$ devices in channel selection optimization. From the plot, it is evident that the QAOA significantly outperforms SCA in terms of both the sum-rate achieved at convergence and the efficiency of convergence. While QAOA quickly reaches and sustains a sum rate above $80$ Mbps by around the $30^{\text{th}}$ iteration, SCA shows a more gradual increase and levels off at approximately $40$ Mbps by around the $130^{\text{th}}$. We can say that QAOA outperforms SCA in optimizing channel selection, achieving almost $100\%$ increase in sum-rate.

Fig.~\ref{fig4} displays the sum-rate (in Mbps) against the number of iterations, contrasting the performance of the QAOA with the SCA technique for $N = 50$ devices in transmit power optimization. The graph shows that the QAOA surpasses SCA both in the peak sum-rate reached at convergence and in the number of iterations needed reach convergence. Specifically, QAOA achieves and maintains a sum-rate close to $60$ Mbps by approximately the $30^{\text{th}}$ iteration, while SCA shows a slower increase, plateauing around $30$ Mbps by around the $130^{\text{th}}$. Thus, similar to channel selection optimization, QAOA demonstrates superior performance over SCA in transmit power optimization as well, achieving nearly a $100\%$ increase in sum-rate. 

\begin{figure}[h!]
    \centering
    \footnotesize
    \begin{subfigure}[t]{0.49\linewidth} 
        \centering
        \includegraphics[width=\linewidth]{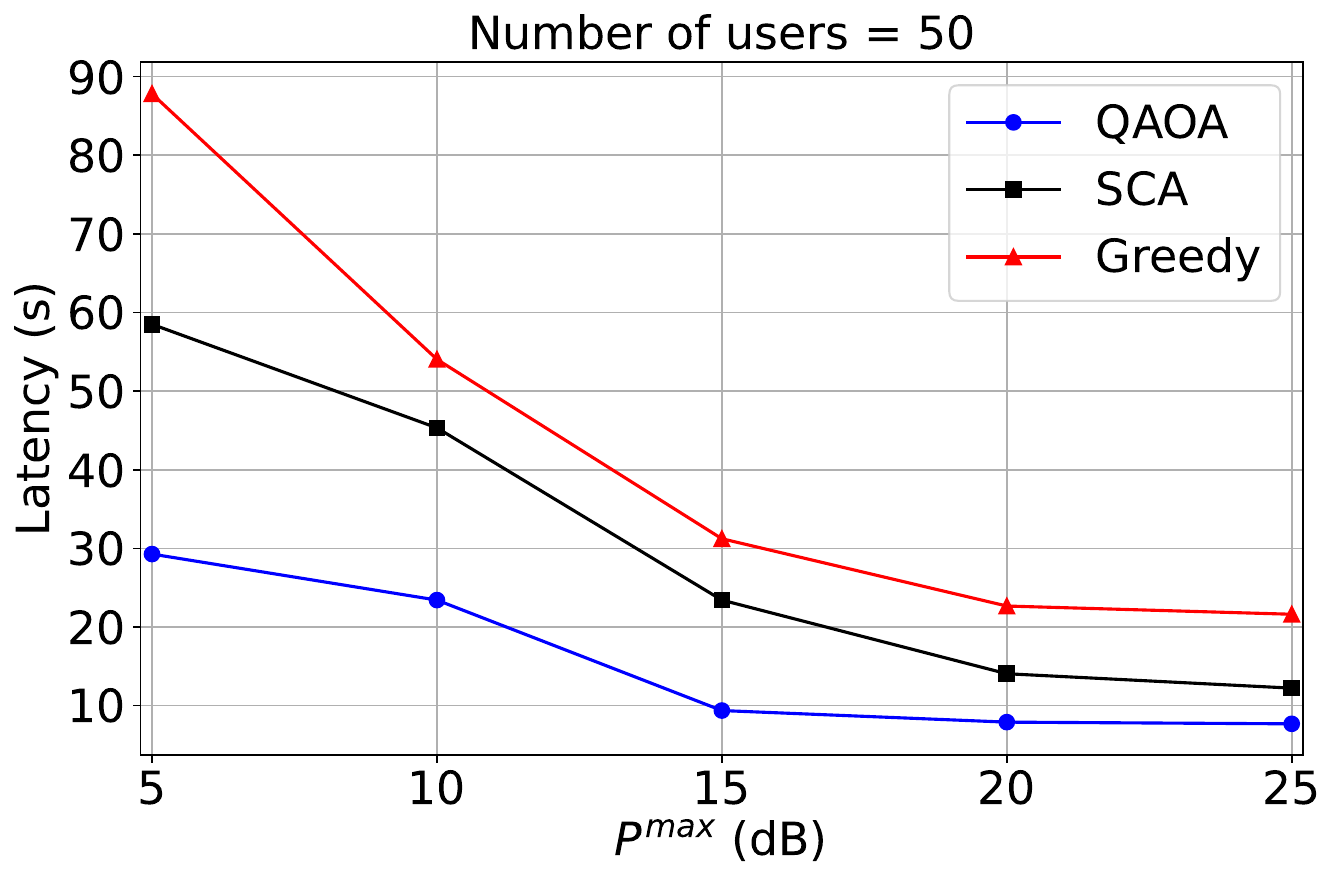}
        \caption{\footnotesize Comparison of latency performance across different schemes as device's maximum transmit power increases for number of devices, $N = 50$.}
        \label{sub21}
    \end{subfigure}
    \hfill 
    \begin{subfigure}[t]{0.49\linewidth} 
        \centering
        \includegraphics[width=\linewidth]{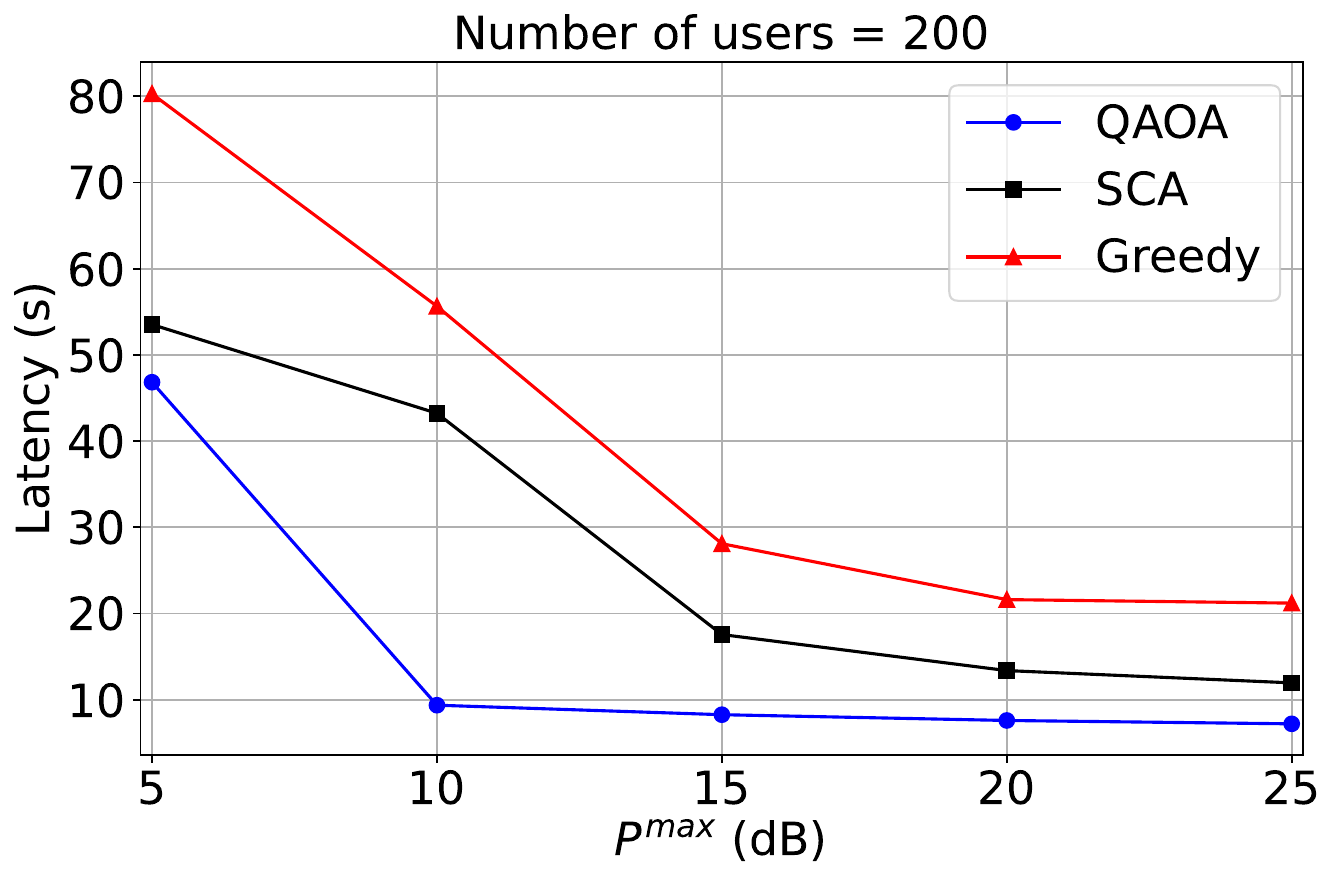}
        \caption{\footnotesize Comparison of latency performance across different schemes as device's maximum transmit power increases for number of devices, $N = 200$.}
        \label{sub22}
    \end{subfigure}
    \caption{\footnotesize Comparison of latency performance across different schemes as device's maximum transmit power increases for varying number of devices.}
    \label{figcomp5}
\end{figure}

Fig.~\ref{fig6} compares the performance of QAOA with the SCA technique in optimizing channel selection for $N = 200$ devices, illustrating the sum-rate (in Mbps) versus the number of iterations. Similar to earlier observations with fewer devices, QAOA continues to outperform SCA significantly. In this scenario, QAOA achieves a sum-rate that is approximately $75\%$ higher than that of SCA, reaching convergence significantly earlier than SCA. Likewise, Fig.~\ref{fig7} compares the performance of QAOA with the SCA technique in optimizing transmit power for $N = 200$ devices, illustrating the sum-rate (in Mbps) versus the number of iterations. From the figure, it is clear that the performance trend of Fig.~\ref{fig4} reiterates here as well, as QAOA's continues to show its dominance over SCA. While QAOA achieves an almost $100\%$ increase in sum-rate, the plot also highlights its ability to reach convergence significantly faster than SCA.

\begin{figure*}[ht!]
    \centering
    \begin{subfigure}{0.32\textwidth}
        \centering
        \includegraphics[width=\linewidth]{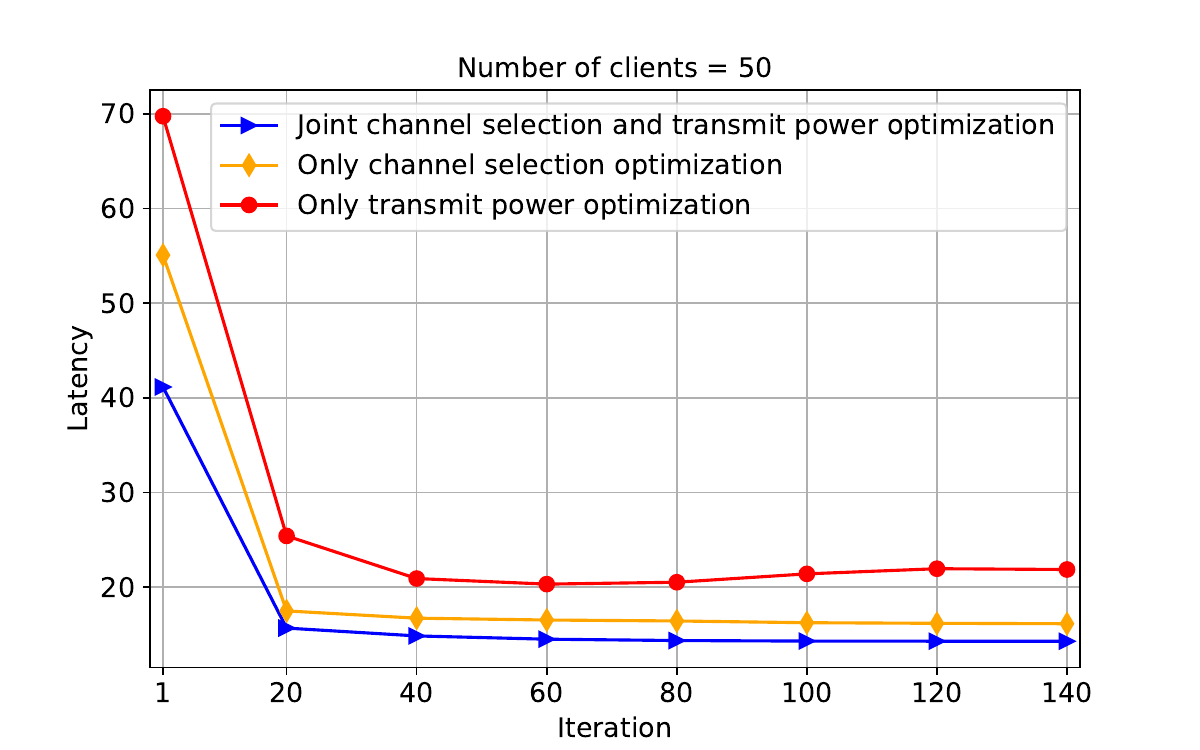} 
        \caption{\footnotesize Comparative analysis of latency performance across different optimization schemes for number of devices $N = 50$.}
        \label{fig:latsub1}
    \end{subfigure}
    \hfill 
    \begin{subfigure}{0.32\textwidth}
        \centering
        \includegraphics[width=\linewidth]{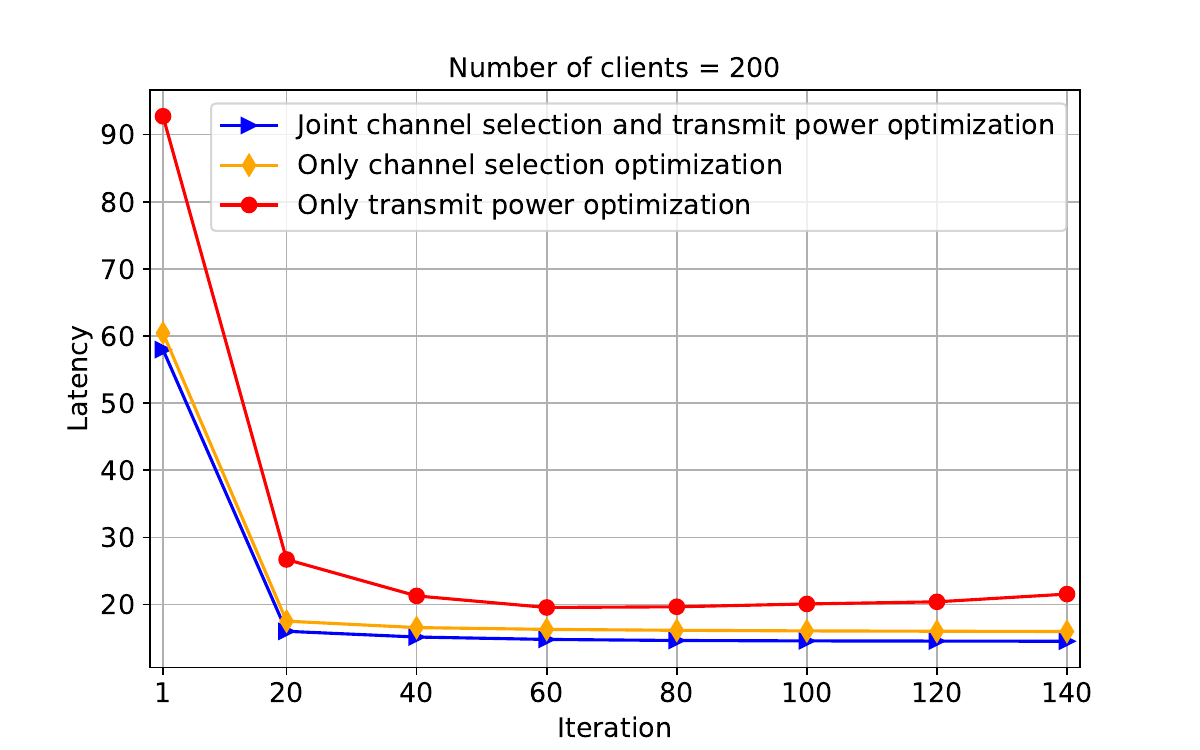} 
        \caption{\footnotesize Comparative analysis of latency performance across different optimization schemes for number of devices $N = 200$.}
        \label{fig:latsub2}
    \end{subfigure}
    \hfill 
    \begin{subfigure}{0.32\textwidth}
        \centering
        \includegraphics[width=\linewidth]{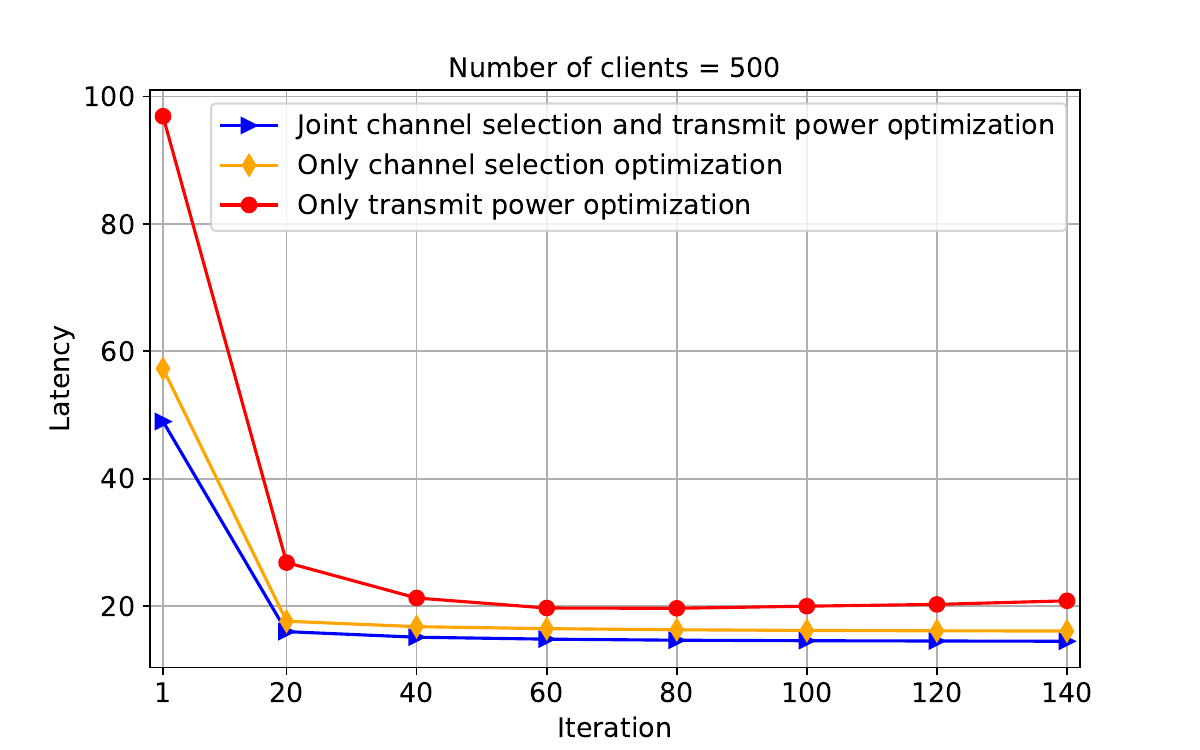} 
        \caption{\footnotesize Comparative analysis of latency performance across different optimization schemes for number of devices $N = 500$.}
        \label{fig:latsub3}
    \end{subfigure}
    \caption{\footnotesize Comparative analysis of latency performance across optimization schemes for varying number of devices.}
    \label{figcomp4}
\end{figure*}

Fig.~\ref{figcomp2} illustrates a comparative analysis of sum-rate performance across different optimization schemes using QAOA for varying number of devices. Fig.~\ref{fig:sub1} shows the sum-rate (in Mbps) versus the number of iterations for $N = 50$ devices, comparing our proposed scheme that jointly optimizes channel selection and transmit power with schemes that optimize only one of these parameters. From the graph, it is clear that our proposed scheme attains a consistent level of sum-rate after around the $40^{\text{th}}$ iteration, significantly outperforming the other two schemes in terms of both maximizing the sum-rate and achieving quicker convergence. Numerically, our proposed scheme achieves sum-rate that is $12.87\%$ and $54.87\%$ higher than schemes optimizing only channel selection and transmit power of the devices, respectively.  Fig.~\ref{fig:sub2} and Fig.~\ref{fig:sub3} extend this analysis to larger networks of $N = 200$ and $N = 500$ devices, respectively, and show a similar trend. Specifically, in Fig.~\ref{fig:sub2}, our proposed joint optimization scheme demonstrates superior performance than other schemes, with sum-rate that is $9.99\%$ and $52.03\%$ higher than those achieved by schemes focusing solely on optimizing channel selection or transmit power, respectively. In Fig.~\ref{fig:sub3} for $N = 500$ devices, our proposed scheme further demonstrates its robustness and efficiency, achieving sum-rate that is $11.04\%$ and $46.55\%$ higher than those of schemes focusing solely on channel selection and transmit power, respectively. These results highlight the strong scalability and effectiveness of our joint optimization approach even in substantially larger network configurations.


Fig.~\ref{figcomp3} shows the sum-rate (in Mbps) as a function of the maximum transmit power (in dB) of the devices, comparing the performance of the QAOA, SCA, and greedy schemes. Fig.~\ref{sub11} presents this comparison for $N = 50$ devices. The plot clearly shows that all optimization techniques exhibit an increase in sum-rate as the maximum transmit power of the devices increases. However, our proposed quantum-based optimization technique significantly outperforms the other algorithms, achieving notably higher sum-rate than both SCA and Greedy techniques. Numerically, QAOA attains $59.13\%$ and $181.54\%$ higher sum-rate than SCA and Greedy, respectively. Fig.~\ref{sub12} extends the analysis to $N = 200$ devices, demonstrating a similar trend. In this larger device scenario, QAOA continues to surpass the other optimization algorithms, achieving $65.96\%$ and $194.34\%$ higher sum-rate than SCA and Greedy, respectively. 

Similarly, our latency evaluation, as depicted in Fig.~\ref{figcomp5}, also demonstrates the efficacy of our QAOA-based optimization approach across larger networks of $N = 50$ and $N = 200$ devices with increasing maximum transmit power of devices. For instance, at $N = 50$ in Fig.~\ref{sub21}, our proposed QAOA scheme reduces latency by $38.46\%$ and $63.64\%$ compared to SCA and greedy schemes, respectively. Likewise, for network of $N = 200$ devices in Fig.~\ref{sub22}, our approach results in $41.67\%$ and $66.67\%$ lower latency.


Moreover, we present a comparative analysis of latency performance across different optimization schemes using QAOA for varying number of devices in Fig.~\ref{figcomp4} . Fig.~\ref{fig:latsub1} portrays the latency (in s) versus the number of iterations for $N = 50$ devices, comparing our proposed scheme that jointly optimizes channel selection and transmit power with schemes that optimize only one of these parameters. Clearly, our proposed scheme achieves a consistent level of latency after around the $20^{\text{th}}$ iteration, significantly surpassing the other two schemes in terms of both minimizing the latency and achieving rapid convergence. More specifically, our proposed scheme attains latency that is $11.53\%$ and $34.69\%$ lower than schemes optimizing only channel selection and transmit power of the devices, respectively. 


Fig.~\ref{fig:latsub2} and Fig.~\ref{fig:latsub3} extend this analysis to larger networks of $N = 200$ and $N = 500$ devices, respectively, demonstrating consistent trends across these larger scales. Numerically, in Fig.~\ref{fig:latsub2}, our proposed joint optimization approach of channel selection and transmit power shows superior performance than other schemes, achieving latency that is $9.07\%$ and $32.61\%$ lower than those achieved by schemes focusing solely on optimizing channel selection or transmit power, respectively. In Fig.~\ref{fig:latsub3} for $N = 500$ devices, our proposed scheme further demonstrates its robustness and efficiency, achieving latency that is $9.94\%$ and $30.47\%$ lower than those of schemes focusing solely on channel selection and transmit power, respectively.


\subsection{Comparison of Computation Overhead} \label{computationoverhead}
We also investigate and compare the code execution time for the QAOA and the benchmark, i.e., the SCA algorithm, when applied to joint optimization tasks under varying numbers of devices. Table \ref{table1} presents these values, clearly illustrating that QAOA, leveraging quantum computing capabilities, executes significantly faster than SCA across different device counts. This highlights the efficiency of quantum-based methods, i.e., QAOA, in optimizing complex network operations. \textcolor{black}{While traditional algorithms such as SCA is widely used for wireless resource allocation, QAOA offers substantial speedup, over 50 times in some cases, making it highly practical for large-scale QFL systems where resource coordination must be completed within realistic round durations.}
\begin{table}[h]
\centering
\caption{Execution Times of QAOA and SCA}
\label{table1}
\begin{tabular}{|c|c|c|}
\hline
\textbf{Number of devices} & \textbf{QAOA Execution} & \textbf{SCA Execution} \\
                         & \textbf{Time (s)}       & \textbf{Time (s)}      \\ \hline
50                       & \textcolor{black}{187.33}                      & 2245.08                     \\ \hline
200                      & 730.75                      & 36562.43                     \\ \hline
\end{tabular}
\end{table}

Additionally, we have documented the execution times for optimizing different parameters using QAOA over 150 iterations, with device numbers at 50, 200, and 500 in table \ref{tab:runtime_data}. This data provides a clear demonstration of how execution time scales with the complexity of the optimization task and the number of devices involved. Specifically, the joint optimization of both parameters consistently requires more time than single-parameter optimizations, reflecting its higher computational demand. However, this increase in execution time is expected to correspond to more comprehensive optimization outcomes, highlighting the trade-off between computational expense and optimization efficacy. This analysis is crucial for assessing the practical scalability of QAOA in real-world applications where performance and speed are both pivotal.

\begin{table}[h]
\centering
\caption{Execution Times for Different QAOA Optimization Strategies and device Numbers}
\label{tab:runtime_data}
\begin{tabular}{|c|c|c|c|}
\hline
\textbf{Number of devices} & \textbf{Sub-problem} & \textbf{Sub-problem} & \textbf{Joint} \\
                         & \textbf{1 (s)}       & \textbf{2 (s)}      &\textbf{(s)}\\ \hline
50 & 46.74 & 140.33 & 187.33 \\
\hline
200 & 188.53 & 540.13 & 730.75 \\
\hline
500 & 470.79 & 1365.88 & 1839.88 \\
\hline
\end{tabular}
\end{table}

\section{Conclusions}
This paper has studied an optimized QFL framework, specifically tailored for NOMA-based large-scale wireless networks, by jointly optimizing quantum device's channel selection and transmit power for  sum-rate maximization. A rigorous analysis of the convergence behavior of our proposed QFL framework has been also given. The formulated sum-rate maximization problem is non-convex, MINLP challenge that remains NP-hard even for known channel selection parameters. Hence, we have proposed an efficient quantum-centric iterative optimization approach via the QAOA and the BCD techniques to obtain optimal solutions. Extensive numerical results demonstrate that our multi-channel NOMA-based QFL framework enhances both model training  and sum-rate performance.

\bibliographystyle{IEEEtran}
\bibliography{Main-Bibliography}

\vspace{5pt}


\begin{IEEEbiographynophoto}{Shaba Shaon}  is currently pursuing her Ph.D. in the Department of Electrical and Computer Engineering at The University of Alabama in Huntsville, USA. Her research interests focus on distributed quantum computing, federated learning, wireless networking.
\end{IEEEbiographynophoto}

\begin{IEEEbiographynophoto}{Christopher G. Brinton} [Senior Member, IEEE] is currently the Elmore Associate
Professor of Electrical and Computer Engineering with Purdue University, USA. He was a recipient of the four U.S. Top Early Career Awards from the National Science Foundation (CAREER), the Office of Naval Research (YIP), the Defense Advanced Research Projects Agency (YFA), and the Air Force Office of Scientific Research (YIP).

\end{IEEEbiographynophoto}

\begin{IEEEbiographynophoto}{Dinh C. Nguyen} is an assistant professor at the Department of Electrical and Computer Engineering, The University of Alabama in Huntsville, USA.  His research interests focus on distributed quantum computing, federated learning, wireless networking.
\end{IEEEbiographynophoto}

\appendix
\section{Detailed Proofs} \label{DetailedProofs}
In this section, we present proofs of lemmas and theorems used the above section.
\subsection{Proof of Lemma 1} \label{prooflemma1}
As mentioned before, let $\mathcal{N} = \{1, 2, \dots, N\}$ denote the set of devices, and let $\Tilde{g}_{k}^{t} = \frac{1}{N} \sum_{n \in \mathcal{N}} \Tilde{g}_{n,k}^{t}$ represent the average of their local stochastic gradients at local iteration $t$ at global round $k$. We have
\begin{align}
    &- \mathbb{E}_{\{\xi_{1,k}^{t}, \dots, \xi_{n,k}^{t} | \boldsymbol{\theta}_{1,k}^{t}, \dots, \boldsymbol{\theta}_{N,k}^{t}\}} \mathbb{E}_{\{1, 2, \dots\, N \}\in \mathcal{N}} \bigg[\langle \nabla f(\bar{\boldsymbol{\theta}}_{k}^{t}), \Tilde{g}_{k}^{t} \rangle\bigg] \nonumber \\
    &= - \mathbb{E}_{\{\xi_{1,k}^{t}, \dots, \xi_{n,k}^{t} | \boldsymbol{\theta}_{1,k}^{t}, \dots, \boldsymbol{\theta}_{N,k}^{t}\}} \nonumber \\
    & \hspace{7em} \mathbb{E}_{\{1, 2, \dots\, N \}\in \mathcal{N}} \bigg[\langle \nabla f(\bar{\boldsymbol{\theta}}_{k}^{t}), \frac{1}{N}\sum_{n \in \mathcal{N}}\Tilde{g}_{n,k}^{t} \rangle\bigg] \nonumber \\
    &\stackrel{\textcircled{\footnotesize 1}}{=} - \mathbb{E}_{\{1, 2, \dots\, N \}\in \mathcal{N}} \nonumber \\
    & \hspace{2.5em} \mathbb{E}_{\{\xi_{1,k}^{t}, \dots, \xi_{n,k}^{t} | \boldsymbol{\theta}_{1,k}^{t}, \dots, \boldsymbol{\theta}_{N,k}^{t}\}}  \bigg[\langle \nabla f(\bar{\boldsymbol{\theta}}_{k}^{t}), \frac{1}{N}\sum_{n \in \mathcal{N}}\Tilde{g}_{n,k}^{t} \rangle\bigg] \nonumber \\
    &= - \bigg \langle \nabla f(\bar{\boldsymbol{\theta}}_{k}^{t}), \mathbb{E}_{\mathcal{N}} \bigg[\frac{1}{N}\sum_{n \in \mathcal{N}} \mathbb{E}_{k,t}\Tilde{g}_{n} \bigg] \bigg \rangle \nonumber \\
\end{align}
\begin{align}
    &= - \bigg \langle \nabla f(\bar{\boldsymbol{\theta}}_{k}^{t}), \mathbb{E}_{\mathcal{N}} \bigg[\frac{1}{N}\sum_{n \in \mathcal{N}} \nabla f_{n}(\boldsymbol{\theta}_{n,k}^{t}) \bigg] \bigg \rangle \nonumber \\
    &= - \bigg \langle \nabla f(\bar{\boldsymbol{\theta}}_{k}^{t}), \frac{1}{N} \bigg[N \sum_{n \in \mathcal{N}} \nabla f_{n}(\boldsymbol{\theta}_{n,k}^{t}) \bigg] \bigg \rangle \nonumber \\
    &\stackrel{\textcircled{\footnotesize 2}}{=} \frac{1}{2} \bigg[ 
    -||\nabla f(\bar{\boldsymbol{\theta}}_{k}^{t})||_{2}^{2} - || \sum_{n=0}^{N} \nabla f_{n}(\boldsymbol{\theta}_{n,k}^{t})||_{2}^{2} + ||\nabla f(\bar{\boldsymbol{\theta}}_{k}^{t}) \nonumber \\
    & \hspace{12em} - \sum_{n=0}^{N} \nabla f_{n}(\boldsymbol{\theta}_{n,k}^{t}) ||_{2}^{2}\bigg] \nonumber \\
    &= \frac{1}{2} \bigg[ 
    -||\nabla f(\bar{\boldsymbol{\theta}}_{k}^{t})||_{2}^{2} - || \sum_{n=0}^{N} \nabla f_{n}(\boldsymbol{\theta}_{n,k}^{t})||_{2}^{2} \nonumber \\
    & \hspace{6em} + || \sum_{n=0}^{N} \bigg(\nabla f_{n}(\bar{\boldsymbol{\theta}}_{k}^{t}) - \nabla f_{n}(\boldsymbol{\theta}_{n,k}^{t}) \bigg) ||_{2}^{2}\bigg] \nonumber \\
    &\stackrel{\textcircled{\footnotesize 3}}{\leq} \frac{1}{2} \bigg[ 
    -||\nabla f(\bar{\boldsymbol{\theta}}_{k}^{t})||_{2}^{2} - || \sum_{n=0}^{N} \nabla f_{n}(\boldsymbol{\theta}_{n,k}^{t})||_{2}^{2} \nonumber \\
    & \hspace{8em} + \sum_{n=0}^{N} || \nabla f_{n}(\bar{\boldsymbol{\theta}}_{k}^{t}) - \nabla f_{n}(\boldsymbol{\theta}_{n,k}^{t}) ||_{2}^{2}\bigg] \nonumber \\
    &\stackrel{\textcircled{\footnotesize 4}}{\leq} \frac{1}{2} \bigg[ 
    -||\nabla f(\bar{\boldsymbol{\theta}}_{k}^{t})||_{2}^{2} - || \sum_{n=0}^{N} \nabla f_{n}(\boldsymbol{\theta}_{n,k}^{t})||_{2}^{2} \nonumber \\
    & \hspace{12em} + \sum_{n=0}^{N} L^{2} || \bar{\boldsymbol{\theta}}_{k}^{t} - \boldsymbol{\theta}_{n,k}^{t}||_{2}^{2}\bigg],
\end{align} 

where \textcircled{\footnotesize 1} is due to the fact that random variables $\xi_{n,k}^{t}$ and $\mathcal{N}$ are independent, \textcircled{\footnotesize 1} is because \textcircled{\footnotesize 2} $2\langle a,b \rangle = ||a||^{2} + ||b||^{2} - ||a-b||^{2}$, \textcircled{\footnotesize 3} holds due to the convexity of $||.||_{2}$, and \textcircled{\footnotesize 4} follows from Assumption 1.

\subsection{Proof of Lemma 2} \label{prooflemma2}
We denote $k = i_{c}$ as the most recent global communication round, hence $\bar{\boldsymbol{\theta}}^{i_{c}+1} = \frac{1}{N} \sum_{n \in \mathcal{N}} \boldsymbol{\theta}_{n}^{i_{c + 1}}$. The local solution at device $n$ at any particular iteration $i > i_{c}$, where $i$ is assumed to represent the most recent iteration, encompassing all global and local iterations up to the current point, can be written as:
\begin{align}
    \boldsymbol{\theta}_{n,k}^{t} = \boldsymbol{\theta}_{n}^{i} &= \boldsymbol{\theta}_{n}^{i-1} - \eta_{i_{c}} \Tilde{g}_{n}^{i-1} \nonumber \\
    &\stackrel{\textcircled{\footnotesize 1}}{=} \boldsymbol{\theta}_{n}^{i-2} - [\eta_{i_{c}} \Tilde{g}_{n}^{i-2} + \eta_{i_{c}} \Tilde{g}_{n}^{i-1}] \nonumber \\
    & = \bar{\boldsymbol{\theta}}^{i_{c}+1} - \sum_{z=i_{c}+1}^{i-1} \eta_{i_{c}} \Tilde{g}_{n}^{z}, \label{eqn51}
\end{align}
where \textcircled{\footnotesize 2} follows from the update rule of local solutions. Now, we compute the average virtual model at iteration $i$ from \eqref{eqn51} as follows:
\begin{align}
    \bar{\boldsymbol{\theta}}^{i} = \bar{\boldsymbol{\theta}}^{i_{c}+1} - \frac{1}{N} \sum_{n \in \mathcal{N}} \sum_{z=i_{c}+1}^{i-1} \eta_{i_{c}} \Tilde{g}_{n}^{z}.
\end{align}
Firstly, without loss of generality, suppose $i = s_{t} T_{n} + r$, with $s_{t}$ and $r$ denoting the indices of global communication round and local updates, respectively. Next, we consider that for $i_{c}+1 < i \leq i_{c} + T_{n}$, $\mathbb{E}_{i}||\bar{\boldsymbol{\theta}}^{i} - \boldsymbol{\theta}_{n}^{i}||$ does not depend on time $i \leq i_{c}$ for $1 \leq n \leq N$. Therefore, for all iterations $1 \leq i \leq I$, where $I = KT_{n}$, we can write,
\begin{align}
    & \frac{1}{K} \sum_{k=1}^{K} \sum_{n=1}^{N} \frac{1}{T_{n}} \sum_{t=1}^{T_{n}} \mathbb{E}||\bar{\boldsymbol{\theta}}_{k}^{t} - \boldsymbol{\theta}_{n,k}^{t}||^{2} \nonumber \\
    &= \frac{1}{I} \sum_{i=1}^{I} \sum_{n=1}^{N} \mathbb{E}||\bar{\boldsymbol{\theta}}^{i} - \boldsymbol{\theta}_{n}^{i}||^{2} \nonumber \\
    &= \frac{1}{I} \sum_{s_{t}=1}^{\frac{I}{T_{n}} - 1} \sum_{n=1}^{N} \sum_{r=1}^{T_{n}} \mathbb{E}||\bar{\boldsymbol{\theta}}^{s_{t} E + r} - \boldsymbol{\theta}_{n}^{s_{t} E + r}||^{2}. 
\end{align}
We bound the term $\mathbb{E}||\bar{\boldsymbol{\theta}}^{i} - \boldsymbol{\theta}_{l}^{i}||^{2}$ for $i_{c}+1 \leq i = s_{t} T_{n} + r \leq i_{c} + T_{n} $ in threes steps: (1) We first relate this quantity to the variance between stochastic gradient and full gradient, (2) We use Assumption 1 on unbiased estimation and i.i.d. mini-batch sampling, (3) We use Assumption 3 to bound the final terms. In the following parts, we proceed to implement each of these steps. It is to note that $l$ is associated with individual device while $n$ is used for summing over devices.

\noindent
\textit{Relating to variance:}
\vspace{-8pt}
\begin{align*}
    &\mathbb{E}||\bar{\boldsymbol{\theta}}^{s_{t} E + r} - \boldsymbol{\theta}_{l}^{s_{t} E + r}||^{2} \nonumber \\
    &= \mathbb{E}|| \bar{\boldsymbol{\theta}}^{i_{c} + 1} - \bigg[ \sum_{z=i_{c}+1}^{i-1} \eta_{i_{c}} \Tilde{g}_{l}^{z} \bigg] - \bar{\boldsymbol{\theta}}^{i_{c}+1} \nonumber \\
    &\hspace{12em} + \bigg[ \frac{1}{N} \sum_{n \in \mathcal{N}} \sum_{z=i_{c}+1}^{i-1} \eta_{i_{c}} \Tilde{g}_{n}^{z} \bigg] ||^{2} \nonumber \\
    &\stackrel{\textcircled{\footnotesize 1}}{=} \mathbb{E}|| \sum_{z=1}^{r} \eta_{i_{c}} \Tilde{g}_{l}^{s_{t}+z} - \frac{1}{N} \sum_{n \in \mathcal{N}} \sum_{z=1}^{r} \eta_{i_{c}} \Tilde{g}_{n}^{s_{t}+z}||^{2} \nonumber \\
    &\stackrel{\textcircled{\footnotesize 2}}{\leq} 2 \bigg[ \mathbb{E}|| \sum_{z=1}^{r} \eta_{i_{c}} \Tilde{g}_{l}^{s_{t}+z}||^{2} - \mathbb{E} ||\frac{1}{N} \sum_{n \in \mathcal{N}} \sum_{z=1}^{r} \eta_{i_{c}} \Tilde{g}_{n}^{s_{t}+z}||^{2} \bigg] \nonumber \\
\end{align*}
\begin{align}
    &\stackrel{\textcircled{\footnotesize 3}}{=} 2 \bigg[ \mathbb{E}|| \sum_{z=1}^{r} \eta_{i_{c}} \Tilde{g}_{l}^{s_{t}+z} - \mathbb{E}\bigg[\sum_{z=1}^{r} \eta_{i_{c}} \Tilde{g}_{l}^{s_{t}+z}\bigg]||^{2} \nonumber \\
    &- \mathbb{E} ||\sum_{z=1}^{r} \eta_{i_{c}} \Tilde{g}_{l}^{s_{t}+z}||^{2} + \mathbb{E} ||\frac{1}{N} \sum_{n \in \mathcal{N}} \sum_{z=1}^{r} \eta_{i_{c}} \Tilde{g}_{n}^{s_{t}+z} - \mathbb{E}\bigg[\frac{1}{N} \nonumber \\
    &\times \sum_{n \in \mathcal{N}} \sum_{z=1}^{r} \eta_{i_{c}} \Tilde{g}_{n}^{s_{t}+z} \bigg] ||^{2} \bigg] + ||\mathbb{E}\bigg[\frac{1}{N} \sum_{n \in \mathcal{N}} \sum_{z=1}^{r} \eta_{i_{c}} \Tilde{g}_{n}^{s_{t}+z}\bigg]||^{2} \nonumber \\
    &\stackrel{\textcircled{\footnotesize 4}}{=} 2 \mathbb{E}\bigg( \bigg[ ||\sum_{z=1}^{r} \eta_{i_{c}} \bigg[\Tilde{g}_{l}^{s_{t}T_{n}+z} - g_{l}^{s_{t}T_{n}+z}\bigg]||^{2}  \nonumber \\
    &+ ||\sum_{z=1}^{r} \eta_{i_{c}} g_{l}^{s_{t}T_{n}+z}||^{2} \bigg] \nonumber \\
    &+ ||\frac{1}{N} \sum_{n \in \mathcal{N}} \sum_{z=1}^{r} \eta_{i_{c}} \bigg[\Tilde{g}_{n}^{s_{t}T_{n}+z} - g_{n}^{s_{t}T_{n}+z}\bigg]||^{2} \nonumber \\
    &+ ||\frac{1}{N} \sum_{n \in \mathcal{N}} \sum_{z=1}^{r} \eta_{i_{c}} g_{n}^{s_{t}T_{n}+z}||^{2}\bigg),
\end{align}
where \textcircled{\footnotesize 1} holds because $i = s_{t}T_{n}+r \leq i_{c} + T_{n}$, \textcircled{\footnotesize 2} is due to $||a-b||^{2} \leq 2(||a||^{2} + ||b||^{2})$, \textcircled{\footnotesize 3} arises because of $\mathbb{E}[\boldsymbol{\theta}^{2}] = \mathbb{E}[[\boldsymbol{\theta} - \mathbb{E}[\boldsymbol{\theta}]]^{2}] + \mathbb{E}[\boldsymbol{\theta}]^{2}$, \textcircled{\footnotesize 4} comes from Assumption 1.

\noindent
\textit{Unbiased estimation and i.i.d. sampling}
\begin{align*}
    &= 2 \mathbb{E}\bigg( \bigg[ \sum_{z=1}^{r} \eta_{i_{c}}^{2} ||\Tilde{g}_{l}^{s_{t}T_{n}+z} - g_{l}^{s_{t}T_{n}+z}||^{2} \nonumber \\
    &+ \sum_{p \neq q \vee l \neq v} \bigg\langle \eta_{i_{c}} \Tilde{g}_{l}^{p}-\eta_{i_{c}} g_{l}^{p}, \eta_{i_{c}} \Tilde{g}_{v}^{q}-\eta_{i_{c}} g_{v}^{q} \bigg\rangle \nonumber \\
    &+ ||\sum_{z=1}^{r}\eta_{i_{c}} g_{l}^{s_{t}T_{n}+z}||^{2}\bigg] \nonumber \\
    &+ \frac{1}{N^{2}} \sum_{l \in \mathcal{N}} \sum_{z=1}^{r} \eta_{i_{c}}^{2} ||\Tilde{g}_{l}^{s_{t}T_{n}+z} - g_{l}^{s_{t}T_{n}+z}||^{2} \nonumber \\
    &+ \frac{1}{N^{2}} \sum_{p \neq q \vee l \neq v} \bigg\langle \eta_{i_{c}} \Tilde{g}_{l}^{p}-\eta_{i_{c}} g_{l}^{p}, \eta_{i_{c}} \Tilde{g}_{v}^{q}-\eta_{i_{c}} g_{v}^{q} \bigg\rangle \nonumber \\
    &+ ||\frac{1}{N} \sum_{n \in \mathcal{N}} \sum_{z=1}^{r} \eta_{i_{c}} g_{n}^{s_{t}T_{n}+z} ||^{2} \bigg) \nonumber \\
    &\stackrel{\textcircled{\footnotesize 5}}{=} 2 \mathbb{E}\bigg( \bigg[ \sum_{z=1}^{r} \eta_{i_{c}}^{2} ||\Tilde{g}_{l}^{s_{t}T_{n}+z} - g_{l}^{s_{t}T_{n}+z}||^{2} \nonumber \\
    &+ ||\sum_{z=1}^{r}\eta_{i_{c}} g_{l}^{s_{t}T_{n}+z}||^{2}\bigg] + \frac{1}{N^{2}} \sum_{n \in \mathcal{N}} \sum_{z=1}^{r} \eta_{i_{c}}^{2} ||\Tilde{g}_{n}^{s_{t}T_{n}+z} \nonumber \\
    & - g_{n}^{s_{t}T_{n}+z}||^{2} + ||\frac{1}{N} \sum_{n \in \mathcal{N}} \sum_{z=1}^{r} \eta_{i_{c}} g_{n}^{s_{t}T_{n}+z} ||^{2} \bigg) \nonumber \\
    &\stackrel{\textcircled{\footnotesize 6}}{\leq} 2 \mathbb{E}\bigg( \bigg[ \sum_{z=1}^{r} \eta_{i_{c}}^{2} ||\Tilde{g}_{l}^{s_{t}T_{n}+z} - g_{l}^{s_{t}T_{n}+z}||^{2} + r \sum_{z=1}^{r}\eta_{i_{c}}^{2} \nonumber \\
    &\times ||g_{l}^{s_{t}T_{n}+z}||^{2}\bigg] + \frac{1}{N^{2}} \sum_{n \in \mathcal{N}} \sum_{z=1}^{r} ||\Tilde{g}_{n}^{s_{t}T_{n}+z} - g_{n}^{s_{t}T_{n}+z}||^{2} \nonumber \\ 
    &+ \frac{r}{N^{2}}  \sum_{n \in \mathcal{N}} \sum_{z=1}^{r} \eta_{s_{t}T_{n}+z}^{2} ||g_{n}^{s_{t}T_{n}+z} ||^{2} \nonumber \\
\end{align*}
\begin{align}    
    &= 2 \bigg( \bigg[ \sum_{z=1}^{r} \eta_{i_{c}}^{2} \mathbb{E} ||\Tilde{g}_{l}^{s_{t}T_{n}+z} - g_{l}^{s_{t}T_{n}+z}||^{2} + r \sum_{z=1}^{r}\eta_{i_{c}}^{2} \nonumber \\
    & \times \mathbb{E} ||g_{l}^{s_{t}T_{n}+z}||^{2}\bigg] + \frac{1}{N^{2}} \sum_{n \in \mathcal{N}} \sum_{z=1}^{r} \eta_{i_{c}}^{2} \mathbb{E} ||\Tilde{g}_{n}^{s_{t}T_{n}+z} \nonumber \\
    & - g_{n}^{s_{t}T_{n}+z}||^{2} + \frac{r}{N^{2}}  \sum_{n \in \mathcal{N}} \sum_{z=1}^{r} \eta_{i_{c}}^{2} \mathbb{E} ||g_{n}^{s_{t}T_{n}+z} ||^{2} \bigg), \label{eqn55}
\end{align}
where \textcircled{\footnotesize 5} is due to independent mini-batch sampling as well as unbiased estimation assumption, and \textcircled{\footnotesize 6} follows from the inequality $||\sum_{i=1}^{m} a_{i}||^{2} \leq m \sum_{i=1}^{m} ||a_{i}||^{2}$.

\noindent
\textit{Using Assumption 3:}
Our next step is to bound the terms in \eqref{eqn55} using Assumption 3 as follows:
\begin{align}
    & \mathbb{E} ||\bar{\boldsymbol{\theta}}_{k}^{t} - \boldsymbol{\theta}_{l,k}^{t}||^{2} \leq 2 \bigg( \bigg[ \sum_{z=1}^{r} \eta_{i_{c}}^{2} \bigg[ C_{1} ||g_{l}^{s_{t}T_{n}+z}||^{2} + \frac{\sigma^{2}}{B} \bigg] \nonumber \\
    &+ r \sum_{z=1}^{r} \eta_{i_{c}}^{2} || g_{l}^{s_{t}T_{n}+z}||^{2} + \frac{1}{N^{2}} \sum_{n \in \mathcal{N}} \sum_{z=1}^{r} \eta_{i_{c}}^{2} \bigg[ C_{1} ||g_{n}^{s_{t}T_{n}+z}||^{2} \nonumber \\
    & + \frac{\sigma^{2}}{B} \bigg] + \frac{r}{N^{2}} \sum_{n \in \mathcal{N}} \sum_{z=1}^{r} \eta_{i_{c}}^{2} || g_{n}^{s_{t}T_{n}+z}||^{2} \bigg) \nonumber \\
    &= 2 \bigg( \bigg[ \sum_{z=1}^{r} \eta_{i_{c}}^{2} C_{1} ||g_{l}^{s_{t}T_{n}+z}||^{2} + \sum_{z=1}^{r} \eta_{i_{c}}^{2} \frac{\sigma^{2}}{B} \nonumber \\
    &+ r \sum_{z=1}^{r} \eta_{i_{c}}^{2} || g_{l}^{s_{t}T_{n}+z}||^{2} \bigg] + \frac{1}{N^{2}} \sum_{n \in \mathcal{N}} \sum_{z=1}^{r} \eta_{i_{c}}^{2} C_{1} ||g_{n}^{s_{t}T_{n}+z}||^{2} \nonumber \\
    &+ \sum_{z=1}^{r} \eta_{i_{c}}^{2} \frac{\sigma^{2}}{NB} + \frac{r}{N^{2}} \sum_{n \in \mathcal{N}} \sum_{z=1}^{r} \eta_{i_{c}}^{2} || g_{n}^{s_{t}T_{n}+z}||^{2} \bigg). \label{eqn56} \nonumber \\
\end{align}
Now we determine the upper bound for $\sum_{r=1}^{T_{n}} \sum_{n=1}^{N} [\mathbb{E} ||\bar{\boldsymbol{\theta}}_{k}^{t} - \boldsymbol{\theta}_{n,k}^{t}||]$ using \eqref{eqn56} as follows:
\begin{align*}
    &\sum_{n=1}^{N} \sum_{r=1}^{T_{n}} \bigg[\mathbb{E} ||\bar{\boldsymbol{\theta}}^{s_{t}T_{n}+z} - \boldsymbol{\theta}_{n}^{s_{t}T_{n}+z}||\bigg] \nonumber \\
    &\leq 2 \sum_{l=1}^{N} \sum_{r=1}^{T_{n}} \bigg( \bigg[ \sum_{z=1}^{r} \eta_{i_{c}}^{2} C_{1} ||g_{l}^{s_{t}T_{n}+z}||^{2} + \sum_{z=1}^{r} \eta_{i_{c}}^{2} \frac{\sigma^{2}}{B} \nonumber \\
    &+ r \sum_{z=1}^{r} \eta_{i_{c}}^{2} || g_{l}^{s_{t}T_{n}+z}||^{2} \bigg] + \frac{1}{N^{2}} \sum_{n \in \mathcal{N}} \sum_{z=1}^{r} \eta_{i_{c}}^{2} C_{1} ||g_{n}^{s_{t}T_{n}+z}||^{2} \nonumber \\
    &+ \sum_{z=1}^{r} \eta_{i_{c}}^{2} \frac{\sigma^{2}}{NB} + \frac{r}{N^{2}} \sum_{n \in \mathcal{N}} \sum_{z=1}^{r} \eta_{i_{c}}^{2} || g_{n}^{s_{t}T_{n}+z}||^{2} \bigg) \nonumber \\
    &= 2 \eta_{i_{c}}^{2} \bigg( \bigg[ C_{1} \sum_{l=1}^{N} \sum_{r=1}^{T_{n}} \sum_{z=1}^{r} ||g_{l}^{s_{t}T_{n}+z}||^{2} + \sum_{r=1}^{T_{n}} \frac{r \sigma^{2}}{B} \nonumber \\
    &+ \sum_{l=1}^{N} \sum_{r=1}^{T_{n}} \sum_{z=1}^{r} r || g_{l}^{s_{t}T_{n}+z}||^{2} \bigg] + \sum_{r=1}^{T_{n}} \frac{1}{N^{2}} \sum_{n \in \mathcal{N}} \sum_{z=1}^{r}  C_{1} \nonumber \\
    &\times ||g_{n}^{s_{t}T_{n}+z}||^{2} + \sum_{r=1}^{T_{n}} \frac{r \sigma^{2}}{NB} + \sum_{r=1}^{T_{n}} \frac{r}{N^{2}} \sum_{n \in \mathcal{N}} \sum_{z=1}^{r} || g_{n}^{s_{t}T_{n}+z}||^{2} \bigg) \nonumber \\
\end{align*}
\begin{align}
    &\stackrel{\textcircled{\footnotesize 1}}{\leq} 2 \eta_{i_{c}}^{2} \bigg( \bigg[ \sum_{z=1}^{T_{n}}  C_{1} \sum_{l=1}^{N} ||g_{l}^{s_{t}T_{n}+z}||^{2} + \frac{T_{n}(T_{n}+1) \sigma^{2}}{2B} \nonumber \\
    &+ \frac{T_{n}(T_{n}+1)}{2} \sum_{z=1}^{T_{n}} \sum_{l=1}^{N} ||g_{l}^{s_{t}T_{n}+z}||^{2} \nonumber \\
    &+ \frac{1}{N^{2}} \sum_{n \in \mathcal{N}} \sum_{z=1}^{T_{n}}  C_{1} ||g_{n}^{s_{t}T_{n}+z}||^{2} \nonumber \\
    &+ \frac{T_{n}(T_{n}+1)\sigma^{2}}{2NB} +\frac{T_{n}(T_{n}+1)}{2N^{2}} \sum_{n \in \mathcal{N}} \sum_{z=1}^{T_{n}} || g_{n}^{s_{t}T_{n}+z}||^{2} \nonumber \\
    &= \frac{\eta_{i_{c}}^{2} (N+1)}{N} \bigg( \bigg[ (2 C_{1} + T_{n}(T_{n}+1)) \sum_{z=1}^{T_{n}} \sum_{n=1}^{N} || g_{n}^{s_{t}T_{n}+z}||^{2} \bigg] \nonumber \\
    & \hspace{12em} + \frac{T_{n}(T_{n}+1)\sigma^{2}}{B} \bigg), \label{eqn57}
\end{align}
where \textcircled{\footnotesize 1} follows from the fact that the terms $||g_{l}||^{2}$ are positive. Now, taking summation over global communication rounds in \eqref{eqn57} gives:
\begin{align}
     &\sum_{n=1}^{N} \sum_{s_{t}=1}^{I/T_{n}-1} \sum_{r=1}^{T_{n}} \bigg[\mathbb{E} ||\bar{\boldsymbol{\theta}}^{s_{t}T_{n}+z} - \boldsymbol{\theta}_{n}^{s_{t}T_{n}+z}||\bigg] \nonumber \\
     &\leq \frac{\eta_{i_{c}}^{2} (N+1)}{N} \bigg( \bigg[ (2 C_{1} \nonumber \\
     & \hspace{2em} + T_{n}(T_{n}+1)) \sum_{n=1}^{N} \sum_{s_{t}=1}^{I/T_{n}-1} \sum_{z=1}^{T_{n}} || g_{n}^{s_{t}T_{n}+z}||^{2} \bigg] \nonumber \\
     & \hspace{6em} + \frac{I(T_{n}+1)\sigma^{2}}{B} \bigg) \nonumber \\
     &= \frac{\eta_{i_{c}}^{2} (N+1)}{N} \bigg( \bigg[ (2 C_{1} + T_{n}(T_{n}+1)) \sum_{n=1}^{N} \sum_{i=1}^{I} || g_{n}^{i}||^{2} \bigg] \nonumber \\
     &\hspace{6em} + \frac{I(T_{n}+1)\sigma^{2}}{B} \bigg),
\end{align}
which leads to
\begin{align}
     &\sum_{n=1}^{N} \frac{1}{I} \sum_{i=1}^{I} \bigg[\mathbb{E} ||\bar{\boldsymbol{\theta}}^{i} - \boldsymbol{\theta}_{n}^{i}||\bigg] \nonumber \\
     &\leq \frac{(2 C_{1} + T_{n}(T_{n}+1))}{I} \frac{\eta_{i_{c}}^{2} (N+1)}{N} \sum_{n=1}^{N} \sum_{i=0}^{I-1} || g_{n}^{i}||^{2} \nonumber \\
     & \hspace{8em} + \frac{\eta_{i_{c}}^{2} I(N+1)(T_{n}+1)\sigma^{2}}{NB} \nonumber \\
     &\stackrel{\textcircled{\footnotesize 1}}{\leq} \frac{(2 C_{1} + T_{n}(T_{n}+1))}{I} \frac{\lambda \eta_{i_{c}}^{2} (N+1)}{N} \sum_{n=1}^{N} \sum_{i=0}^{I-1} || g_{n}^{i}||^{2} \nonumber \\
     & \hspace{8em} + \frac{\eta_{i_{c}}^{2} I(N+1)(T_{n}+1)\sigma^{2}}{NB}, \label{eqn59}
\end{align}
where \textcircled{\footnotesize 1} follows from the definition of weighted gradient diversity and upper bound assumption in (30) of the main paper. Finally, \eqref{eqn59} can be written as:
\begin{align}
     &\frac{1}{K} \sum_{k=1}^{K} \sum_{n=1}^{N} \frac{1}{T_{n}} \sum_{t=1}^{T_{n}} \bigg[\mathbb{E} ||\bar{\boldsymbol{\theta}}_{k}^{t} - \boldsymbol{\theta}_{n,k}^{t}||\bigg] \nonumber \\
     &\leq \frac{(2 C_{1} + T_{n}(T_{n}+1))}{KT_{n}} \frac{\lambda \eta_{i_{c}}^{2} (N+1)}{N} \sum_{k=1}^{K} \sum_{n=1}^{N} \sum_{t=1}^{T_{n}} || g_{n,k}^{t}||^{2} \nonumber \\
     & \hspace{6em} + \frac{\eta_{i_{c}}^{2} KT_{n}(N+1)(T_{n}+1)\sigma^{2}}{NB}. \label{eqn60}
\end{align}

\subsection{Proof of Lemma 3} \label{prooflemma3}
We have
\begin{align}
    &\mathbb{E}\bigg[||\Tilde{g}_{k}^{t} - g_{k}^{t}||^{2}\bigg] \stackrel{\textcircled{\footnotesize 1}}{=} \mathbb{E}\bigg[||\frac{1}{N} \sum_{n=0}^{N} \Tilde{g}_{n,k}^{t} - g_{n,k}^{t}||^{2}\bigg] \nonumber \\
    &= \frac{1}{N^{2}} \mathbb{E}\bigg[\sum_{n=0}^{N} ||(\Tilde{g}_{n,k}^{t} - g_{n,k}^{t})||^{2}\bigg] \nonumber \\
    & \hspace{6em} + \sum_{i \neq n} \langle \Tilde{g}_{i,k}^{t} - g_{i,k}^{t}, \Tilde{g}_{n,k}^{t} - g_{n,k}^{t} \rangle \nonumber \\
    &= \frac{1}{N^{2}} \sum_{n=0}^{N} \mathbb{E} ||(\Tilde{g}_{n,k}^{t} - g_{n,k}^{t})||^{2} \nonumber \\
    & \hspace{6em} + \sum_{i \neq n} \frac{1}{N^{2}} \mathbb{E} \bigg[\langle \Tilde{g}_{i,k}^{t} - g_{i,k}^{t}, \Tilde{g}_{n,k}^{t} - g_{n,k}^{t}\rangle \bigg] \nonumber \\
    &\stackrel{\textcircled{\footnotesize 2}}{=} \frac{1}{N^{2}} \sum_{n=0}^{N} \mathbb{E} ||(\Tilde{g}_{n,k}^{t} - g_{n,k}^{t})||^{2} \nonumber \\
    & \hspace{4em} + \frac{1}{N^{2}} \sum_{i \neq n} \langle  \mathbb{E} \bigg[\Tilde{g}_{i,k}^{t} - g_{i,k}^{t}\bigg], \mathbb{E} \bigg[\Tilde{g}_{n,k}^{t} - g_{n,k}^{t} \bigg] \rangle  \nonumber \\
    &\stackrel{\textcircled{\footnotesize 3}}{\leq} \frac{1}{N^{2}} \sum_{n=0}^{N} \bigg[ C_{1} ||g_{n,k}^{t}||^{2} + C_{2}^{2}\bigg] = \frac{C_{1}}{N^{2}} \sum_{n=0}^{N} ||g_{n,k}^{t}||^{2} + \frac{C_{2}^{2}}{N},
\end{align}
where we use the definition of $\Tilde{g}_{k}^{t}$ and $g_{k}^{t}$ in \textcircled{\footnotesize 1}, in \textcircled{\footnotesize 2} we use the fact that mini-batches are chosen in i.i.d. manner at each device, and \textcircled{\footnotesize 3} follows directly from Assumption 3. We note that Assumption 3 implies $\mathbb{E}[\Tilde{g}_{n,k}^{t}] = g_{n,k}^{t}$. Therefore. we have
\begin{align}
    \mathbb{E}\bigg[ ||\Tilde{g}_{k}^{t}||^{2} \bigg] &= \mathbb{E}\bigg[ ||\Tilde{g}_{k}^{t} - \mathbb{E} [ \Tilde{g}_{k}^{t} ]||^{2} \bigg] + ||\mathbb{E} [ \Tilde{g}_{k}^{t} ]||^{2} \nonumber \\
    &= \mathbb{E}\bigg[ ||\Tilde{g}_{k}^{t} - g_{k}^{t} ||^{2} \bigg] + ||g_{k}^{t}||^{2} \nonumber \\
    &\stackrel{\textcircled{\footnotesize 1}}{\leq} \frac{C_{1}}{N^{2}} \sum_{n=0}^{N} ||g_{n,k}^{t}||^{2} + \frac{C_{2}^{2}}{N} + ||\frac{1}{N} \sum_{n=0}^{N} g_{n,k}^{t}||^{2} \nonumber \\
    &\stackrel{\textcircled{\footnotesize 2}}{\leq} \frac{C_{1}}{N^{2}} \sum_{n=0}^{N} ||g_{n,k}^{t}||^{2} + \frac{C_{2}^{2}}{N} + \frac{1}{N} \sum_{n=0}^{N} ||g_{n,k}^{t}||^{2} \nonumber \\
    &= \bigg(\frac{C_{1}+N}{N^{2}}\bigg) \sum_{n=0}^{N} ||g_{n,k}^{t}||^{2} + \frac{C_{2}^{2}}{N},
\end{align}
where \textcircled{\footnotesize 1} and \textcircled{\footnotesize 2} follows from the fact that $||\sum_{i=1}^{m} a_{i}||^{2} \leq m \sum_{i=1}^{m} ||a_{i}||^{2}$, with $a_{i} \in \mathbb{R}^{n}$. Using the upper bound over the weighted gradient diversity, $\lambda$,
\begin{align}
    \mathbb{E}\bigg[ ||\Tilde{g}_{k}^{t}||^{2}\bigg] \leq \lambda \bigg(\frac{C_{1}+N}{N^{2}}\bigg) ||\sum_{n=0}^{N} g_{n,k}^{t}||^{2} + \frac{C_{2}^{2}}{N},
\end{align}
results in the stated bound.

\subsection{Proof of Lemma 4} \label{prooflemma4}
To prove Lemma 4, we fix the indices related to global and local iteration $k$ and $t$, consequently dropping them from notations temporarily. Let $X_{n,d,\pm} = {\langle \hat{Z} \rangle}_{\lvert \Psi_{n}(\boldsymbol{\theta}_{n} \pm \frac{\pi}{2} e_{d}) \rangle} - {\langle Z \rangle}_{\lvert \Psi_{n}(\boldsymbol{\theta}_{n} \pm \frac{\pi}{2} e_{d}) \rangle}$ denote the difference between the estimated and true expectation of the observable $Z$ under the quantum state ${\lvert \Psi_{n}(\boldsymbol{\theta}_{n} \pm \frac{\pi}{2} e_{d}) \rangle}$ whose $d^{\text{th}}$ parameter is phase shifted by $\pm \frac{\pi}{2}$. In the following analysis, we use the notation $\lvert \Psi_{n,d,\pm} \rangle = {\lvert \Psi_{n}(\boldsymbol{\theta}_{n} \pm \frac{\pi}{2} e_{d}) \rangle}$ for brevity. The variance of the gradient estimate in \eqref{eqn44} is written as
\begin{align}
    \text{var}(\xi_{n}) &= \mathbb{E} \bigg[ \sum_{d=1}^{D} \bigg( \frac{1}{2} ( {\langle \hat{Z} \rangle}_{\lvert \Psi_{n,d,+} \rangle} - {\langle \hat{Z} \rangle}_{\lvert \Psi_{n,d,-} \rangle}) \nonumber \\
    & \hspace{1em} - \frac{1}{2} ( {\langle Z \rangle}_{\lvert \Psi_{n,d,+} \rangle} - {\langle Z \rangle}_{\lvert \Psi_{n,d,-} \rangle}) \bigg)^{2} \bigg] \nonumber \\
    &= \sum_{d=1}^{D} \frac{1}{4} \mathbb{E} \bigg[ \bigg( X_{n,d,+} - X_{n,d,-} \bigg)^{2} \bigg] \nonumber \\
    &= \sum_{d=1}^{D} \frac{1}{4}  \bigg( \mathbb{E}[X_{n,d,+}^{2}] - \mathbb{E}[X_{n,d,-}^{2}] \bigg), \label{eqn61}
\end{align}
where the expectation is taken with respect to the $H$ measurements of the quantum states ${\lvert \Psi_{n}(\boldsymbol{\theta}_{n} + \frac{\pi}{2} e_{d}) \rangle}$ and ${\lvert \Psi_{n}(\boldsymbol{\theta}_{n} - \frac{\pi}{2} e_{d}) \rangle}$ for $d = 1, 2, \dots, D$. Hence, the random variables $X_{n,d,+}$ and $X_{n,d,-}$ are independent for $d = 1, 2, \dots, D$, which results in the equality in \eqref{eqn61}. It is to note that the expectation $\mathbb{E}[X_{n,d,+}^{2}]$ is equal to the variance $\text{var}({\langle \hat{Z} \rangle}_{\lvert \Psi_{n,d,+} \rangle})$ of the random variable ${\langle \hat{Z} \rangle}_{\lvert \Psi_{n,d,+} \rangle}$. Let $Y$ be the random variable that defines the index of the measurement of the observable $Z$. Therefore, $Z = h_{Y}$ represents the corresponding measurement output. We denote the Bernoulli random variable as $W_{y} = \mathbb{I}\{Y = y\}$ determining whether $Y = y (W_{y} = 1)$ or not $(W_{y} = 0)$. We also mention that the quantum measurements are i.i.d., and thus it follows from the definition of expectation of ${\langle \hat{Z} \rangle}_{\lvert \Psi_{n,d,+} \rangle}$ that
\begin{align}
    \mathbb{E}[X_{n,d,+}^{2}] &= \frac{1}{H} \text{var}\bigg( \sum_{y=1}^{N_{z}} h_{y} W_{y} \bigg) \nonumber \\
    &= \frac{1}{H} \mathbb{E} \bigg[ \bigg( \sum_{y=1}^{N_{z}} h_{y} (W_{y} - p(y|\boldsymbol{\theta}_{n} + e_{d} \frac{\pi}{2})) \bigg)^{2} \bigg] \nonumber \\
    &\stackrel{\textcircled{\footnotesize 1}}{\leq} \frac{1}{H} \bigg(\sum_{y=1}^{N_{z}} h_{y}^{2}\bigg) \sum_{y=1}^{N_{z}} \text{var}(W_{y}) \nonumber \\
    &\stackrel{\textcircled{\footnotesize 2}}{=} \frac{1}{H} \bigg(\sum_{y=1}^{N_{z}} h_{y}^{2}\bigg) \sum_{y=1}^{N_{z}} v\bigg(p(y|\boldsymbol{\theta}_{n} + e_{d} \frac{\pi}{2})\bigg) \nonumber \\
    &\leq \frac{N_{z}}{N_{y}} \bigg(\sum_{y=1}^{N_{z}} h_{y}^{2}\bigg) v = \frac{N_{z} \operatorname{Tr}(Z^{2})}{H} v, \label{eqn62}
\end{align}
where \textcircled{\footnotesize 1} follows from the Cauchy-Schwarz inequality, \textcircled{\footnotesize 2} is due to the fact that the variance of the Bernoulli random variable $W_{y}$ is computed as
\begin{align}
    \text{var}\bigg(W_{y}\bigg) = \mathbb{E}\bigg[W_{y}^{2}\bigg] - \bigg(\mathbb{E}\bigg[W_{y}\bigg]\bigg)^{2} = v\bigg(p(y|\boldsymbol{\theta}_{n} + e_{d} \frac{\pi}{2})\bigg),
\end{align}
where $v(x) = x (1-x)$ for $x \in (0,1)$. The last yields from the definition of the quantity $v$. In a similar way, it can be shown that the following inequality holds
\begin{align}
    \mathbb{E}[X_{n,d,-}^{2}] \leq \frac{N_{z} \operatorname{Tr}(Z^{2})}{H} v. \label{eqn64}
\end{align}
From \eqref{eqn62} and \eqref{eqn64}, we can write while bringing the omitted indices back
\begin{align}
    \text{var}(\xi_{n,k}^{t}) \leq \frac{\nu N_{z} D Tr(Z^{2})}{2H}.
\end{align}
For $N$ number of QFL clients, we get
\begin{align}
    \text{var}(\xi_{k}^{t}) \leq \frac{1}{N} \sum_{n \in \mathcal{N}} \frac{\nu N_{z} D Tr(Z^{2})}{2H},
\end{align}
concluding the proof.

\subsection{Proof of Theorem 1} \label{prooftheorem1}
Using Lemma 1 and Lemma 2, we continue to further upper bound (34) of main paper as follows:
\begin{align}
    &\frac{1}{K} \frac{1}{\frac{1}{N}\sum_{n=1}^{N}T_{n}} \sum_{k=1}^{K} \sum_{t=1}^{T_{n}} \mathbb{E}[f(\bar{\boldsymbol{\theta}}_{k}^{t+1}) - f(\bar{\boldsymbol{\theta}}_{k}^{t})] \nonumber \\
    &\leq \frac{1}{K} \frac{1}{\frac{1}{N}\sum_{n=1}^{N}T_{n}} \sum_{k=1}^{K} \sum_{t=1}^{T_{n}} \Bigg(-\eta_{k} \mathbb{E}\bigg[\langle \nabla f(\bar{\boldsymbol{\theta}}_{k}^{t}), \Tilde{g}_{k}^{t} \rangle\bigg]\Bigg) \nonumber \\
    & \hspace{6em} + \frac{1}{K} \frac{1}{\frac{1}{N}\sum_{n=1}^{N}T_{n}} \sum_{k=1}^{K} \sum_{t=1}^{T_{n}} \frac{\eta_{k}^{2}  L}{2} \mathbb{E}\bigg[||\Tilde{g}_{k}^{t}||^{2}\bigg] \nonumber \\
    &\leq \frac{1}{K} \frac{1}{\frac{1}{N}\sum_{n=1}^{N}T_{n}} \sum_{k=1}^{K} \sum_{t=1}^{T_{n}} \bigg(-\frac{\eta_{k}}{2}||\nabla f(\bar{\boldsymbol{\theta}}_{k}^{t})||^2 \nonumber \\
    &- \frac{\eta_{k}}{2}||\sum_{n=1}^{N}\nabla f_{n}(\boldsymbol{\theta}_{n,k}^{t})||^{2}\bigg) + \frac{\lambda \eta_{k} L^{2}}{2K\frac{1}{N}\sum_{n=1}^{N}T_{n}} \frac{N+1}{N} \nonumber \\
    &\times \Bigg(\bigg[2C_{1}+T_{n}(T_{n}+1)\bigg] \eta_{k}^{2} \frac{1}{K} \frac{1}{\frac{1}{N}\sum_{n=1}^{N}T_{n}} \sum_{k=1}^{K} \sum_{t=1}^{T_{n}}||^2 \nonumber \\
    &- \frac{\eta_{k}}{2}||\sum_{n=1}^{N}\nabla f_{n}(\boldsymbol{\theta}_{n,k}^{t})||^{2}\Bigg) \nonumber \\
    &+ \frac{\eta_{k} L^{2}}{2K\frac{1}{N}\sum_{n=1}^{N}T_{n}} \bigg(\frac{N+1}{N}\bigg)\bigg(\frac{KT_{n}(T_{n}+1)\eta_{k}^{2}\sigma^{2}}{B}\bigg) \nonumber \\
    &+ \frac{1}{K} \frac{1}{\frac{1}{N}\sum_{n=1}^{N}T_{n}} \sum_{k=1}^{K} \sum_{t=1}^{T_{n}} \frac{L \eta_{k}^{2}}{2} \Bigg(\lambda \bigg(\frac{C_{1}}{N}+1\bigg) \nonumber \\
    &\times \bigg[||\sum_{n=1}^{N}\nabla f_{n}(\boldsymbol{\theta}_{n,k}^{t})||^{2}\bigg] + \frac{\sigma^{2}}{NB}\Bigg) \nonumber \\
    &= \frac{1}{K} \frac{1}{\frac{1}{N}\sum_{n=1}^{N}T_{n}} \sum_{k=1}^{K} \sum_{t=1}^{T_{n}} \bigg(-\frac{\eta_{k}}{2}||\nabla f(\bar{\boldsymbol{\theta}}_{k}^{t})||^2 \nonumber \\
    &- \frac{\eta_{k}}{2}||\sum_{n=1}^{N}\nabla f_{n}(\boldsymbol{\theta}_{n,k}^{t})||^{2}\bigg) + \frac{\lambda \eta_{k} L^{2}}{2K\frac{1}{N}\sum_{n=1}^{N}T_{n}} \frac{N+1}{N} \nonumber \\
    &\times \Bigg(\lambda \bigg[2C_{1}+T_{n}(T_{n}+1)\bigg] \eta_{k}^{2} \frac{1}{K} \frac{1}{\frac{1}{N}\sum_{n=1}^{N}T_{n}} \sum_{k=1}^{K} \sum_{t=1}^{T_{n}}||^2 \nonumber \\
    &- \frac{\eta_{k}}{2}||\sum_{n=1}^{N}\nabla f_{n}(\boldsymbol{\theta}_{n,k}^{t})||^{2} \Bigg) + \frac{KT_{n}(L+1)\eta_{k}^{2} \sigma^{2}}{B} \nonumber \\
    &+ \frac{1}{K} \frac{1}{\frac{1}{N}\sum_{n=1}^{N}T_{n}} \sum_{k=1}^{K} \sum_{t=1}^{T_{n}} \frac{\lambda L \eta_{k}^{2}}{2} \lambda \bigg(\frac{C_{1}}{N}+1\bigg) \nonumber \\
    &\bigg[||\sum_{n=1}^{N}\nabla f_{n}(\boldsymbol{\theta}_{n,k}^{t})||^{2}\bigg] + \frac{L \eta_{k}^{2}}{2} \frac{\sigma^{2}}{NB}. \label{39} 
\end{align}

From \eqref{39}, we have
\begin{align}
    \frac{1}{K} &\frac{1}{\frac{1}{N}\sum_{n=1}^{N}T_{n}} \sum_{k=1}^{K} \sum_{t=1}^{T_{n}} \mathbb{E}[f(\bar{\boldsymbol{\theta}}_{k}^{t+1}) - f(\bar{\boldsymbol{\theta}}_{k}^{t})] \nonumber \\
    &\leq - \frac{1}{K} \frac{1}{\frac{1}{N}\sum_{n=1}^{N}T_{n}} \sum_{k=1}^{K} \sum_{t=1}^{T_{n}} \frac{\eta_{k}}{2}||\nabla f(\bar{\boldsymbol{\theta}}_{k}^{t})||^2 \nonumber \\
    &+ \frac{1}{K} \frac{1}{\frac{1}{N}\sum_{n=1}^{N}T_{n}} \sum_{k=1}^{K} \sum_{t=1}^{T_{n}} \bigg[-\frac{\eta_{k}}{2} \nonumber \\
    &+ \frac{\lambda (N+1) L^{2} \eta_{k}^{3} [2C_{1}+T_{n}(T_{n}+1)]}{2N} \nonumber \\
    &+ \frac{\lambda L \eta_{k}^{2}}{2} \bigg(\frac{C_{1}}{N}+1\bigg)\bigg] \bigg[||\sum_{n=1}^{N}\nabla f_{n}(\boldsymbol{\theta}_{n,k}^{t})||^{2}\bigg] \nonumber \\
    &+ \frac{\eta_{k}^{3} L^{2} (T_{n}+1) \sigma^{2}}{B} \bigg(\frac{N+1}{N}\bigg) + \frac{L \eta_{k}^{2}}{2} \frac{\sigma^{2}}{NB} \nonumber \\
    &\stackrel{\textcircled{\footnotesize 1}}{\leq} - \frac{1}{K} \frac{1}{\frac{1}{N}\sum_{n=1}^{N}T_{n}} \sum_{k=1}^{K} \sum_{t=1}^{T_{n}} \frac{\eta_{k}}{2}||\nabla f(\bar{\boldsymbol{\theta}}_{k}^{t})||^2 \nonumber \\
    &+ \frac{\eta_{k}^{3} L^{2} (T_{n}+1) \sigma^{2}}{B} \bigg(\frac{N+1}{N}\bigg) + \frac{L \eta_{k}^{2}}{2} \frac{\sigma^{2}}{NB}, \label{eqn40}
\end{align}
where \textcircled{\footnotesize 1} follows if the following condition holds:
\begin{align}
    &-\frac{\eta_{k}}{2} + \frac{\lambda (N+1) L^{2} \eta_{k}^{3} [2C_{1}+T_{n}(T_{n}+1)]}{2N} \nonumber \\
    & \hspace{6em} + \frac{\lambda L \eta_{k}^{2}}{2} \bigg(\frac{C_{1}}{N}+1\bigg) \leq 0.
\end{align}
In any kind of FL framework, setting the coefficient of the local gradients' sum to zero helps control variance from diverse client updates, ensuring stable convergence. This condition limits the influence of individual clients on the global model, preventing oscillations or divergence. It keeps updates bounded, promoting reliable convergence toward an optimal solution.
By rearranging \eqref{eqn40}, we get
\begin{align}
    &\frac{1}{K} \frac{1}{\frac{1}{N}\sum_{n=1}^{N}T_{n}} \sum_{k=1}^{K} \sum_{t=1}^{T_{n}} \mathbb{E}||\nabla f(\bar{\boldsymbol{\theta}}_{k}^{t})||^{2} \leq \frac{2 [f(\bar{\boldsymbol{\theta}}_{1}^{0}) - f^{*}]}{\eta_{k} K\frac{1}{N}\sum_{n=1}^{N}T_{n}} \nonumber \\
    &+ \frac{L \eta \sigma^{2}}{NB} + \frac{2 \eta_{k}^{2} \sigma^{2} L^{2} (\frac{1}{N}\sum_{n=1}^{N}T_{n} + 1)}{B} \bigg(1+\frac{1}{N}\bigg). \label{eqn42new1}
\end{align}
Upto this point, we did not consider noise term in the local gradient of the quantum client. However, we have to consider that because the noise term in the local gradient of each quantum device will affect the convergence of the overall global gradient. Since the global gradient in QFL is an aggregation of the local gradients from all the devices, any noise or error in the local gradient estimates will also accumulate at the global level. Hence, we find the upper bound of the variance of the error introduced in the local gradient of each client due to quantum shot noise and add it to the upper bound of the global gradient in \eqref{eqn42new1}, resulting in (39) of main paper.

\noindent
In QFL, the gradient is estimated rather than explicitly computed. This approach leverages quantum computations to approximate the gradient, allowing for efficient optimization processes without relying on exact gradient calculations. Our assumption $\mathbb{E}[\Tilde{g_{n}}] = g_{n}$ means that the estimate is unbiased. Therefore, we can write
\begin{align}
    \Tilde{g}_{n,k}^{t} = g_{n,k}^{t} + \xi_{n,k}^{t}, \label{eqn81'}
\end{align}
where $\Tilde{g}_{n,k}^{t}$ is the stochastic estimate of the gradient, $g_{n,k}^{t}$ is the true gradient, and $\xi_{n,k}^{t}$ is the error or noise introduced in the estimation process, with the noise term satifying the conditions $\mathbb{E}[\xi_{n,k}^{t}] = 0$ and $\text{var}(\xi_{n,k}^{t}) = \mathbb{E}[||\Tilde{g}_{n,k}^{t} - g_{n,k}^{t}||]$ Taking average across all the devices, we get
\begin{align}
    \Tilde{g}_{k}^{t} = g_{k}^{t} + \xi_{k}^{t}, \label{eqn44}
\end{align}
where $\xi_{k}^{t} = \sum_{n=0}^{N} \xi_{n,k}^{t}$. Since the global gradient in QFL is an aggregation of the local gradients from all the devices, any noise or error in the local gradient estimates will also accumulate at the global level. If the errors are significant, they may cause the aggregated global gradient to deviate from the true direction of descent, slowing down convergence or leading to suboptimal solutions. Hence, we use Lemma 4 to find the upper bound of the variance of the gradient estimate. Therefore, The expected value of the squared norm of the global gradient in \eqref{eqn42new1} will be additionally bounded by the left hand side of Lemma 4 in the following way:
\begin{align}
        &\frac{1}{K} \sum_{k=1}^{K} \frac{1}{\frac{1}{N}\sum_{n=1}^{N}T_{n}} \sum_{t=1}^{T_{n}} \mathbb{E}||\nabla f(\bar{\boldsymbol{\theta}}_{k}^{t})||^{2} \leq \frac{2 [f(\bar{\boldsymbol{\theta}}_{1}^{0}) - f^{*}]}{\eta_{k} K\frac{1}{N}\sum_{n=1}^{N}T_{n}} \nonumber \\
        &+ \frac{L \eta \sigma^{2}}{NB} + \frac{2 \eta_{k}^{2} \sigma^{2} L^{2} (\frac{1}{N}\sum_{n=1}^{N}T_{n}+1)}{B} \bigg(\frac{N+1}{N}\bigg) \nonumber \\
        &+ \frac{1}{N} \sum_{n \in \mathcal{N}} \frac{\nu N_{z} D Tr(Z^{2})}{2H}.
\end{align}
In non-convex optimization, achieving a global minimum is often infeasible due to the landscape's complexity, filled with local minima and saddle points. Instead of focusing on bounding the distance between consecutive points, an alternative approach is to bound the squared norm of the gradient estimate. This approach helps gauge how close we are to a stationary point, where the gradient’s magnitude is minimal, indicating minimal change. By upper bounding the squared gradient, we can evaluate convergence towards a solution that may not be globally optimal, however is practically effective in reducing the loss.
\end{document}